\documentclass{llncs}

\usepackage{makeidx}

\usepackage[T1]{fontenc}

\usepackage{bm}
\usepackage{pgfplots}
\usepackage{enumerate,paralist}
\usepackage{pbox}
\usepackage{comment}
\usepackage{parskip}
\usepackage{multirow}
\usepackage{url}
\usepackage{graphicx}

\usepackage{tikz,pgffor}
\usetikzlibrary{arrows}
\usetikzlibrary{shapes}
\usetikzlibrary{calc}
\usetikzlibrary{automata}
\usetikzlibrary{positioning}

\tikzstyle{label}=[shape=circle,draw,inner sep=0pt,minimum size=5mm]
\tikzstyle{tran}=[draw,->,>=stealth, rounded corners]

\usepackage{lmodern}
\usepackage{amsmath}
\usepackage{amssymb}

\usepackage{mathtools}
\usepackage{times}
\usepackage{stmaryrd}

\usepackage{listings}
\lstdefinelanguage{prog}
{
morekeywords={if, then, else, fi, while, do, od, true, false, and, or, skip, newar, return},
sensitive = false
}

\newcommand{\POLYS}{{\mathfrak{R}}{\left[x_1,\dots, x_m\right]}}

\newcommand{\loc}{\ell}
\newcommand{\fn}[1]{\mathsf{#1}}
\newcommand{\locs}[1]{\mathit{L}^{\mathsf{#1}}}
\newcommand{\pvars}[1]{V^{\mathsf{#1}}}

\newcommand{\fnames}{\mathit{F}}
\newcommand{\clocs}[1]{\mathit{L}_{\mathrm{b}}^{\mathsf{#1}}}
\newcommand{\alocs}[1]{\mathit{L}_{\mathrm{a}}^{\mathsf{#1}}}
\newcommand{\flocs}[1]{\mathit{L}_{\mathrm{c}}^{\mathsf{#1}}}
\newcommand{\slocs}[1]{\mathit{L}_{\mathrm{s}}^{\mathsf{#1}}}
\newcommand{\dlocs}[1]{\mathit{L}_{\mathrm{d}}^{\mathsf{#1}}}
\newcommand{\transitions}[1]{{\rightarrow}_{#1}}
\newcommand{\assgn}[2]{\left[#1/#2\right]}

\newcommand{\Val}[1]{\mbox{\sl Val}_{\mathsf{#1}}}
\newcommand{\Aval}[1]{\mbox{\sl Val}_{\mathsf{#1}}}

\newcommand{\Rset}{\mathbb{R}}
\newcommand{\Nset}{\mathbb{N}}
\newcommand{\Zset}{\mathbb{Z}}

\newcommand{\lin}[1]{\loc_\mathrm{in}^{#1}}
\newcommand{\lout}[1]{\loc_\mathrm{out}^{#1}}

\newcommand{\inv}{I}

\newcommand{\recterm}{{\sc RecTermBou}}

\begin{document}

\title{Non-polynomial Worst-Case Analysis of \\Recursive Programs}


\author{
Krishnendu Chatterjee\inst{1} \and Hongfei Fu\inst{2} \and Amir Kafshdar Goharshady\inst{1}
}

\authorrunning{Chatterjee et al.}

\institute{
IST Austria, Klosterneuburg, Austria
\and
State Key Laboratory of Computer Science, Institute of Software\\ Chinese
Academy of Sciences, Beijing, P.R. China
}

\maketitle

\begin{abstract}
We study the problem of developing efficient approaches for proving worst-case bounds
of non-deterministic recursive programs.
Ranking functions are sound and complete for proving termination and
worst-case bounds of non-recursive programs.
First, we apply ranking functions to recursion, resulting in measure functions. We show that
measure functions provide a sound and complete approach to prove 
worst-case bounds of non-deterministic recursive programs.
Our second contribution is the synthesis of measure functions in non-polynomial forms.
We show that non-polynomial measure functions with 
logarithm and exponentiation can be synthesized
through abstraction of logarithmic or exponentiation terms, Farkas' Lemma,
and Handelman's Theorem using linear programming.
While previous methods obtain worst-case polynomial bounds, 
our 
approach can synthesize bounds of the
form $\mathcal{O}(n \log n)$ as well as $\mathcal{O}(n^r)$ where $r$ is not an integer.
We present experimental results to demonstrate that
our approach can obtain efficiently worst-case
bounds of classical recursive algorithms such as
(i)~Merge-Sort, the divide-and-conquer algorithm for the Closest-Pair problem,
where we obtain $\mathcal{O}(n \log n)$ worst-case bound, and
(ii)~Karatsuba's algorithm for polynomial multiplication and Strassen's
algorithm for matrix multiplication, where we obtain $\mathcal{O}(n^r)$ bound
such that $r$ is not an integer and close to the best-known bounds for the
respective algorithms.
\end{abstract}

\vspace{-2em}
\section{Introduction}
\vspace{-1em}

Automated analysis to obtain quantitative performance characteristics of
programs is a key feature of static analysis.
Obtaining precise worst-case complexity bounds is a topic of both wide
theoretical and practical interest.
The manual proof of such bounds can be cumbersome as well as require
mathematical ingenuity, e.g., the book {\em The Art of Computer Programming}
by Knuth presents several mathematically involved methods to obtain such
precise bounds~\cite{KnuthAllBooks}.
The derivation of such worst-case bounds requires a lot of mathematical skills
and is not an automated method.
However, the problem of deriving precise worst-case bounds is of huge interest
in program analysis: (a)~first, in applications such as hard real-time
systems, guarantees of worst-case behavior are required; and (b)~the bounds
are useful in early detection of egregious performance problems in large
code bases.
Works such as~\cite{SPEED1,SPEED2,Hoffman1,Hoffman2} provide an excellent
motivation for the study of automatic methods to obtain worst-case bounds
for programs.

Given the importance of the problem of deriving worst-case bounds, the
problem has been studied in various different ways.
\begin{compactenum}
\item {\em WCET Analysis.} The problem of worst-case execution time (WCET)
analysis is a large field of its own, that focuses on (but is not limited to)
sequential loop-free code with low-level hardware aspects~\cite{DBLP:journals/tecs/WilhelmEEHTWBFHMMPPSS08}.

\item {\em Resource Analysis.} The use of abstract interpretation and type
systems to deal with loop, recursion, data-structures has also been considered
~\cite{SPEED2,DBLP:journals/entcs/AlbertAGGPRRZ09,DBLP:conf/popl/JostHLH10}, e.g.,
using linear invariant generation to obtain disjunctive and non-linear bounds~\cite{DBLP:conf/cav/ColonSS03},
potential-based methods for handling recursion and inductive data structures~\cite{Hoffman1,Hoffman2}.

\item {\em Ranking Functions.} The notion of ranking functions is a powerful
technique for termination analysis of (recursive) programs~\cite{BG05,DBLP:conf/cav/BradleyMS05,DBLP:conf/tacas/ColonS01,DBLP:conf/vmcai/PodelskiR04,DBLP:conf/pods/SohnG91,DBLP:conf/vmcai/Cousot05,DBLP:journals/fcsc/YangZZX10,DBLP:journals/jossac/ShenWYZ13}.
They serve as a sound and complete approach for proving termination
of non-recursive programs~\cite{rwfloyd1967programs}, and they have also been
extended as ranking supermatingales for analysis of probabilistic programs~\cite{SriramCAV,HolgerPOPL,DBLP:conf/popl/ChatterjeeFNH16,DBLP:conf/cav/ChatterjeeFG16}.

\end{compactenum}

Given the many
results above, two aspects of the problem have not been addressed.
\begin{compactenum}
\item {\em WCET Analysis of Recursive Programs through Ranking Functions.}
The use of ranking functions has been limited mostly to non-recursive programs,
and 
their use to
obtain worst-case bounds for recursive programs has not been explored in depth.

\item {\em Efficient Methods for Precise Bounds.}
While previous works present methods for disjunctive polynomial
bounds~\cite{SPEED2} (such as $\max(0,n) \cdot (1+max(n,m))$), or multivariate
polynomial analysis~\cite{Hoffman1}, these works do not provide efficient
methods to synthesize bounds such as $\mathcal{O}(n \log n)$
or $\mathcal{O}(n^r)$, where $r$ is not an integer.

\end{compactenum}
We address these two aspects, i.e., efficient
methods for obtaining non-polynomial bounds such as $\mathcal{O}(n \log n)$, $\mathcal{O}(n^r)$ for recursive programs,
where $r$ is not an integer.

\smallskip\noindent{\em Our Contributions.}
Our main contributions are as follows:
\begin{compactenum}

\item First, we apply ranking functions to recursion,
resulting in {\em measure} functions, and show that
they provide a sound and complete method to prove termination
and worst-case bounds of non-deterministic recursive programs.

\item
Second, we present a sound approach for handling measure functions
of specific forms.
More precisely, we show that {\em non-polynomial} measure functions involving
logarithm and exponentiation can be synthesized using {\em linear
programming} through abstraction of logarithmic or exponentiation terms,
Farkas' Lemma, and Handelman's Theorem.

\item
A key application of our method is the worst-case analysis of recursive programs.
Our procedure can synthesize non-polynomial bounds of the form $\mathcal{O}(n \log n)$,
as well as $\mathcal{O}(n^r)$, where $r$ is not an integer.
We show the applicability of our technique to obtain worst-case complexity
bounds for several classical recursive programs:
\begin{compactitem}
\item
For {\em Merge-Sort}~\cite[Chapter 2]{DBLP:books/daglib/0023376} and the divide-and-conquer algorithm for the
{\em Closest-Pair problem}~\cite[Chapter~33]{DBLP:books/daglib/0023376}, we obtain $\mathcal{O}(n \log n)$ worst-case bound,
and the bounds we obtain are asymptotically optimal.
Note that previous methods are either not applicable,
or grossly over-estimate the bounds as $\mathcal{O}(n^2)$.

\item For {\em Karatsuba's algorithm} for polynomial multiplication (cf.~\cite{KnuthAllBooks}) we obtain
a bound of $\mathcal{O}(n^{1.6})$, whereas the optimal bound is $n^{\log_23}\approx \mathcal{O}(n^{1.585})$, and
for the classical {\em Strassen's algorithm} for fast matrix multiplication (cf.~\cite[Chapter 4]{DBLP:books/daglib/0023376})
we obtain a bound of $\mathcal{O}(n^{2.9})$ whereas the optimal bound is $n^{\log_27} \approx
\mathcal{O}(n^{2.8074})$.
Note that previous methods are either not applicable, or
grossly over-estimate the bounds as $\mathcal{O}(n^2)$ and $\mathcal{O}(n^3)$, respectively.
\end{compactitem}

\item We present experimental results to demonstrate the effectiveness of our approach.

\end{compactenum}

\noindent{\em Applicability.}
In general, our approach can be applied to (recursive) programs where the worst-case behaviour can be
obtained by an analysis that involves only the structure of the program.
For example, our approach cannot handle  the Euclidean algorithm for computing the greatest common divisor of two given natural numbers, since the worst-case behaviour of this algorithm relies on Lam\'{e}'s Theorem~\cite{KnuthAllBooks}.

\smallskip\noindent{\em Key Novelty.}
The key novelty of our approach is that we show how non-trivial
non-polynomial worst-case upper bounds such as $\mathcal{O}(n \log n)$
and $\mathcal{O}(n^r)$, where $r$ is non-integral, can be soundly obtained,
even for recursive programs, using linear programming.
Moreover, as our computational tool is linear programming, the approach we
provide is also a relatively scalable one (see Remark~\ref{rem:novelty}).


\newcommand{\synalgo}{{\sc SynAlgo}}
\vspace{-1.2em}
\section{Non-deterministic Recursive Programs}
\vspace{-0.8em}
In this work, our main contributions involve a new approach for non-polynomial worst-case analysis of
recursive programs.
To focus on the new contributions, we consider a simple programming language
for non-deterministic recursive programs.
In our language,
(a)~all scalar variables hold integers,
(b)~all assignments to scalar variables are restricted to linear expressions with floored operation,
and (c)~we do not consider return statements.
The reason to consider such a simple language is that (i) non-polynomial worst-case running time often involves  non-polynomial terms over integer-valued variables (such as array length) only, (ii) assignments to variables are often linear with possible floored expressions (in e.g. divide-and-conquer programs) and (iii) return value is often not related to worst-case behaviour of programs.


For a set $A$, we denote by $|A|$ the cardinality of $A$ and $\mathbf{1}_A$ the indicator function on $A$.
We denote by $\Nset$, $\Nset_0$, $\Zset$, and $\Rset$ the sets of all positive
integers, non-negative integers, integers, and real numbers, respectively.
Below we fix a set $\mathcal{X}$ of \emph{scalar} variables.

\smallskip\noindent{\bf Arithmetic Expressions, Valuations, and Predicates.}
The set of \emph{(linear) arithmetic expressions} $\mathfrak{e}$ over $\mathcal{X}$ is generated by the following grammar:
$\mathfrak{e}::= c\mid x  \mid \left\lfloor \frac{\mathfrak{e}}{c}\right\rfloor\mid \mathfrak{e}+\mathfrak{e}\mid \mathfrak{e}-\mathfrak{e}\mid c*\mathfrak{e}$
where $c\in\Zset$ and $x\in\mathcal{X}$.
Informally, (i) $\frac{\centerdot}{c}$ refers to division operation,
(ii) $\lfloor\centerdot\rfloor$ refers to the floored operation, and
(iii) $+,-,*$ refer to addition, subtraction and multiplication operation
over integers, respectively.
In order to make sure that division is well-defined, we stipulate that every appearance of $c$ in $\frac{\mathfrak{e}}{c}$ is non-zero.
A \emph{valuation} over $\mathcal{X}$ is a function $\nu$ from $\mathcal{X}$ into $\Zset$.
Informally, a valuation assigns to each scalar variable an integer.
Under a valuation $\nu$ over $\mathcal{X}$, an arithmetic expression $\mathfrak{e}$ can be \emph{evaluated} to an integer in the straightforward way.
We denote by $\mathfrak{e}(\nu)$ the evaluation of $\mathfrak{e}$ under $\nu$.
The set of \emph{propositional arithmetic predicates} $\phi$ over $\mathcal{X}$ is generated by the following grammar:
$\phi::= \mathfrak{e}\le\mathfrak{e} \mid \mathfrak{e}\ge\mathfrak{e} \mid \neg\phi \mid \phi\wedge\phi\mid \phi\vee\phi$
where $\mathfrak{e}$ represents an arithmetic expression.
The satisfaction relation $\models$ between valuations and propositional arithmetic predicates is defined in the straightforward way through evaluation of arithmetic expressions (cf. Appendix~\ref{app:predicate} for details).
For each propositional arithmetic predicate $\phi$, $\mathbf{1}_{\phi}$ is interpreted as the indicator function $\nu\mapsto\mathbf{1}_{\nu\models\phi}$ on valuations, where $\mathbf{1}_{\nu\models\phi}$ is $1$ if $\nu\models\phi$ and $0$ otherwise.


\smallskip\noindent{\bf Syntax of the Programming Language.}
Due to page limit, we present a brief description of our syntax. 
The syntax is essentially a subset of C programming language:
in our setting, we have \emph{scalar variables} which hold integers and \emph{function names}
which corresponds to functions (in programming-language sense);
assignment statements are indicated by `$:=$', whose left-hand-side is a scalar variable and whose right-hand-side is a linear arithmetic expression; 
`\textbf{skip}' is the statement which does nothing;
while-loops and conditional if-branches are indicated by `\textbf{while}' and `\textbf{if}' respectively, together with a propositional arithmetic predicate indicating the relevant condition (or guard);
demonic non-deterministic branches are indicated by `\textbf{if}' and `$\star$';
function declarations are indicated by a function name followed by a bracketed list of non-duplicate scalar variables, while function calls are indicated by a function name followed by a bracketed list of linear arithmetic expressions;
each function declaration is followed by a curly-braced compound statement as function body;
finally, a program is a sequence of function declarations with their function bodies.
(cf. Appendix~\ref{app:detailedsyntax} for details).

\noindent{\em Statement Labeling.}
Given a recursive program in our syntax, we assign a distinct natural number (called \emph{label} in our context)
to every assignment/skip statement, function call, if/while-statement and terminal line in the program.
Informally, each label serves as a program counter which indicates the next statement to be executed.


\smallskip\noindent{\bf Semantics through CFGs.}
We use control-flow graphs (CFGs) to specify the semantics of recursive
programs.
Informally, a CFG specifies how values for scalar variables and the program counter change in a program.

\begin{definition}[Control-Flow Graphs]\label{def:cfg}
A \emph{control-flow graph} (CFG) is a triple which takes the form
$(\dag)~\left(\fnames,\left\{\left(\locs{f}, \clocs{f}, \alocs{f}, \flocs{f},\dlocs{f}, \pvars{f},\lin{\fn{f}},\lout{\fn{f}}\right)\right\}_{\fn{f}\in\fnames}, \left\{\transitions{\fn{f}}\right\}_{\fn{f}\in\fnames}\right)$
where:
\begin{compactitem}
\item $\fnames$ is a finite set of \emph{function names};
\item each $\locs{f}$ is a finite set of \emph{labels} attached to the function name $\fn{f}$, which is partitioned into (i) the set $\clocs{f}$ of \emph{branching} labels, (ii) the set $\alocs{f}$ of \emph{assignment} labels, (iii) the set $\flocs{f}$ of \emph{call} labels and (iv) the set $\dlocs{f}$ of \emph{demonic non-deterministic} labels;
\item each $\pvars{f}$ is the set of \emph{scalar variables} attached to $\fn{f}$;
\item each $\lin{\fn{f}}$ (resp. $\lout{\fn{f}}$) is the \emph{initial label} (resp. \emph{terminal label}) in $\locs{f}$;
\item each $\transitions{\fn{f}}$ is a relation whose every member is a triple of the form $(\loc,\alpha,\loc')$ for which $\loc$ (resp. $\loc'$) is the source label (resp. target label) of the triple such that $\loc\in\locs{f}$ (resp. $\loc'\in\locs{f}$), and $\alpha$ is (i) either a propositional arithmetic predicate $\phi$ over $\pvars{f}$ (as the set of scalar variables) if $\loc\in\clocs{f}$, (ii) or an \emph{update function} from the set of valuations over $\pvars{f}$ into the set of valuations over $\pvars{f}$ if $\loc\in\alocs{f}$, (iii) or a pair $(\fn{g}, h)$ with $\fn{g}\in\fnames$ and $h$ being a \emph{value-passing function} which maps every  valuation over $\pvars{f}$ to a valuation over $\pvars{g}$ if $\loc\in\flocs{f}$, (iv) or $\star$ if $\loc\in\dlocs{f}$.
\end{compactitem}
\end{definition}
W.l.o.g, we consider that all labels are natural numbers.
We denote by $\Val{f}$ the set of valuations over $\pvars{f}$, for each $\fn{f}\in\fnames$. Informally, a function name $\fn{f}$, a label $\loc\in\locs{f}$ and a valuation $\nu\in\Val{f}$ reflects that the current status of a recursive program is under function name $\fn{f}$, right before the execution of the statement labeled $\loc$ in the function body named $\fn{f}$ and with values specified by $\nu$, respectively.

\lstset{language=prog}
\lstset{tabsize=3}
\newsavebox{\progbinarysearch}
\begin{lrbox}{\progbinarysearch}
\begin{lstlisting}[mathescape]
$\mathsf{f}(n)$ {
1:  if $n\ge 2$ then
2:      $\mathsf{f}(\lfloor \frac{n}{2}\rfloor)$
3:    else  skip
    fi
4: }
\end{lstlisting}
\end{lrbox}

\vspace{-1em}
\begin{figure}
\begin{minipage}{0.50\textwidth}
\centering
\usebox{\progbinarysearch}
\caption{A program for {\sc Binary-Search}}
\label{fig:running:binarysearch}
\end{minipage}
\begin{minipage}{0.50\textwidth}
\centering
\begin{tikzpicture}[x = 2cm]

\node[label] (if)         at (0,0)         {$1$};
\node[label] (call)       at (1.5,0)         {$2$};
\node[label] (skip)       at (0,-1.5)        {$3$};
\node[label] (end)        at (1.5,-1.5)        {$4$};

\draw[tran] (if) to node[auto, font=\scriptsize] {$n\ge 2$} (call);
\draw[tran] (call) to node[auto, font=\scriptsize] {$(\mathsf{f}, n\mapsto\lfloor{n}/{2}\rfloor)$} (end);
\draw[tran] (if) to node[auto, font=\scriptsize] {$n\le 1$} (skip);
\draw[tran] (skip) to node[auto, font=\scriptsize] {$n\mapsto n$} (end);

\end{tikzpicture}
\caption{The CFG for Figure~\ref{fig:running:binarysearch}}
\label{fig:cfg:binarysearch}
\end{minipage}
\end{figure}
\vspace{-2em}

\begin{example}\label{ex:new:running}
We consider the running example in Figure~\ref{fig:running:binarysearch}
which abstracts the running time of {\sc Binary-Search}.
The CFG for this example is depicted in Figure~\ref{fig:cfg:binarysearch}. \qed
\end{example}

For a detailed description of CFG and transformation from recursive programs to CFGs, see Appendix~\ref{app:cfg}.
Based on CFG, the semantics models executions of a recursive program as runs, and is defined through the standard notion of
call stack.
Below we fix a recursive program $P$ and its CFG taking the form ($\dag$).
We first define the notion of \emph{stack element} and {\em configurations} which captures all information within a function call.

\noindent{\bf Stack Elements and Configurations.}
A \emph{stack element} $\mathfrak{c}$ (of $P$) is a
triple $(\fn{f},\loc,\nu)$ (treated as a letter) where $\fn{f}\in\fnames$, $\loc\in\locs{f}$
and $\nu\in\Val{f}$; $\mathfrak{c}$ is non-terminal if $\loc\in\locs{f}\setminus\{\lout{\fn{f}}\}$.
A \emph{configuration} (of $P$) is a finite word of non-terminal stack elements (including the empty word $\varepsilon$).
Thus, a stack element $(\fn{f},\loc,\nu)$ specifies that the current function name is $\fn{f}$, the next statement to be executed is the one labelled with $\loc$ and the current valuation
w.r.t $\fn{f}$ is $\nu$;
a configuration captures
the whole trace of the call stack.

\noindent{\bf Schedulers and Runs.}
To resolve non-determinism indicated by $\star$, we consider the standard notion of {\em schedulers},
which have the full ability to look into the whole history for decision.
Formally, a scheduler $\pi$ is a function that maps every sequence of configurations ending in a non-deterministic
location to the next configuration.
A stack element $\mathfrak{c}$ (as the initial stack element) and a scheduler $\pi$ defines a unique infinite sequence
$\{w_j\}_{j\in\Nset_0}$ of configurations as the execution starting from $\mathfrak{c}$ and under $\pi$, which is denoted as the \emph{run} $\rho(\mathfrak{c},\pi)$.
This defines the semantics of recursive programs.

We now define the notion of termination time which corresponds
directly to the running time of a recursive program.
In our setting, execution of every step
takes one time unit.

\vspace{-0.5em}
\begin{definition}[Termination Time]\label{def:tertime}
For each stack element $\mathfrak{c}$ and each scheduler $\pi$,
the \emph{termination time} of the run $\rho(\mathfrak{c},\pi)=\{w_j\}_{j\in\Nset_0}$, denoted by $T(\mathfrak{c},\pi)$,
is defined as $T(\mathfrak{c},\pi):= \min\{j\mid w_j=\varepsilon\}$ (i.e., the earliest time when the stack is empty)
where $\min\emptyset:=\infty$.
For each stack element $\mathfrak{c}$, the \emph{worst-case termination-time function}
$\overline{T}$ is a function on the set of stack elements defined by:
$\overline{T}(\mathfrak{c}):=
\sup\{T(\mathfrak{c}, \pi)\mid \pi\mbox{ is a scheduler for }P\}$.
\end{definition}
\vspace{-0.8em}
Thus $\overline{T}$ captures the worst-case behaviour of the recursive program $P$.


\vspace{-1.5em}
\section{Measure Functions}\label{sect:mfunc}
\vspace{-1em}
In this section, we introduce the notion of measure functions for recursive programs.
We show that measure functions are sound and complete for nondeterministic recursive programs
and serve as upper bounds for the worst-case termination-time function.
In the whole section, we fix a recursive program $P$ together with its CFG taking the form ($\dag$).
We now present the standard notion of {\em invariants} which represent reachable stack elements.
Due to page limit, we omit the intuitive notion of \emph{reachable stack elements}.
Informally, a stack element is \emph{reachable} w.r.t an initial function name and initial valuations satisfying a prerequisite (as a propositional arithmetic predicate) if it can appear in the run under some scheduler (cf. Definition~\ref{def:reachability} for more details).


\vspace{-0.8em}
\begin{definition}[Invariants]\label{def:inv}
A (linear) \emph{invariant} $\inv$ w.r.t a function name $\fn{f}^*$ and a propositional
arithmetic predicate $\phi^*$ over $\pvars{f^*}$ is a function that upon
any pair $(\fn{f},\loc)$ satisfying $\fn{f}\in\fnames$ and
$\loc\in\locs{f}\backslash\{\lout{\fn{f}}\}$, $I(\fn{f},\loc)$ is a propositional arithmetic predicate
over $\pvars{f}$
such that (i) $I(\fn{f},\loc)$ is without the appearance of floored expressions (i.e. $\lfloor\centerdot\rfloor$) and (ii) for all stack elements $(\fn{f},\loc,\nu)$ reachable w.r.t
$\fn{f}^*,\phi^*$, $\nu\models I(\fn{f},\loc)$.
The invariant $I$ is in \emph{disjunctive normal form} if every $I(\fn{f},\loc)$ is in disjunctive normal form.
\end{definition}
\vspace{-0.8em}

Obtaining invariants automatically is a standard problem in programming languages,
and several techniques exist (such as abstract interpretation~\cite{DBLP:conf/popl/CousotC77} or Farkas' Lemma~\cite{DBLP:conf/cav/ColonSS03}).
In the rest of the section we fix a(n initial) function name $\fn{f}^*\in\fnames$ and a(n initial)
propositional arithmetic predicate $\phi^*$ over $\pvars{f^*}$.
For each $\fn{f}\in\fnames$ and $\loc\in\locs{f}\backslash\{\lout{\fn{f}}\}$,
we define $D_{\fn{f},\loc}$ to be the set of all valuations $\nu$ w.r.t $\fn{f}$
such that $(\fn{f},\loc,\nu)$ is reachable w.r.t $\fn{f}^*,\phi^*$.
Below we introduce the notion of measure functions.

\vspace{-0.8em}
\begin{definition}[Measure Functions]\label{def:mfunc}
A \emph{measure function} w.r.t $\fn{f}^*,\phi^*$ is a function $g$ from
the set of stack elements into $[0,\infty]$ such that for all stack elements $(\fn{f},\loc,\nu)$, the following conditions hold:
\begin{compactitem}
\item \textbf{C1:} if $\loc=\lout{\fn{f}}$, then $g(\fn{f},\loc,\nu)=0$;
\item \textbf{C2:} if $\loc\in\alocs{f}\backslash\{\lout{\fn{f}}\}$, $\nu\in D_{\fn{f},\loc}$ and $(\loc,h,\loc')$ is the only triple in $\transitions{\fn{f}}$ with source label $\loc$ and update function $h$, then
$g(\fn{f},\loc',h(\nu))+1\le g(\fn{f},\loc,\nu)$;

\item \textbf{C3:} if $\loc\in\flocs{f}\backslash\{\lout{\fn{f}}\}$, $\nu\in D_{\fn{f},\loc}$ and $(\loc,(\fn{g},h),\loc')$ is the only triple in $\transitions{\fn{f}}$ with source label $\loc$ and value-passing function $h$, then
$1+g(\fn{g},\lin{\fn{g}},h(\nu))+g(\fn{f},\loc',\nu)\le g(\fn{f},\loc,\nu)$;
\item \textbf{C4:} if $\loc\in\clocs{f}\backslash\{\lout{\fn{f}}\}$, $\nu\in D_{\fn{f},\loc}$ and $(\loc, \phi,\loc_1),(\loc, \neg\phi,\loc_2)$ are namely two triples in $\transitions{\fn{f}}$ with source label $\loc$, then
$\mathbf{1}_{\nu\models\phi}\cdot g(\fn{f},\loc_1,\nu)+\mathbf{1}_{\nu\models\neg\phi}\cdot g(\fn{f},\loc_2,\nu)+1\le g(\fn{f},\loc,\nu)$;
\item \textbf{C5:} if $\loc\in\dlocs{f}\backslash\{\lout{\fn{f}}\}$, $\nu\in D_{\fn{f},\loc}$ and $(\loc, \star,\loc_1),(\loc, \star,\loc_2)$ are namely two triples in $\transitions{\fn{f}}$ with source label $\loc$, then
$\max\{g(\fn{f},\loc_1,\nu), g(\fn{f},\loc_2,\nu)\}+1\le g(\fn{f},\loc,\nu)$.
\end{compactitem}
\end{definition}
\vspace{-0.8em}
Intuitively, a measure function is a non-negative function whose values strictly decrease along the executions
regardless of the choice of the demonic scheduler.
By applying ranking functions to configurations, one can prove the following theorem stating that measure functions are sound and complete for the worst-case termination-time function.
The technical proof of the theorem is put in Appendix~\ref{app:thmmfuncsoundness} and Appendix~\ref{app:thmmfunc}.
\vspace{-0.8em}
\begin{theorem}[Soundness and Completeness]\label{thm:mfunc}
(1) {\em (Soundness).} For all measure functions $g$ w.r.t $\fn{f}^*,\phi^*$, it holds that for all valuations $\nu\in\Aval{\fn{f}^*}$ such that $\nu\models\phi^*$, we have $\overline{T}(\fn{f}^*,\lin{\fn{f}^*},\nu)\le g(\fn{f}^*,\lin{\fn{f}^*},\nu)$.
(2) {\em (Completeness).} $\overline{T}$ is a measure function w.r.t $\fn{f}^*,\phi^*$.
\end{theorem}
\vspace{-0.8em}

By Theorem~\ref{thm:mfunc}, to obtain an upper bound on the worst-case termination-time function,
it suffices to synthesize a measure function. Below we show that it suffices to
synthesize measure functions at cut-points (which we refer as {\em significant} labels).
\vspace{-0.8em}
\begin{definition}[Significant Labels]\label{def:slocs}
Let $\fn{f}\in\fnames$.
A label $\loc\in\locs{f}$ is \emph{significant} if
either $\loc=\lin{\fn{f}}$ or $\loc$ is the initial label to some while-loop appearing in the function body of $\fn{f}$.
\end{definition}
\vspace{-0.8em}

We denote by $\slocs{f}$ the set of significant locations in $\locs{f}$.
Informally, a significant label is a label where valuations cannot be easily deduced from other labels,
namely valuations at the start of the function-call and at the initial label of a while loop.

\noindent{\bf The Expansion Construction (from $g$ to $\widehat{g}$).}
Let $g$ be a function from
$\left\{(\fn{f},\loc,\nu)\mid \fn{f}\in\fnames, \loc\in\slocs{f},  \nu\in\Aval{f}\right\}$
into $[0,\infty]$.
One can obtain from $g$ a function $\widehat{g}$ from the set of all stack elements into $[0,\infty]$ in a straightforward way through iterated application of the equality forms of C1--C5 (cf. Appendix~\ref{app:propmfunc} for details).


\newcommand{\Sat}{\mathsf{Sat}}

\vspace{-1.5em}
\section{The Synthesis Algorithm}\label{sect:synalg}
\vspace{-0.5em}
By Theorem~\ref{thm:mfunc}, measure functions are a sound approach for upper bounds of the worst-case
termination-time function, and hence synthesis of measure functions of specific
forms provide upper bounds for worst-case behaviour of recursive programs.
We first define the synthesis problem of measure functions and then
present the synthesis algorithm, where the initial stack element is integrated into the input invariant.
Informally, the input is a recursive program, an invariant for the program and technical parameters for the specific form of a measure function, and the output is a measure function if the algorithm finds one, and fail otherwise.

\noindent{\bf The \recterm\ Problem.}
The \recterm\ problem is defined as follows:
\begin{compactitem}
\item \emph{Input:} a recursive program $P$, an invariant $\inv$ in disjunctive normal form and a quadruple $(d, \mathrm{op}, r,k)$ of technical parameters;
\item \emph{Output:} a measure function $h$ w.r.t the quadruple $(d, \mathrm{op}, r,k)$.
\end{compactitem}
The quadruple $(d, \mathrm{op}, r,k)$ specifies the form of a measure function in the way that $d\in\Nset$ is the degree of the measure function to be synthesized, $\mathrm{op}\in\{\mathrm{log}, \mathrm{exp}\}$ signals either logarithmic (when $\mathrm{op}=\log$) (e.g., $n\ln{n}$) or exponential (when $\mathrm{op}=\mathrm{exp}$) (e.g., $n^{1.6}$) measure functions, $r$ is a rational number greater than $1$ which specifies the exponent in the measure function (i.e., $n^r$) when $\mathrm{op}=\mathrm{exp}$ and $k\in\Nset$ is a technical parameter required by Theorem~\ref{thm:handelman}.

\vspace{-0.8em}
\begin{remark}
In the input for \recterm\ we fix the exponent $r$ when $\mathrm{op}=\mathrm{exp}$.
However, iterating with binary search over an input bounded range we can obtain a
measure function in the given range as precise as possible.
Moreover, the invariants can be obtained automatically through e.g.~\cite{DBLP:conf/cav/ColonSS03}.
\qed
\end{remark}

\vspace{-0.8em}
We present our algorithm \synalgo{} for synthesizing measure
functions  for the \recterm\ problem.
The algorithm is
designed to synthesize one function over
valuations at each function name and appropriate significant
labels  so that C1--C5 are fulfilled.
Due to page limit, we illustrate the main conceptual details of our algorithm (more
details are relegated to Appendix~\ref{app:synalgo}). Below we fix an input to our algorithm.

\smallskip\noindent{\bf Overview.}
We present the overview of our solution which has the following five steps.
\begin{compactenum}
\item {\em Step~1.} Since one key aspect of our result is to obtain bounds of the form
$\mathcal{O}(n \log n)$ as well as $\mathcal{O}(n^r)$, where $r$ is not an integer,
we first consider general form of upper bounds that involve logarithm and exponentiation
(Step~1(a)), and then consider templates with the general form of upper bounds for significant
labels (Step~1(b)).

\item {\em Step~2.} The second step considers the template generated in Step~1 for significant
labels and generate templates for all labels. This step is relatively straightforward.

\item {\em Step~3.} The third step establishes constraint triples according to the invariant
given by the input and the template obtained in Step~2. This step is also straightforward.

\item {\em Step~4.} The fourth step is the significant step which involves transforming the constraint
triples generated in Step~3 into ones without logarithmic and exponentiation terms.
The first substep (Step~4(a)) is to consider abstractions of logarithmic, exponentiation, and floored expressions as fresh variables.
The next step (Step~4(b)) requires to obtain linear constraints over the abstracted variables.
We use Farkas' lemma and Lagrange's Mean-Value Theorem (LMVT) to obtain sound linear inequalities
for those variables.

\item {\em Step~5.} The final step is to solve the unknown coefficients of the template from the constraint triples (without logarithm or exponentiation) obtained from Step 4.
This requires the solution of positive polynomials over polyhedrons through the sound form of Handelman's Theorem (Theorem~\ref{thm:handelman}) to transform into a linear program.
\end{compactenum}
We first present an informal illustration
of the key ideas through a simple example.

\vspace{-0.8em}
\begin{example}\label{ex:narrationb}
Consider the task to synthesize a measure function for Karatsuba's algorithm~\cite{KnuthAllBooks} for polynomial multiplication which runs in $c\cdot n^{1.6}$ steps, where $c$ is a coefficient to be synthesized and $n$ represents the maximal degree of the input polynomials and is a power of $2$.
We describe informally how our algorithm tackles Karatsuba's algorithm.
Let $n$ be the length of the two input polynomials and $c\cdot n^{1.6}$ be the template.
Since Karatsuba's algorithm involves three sub-multiplications and seven additions/subtractions, the condition C3 becomes (*) $c\cdot n^{1.6}-3\cdot c\cdot \left(\frac{n}{2}\right)^{1.6}-7\cdot n\ge 0$
for all $n\ge 2$.
The algorithm first abstracts $n^{1.6}$ as a stand-alone variable
$u$.
Then the algorithm generates the following inequalities through properties of exponentiation:
(**) $u\ge 2^{1.6}, u\ge 2^{0.6}\cdot n$.
Finally, the algorithm transforms (*) into (***)
$c\cdot u-3\cdot \left(\frac{1}{2}\right)^{1.6}\cdot c\cdot u-7\cdot n\ge 0$
and synthesizes a value for $c$ through Handelman's Theorem to ensure that (***) holds under $n\ge 2$ and (**).
One can verify that $c=1000$ is a feasible solution since
\begin{align*}
  & \left(1000-3000\cdot \left({1}/{2}\right)^{1.6}\right)\cdot u - 7\cdot n = \\
  &~~\frac{7}{2^{0.6}}\cdot \left(u-2^{0.6}\cdot n\right)+\frac{1000\cdot 2^{1.6}-3014}{2^{1.6}}\cdot \left(u-2^{1.6}\right)+\left(1000\cdot 2^{1.6}-3014\right)\cdot 1.
\end{align*}
Hence, Karatsuba's algorithm runs in $\mathcal{O}(n^{1.6})$ time.\qed
\end{example}


\vspace{-1.5em}
\subsection{Step~1 of \synalgo}
\vspace{-0.8em}

\noindent{\bf Step~1(a): General Form of A Measure Function.}

\noindent{\em Extended Terms.}
In order to capture non-polynomial worst-case complexity of recursive programs, our algorithm incorporates
two types of extensions of terms.
\begin{compactenum}
\item {\em Logarithmic Terms.}
The first extension, which we call $\log{}$-extension, is the extension with terms from
$\ln{x},\ln{(x-y+1)}$
where $x,y$ are scalar variables appearing in the parameter list of some function name and $\ln{(\centerdot)}$ refers to the natural logarithm function with base $e$.
Our algorithm will take this extension when $\mathrm{op}$ is $\log$.

\item {\em Exponentiation Terms.}
The second extension, which we call $\mathrm{exp}$-extension, is with terms from
$x^r,(x-y+1)^r$
where $x,y$ are scalar variables appearing in the parameter list of some function name.
The algorithm takes this when $\mathrm{op}=\mathrm{exp}$.

\end{compactenum}
The intuition is that $x$ (resp. $x-y+1$) may represent a positive quantity to be halved iteratively (resp. the length between array indexes $y$ and $x$).

\smallskip\noindent{\em General Form.}
The general form for any coordinate function $\eta(\fn{f}, \loc, \centerdot)$ of a measure function $\eta$ (at function name $\fn{f}$ and $\loc\in\slocs{f}$) is a finite sum
\begin{equation}\label{eq:synform}
\textstyle\mathfrak{e}=\sum_{i}c_i\cdot g_{i}
\end{equation}
where (i) each $c_i$ is a constant scalar and each $g_{i}$ is a finite product of no more than $d$ terms (i.e., with degree at most $d$) from scalar variables in $\pvars{f}$ and logarithmic/exponentiation extensions
(depending on $\mathrm{op}$), and (ii) all $g_{i}$'s correspond to all finite products of no more than $d$ terms.
Analogous to arithmetic expressions, for any such finite sum $\mathfrak{e}$ and any valuation $\nu\in\Aval{f}$, we denote by $\mathfrak{e}(\nu)$ the real number evaluated through replacing any scalar variable $x$ appearing in $\mathfrak{e}$ with $\nu(x)$, provided that $\mathfrak{e}(\nu)$ is well-defined.

\smallskip\noindent{\em Semantics of General Form.}
A finite sum $\mathfrak{e}$ at $\fn{f}$ and $\loc\in\slocs{f}$ in the form (\ref{eq:synform}) defines a function $\llbracket\mathfrak{e}\rrbracket$ on $\Aval{f}$ in the way that for each  $\nu\in\Aval{f}$:
$\llbracket\mathfrak{e}\rrbracket(\nu):=\mathfrak{e}(\nu)$ if $\nu\models I(\fn{f},\loc)$, and $\llbracket\mathfrak{e}\rrbracket(\nu):=0$ otherwise.
Note that in the definition of $\llbracket\mathfrak{e}\rrbracket$, we do not consider the case when $\log$ or exponentiation is undefined.
However, we will see in Step 1(b) below that $\log{}$ or exponentiation will always be well-defined.

\smallskip\noindent{\bf Step~1(b): Templates.}
As in all previous works (cf.~\cite{DBLP:conf/tacas/ColonS01,DBLP:conf/vmcai/PodelskiR04,DBLP:conf/vmcai/Cousot05,DBLP:journals/fcsc/YangZZX10,DBLP:journals/jossac/ShenWYZ13,SriramCAV,DBLP:conf/popl/ChatterjeeFNH16,Hoffman1}),
we consider a template for measure function determined by the triple $(d,\mathrm{op},r)$ from the input parameters.
Formally, the template determined by $(d,\mathrm{op},r)$ assigns to every function name $\fn{f}$ and $\loc\in\slocs{f}$ an expression in the form (\ref{eq:synform}) (with degree $d$ and extension option $\mathrm{op}$).
Note that a template here only restricts (i) the degree and (ii) $\log$ or $\mathrm{exp}$ extension for a measure function, rather than its specific form.
Although $r$ is fixed when $\mathrm{op}=\mathrm{exp}$, one can perform a binary search over a bounded range for a suitable or optimal $r$.

In detail, the algorithm sets up a template $\eta$ for a measure function by assigning to each function name $\fn{f}$ and significant label $\loc\in\slocs{f}$ an expression $\eta(\fn{f},\loc)$ in a form similar to
(\ref{eq:synform}), except for that $c_i$'s in (\ref{eq:synform}) are interpreted as distinct \emph{template variables} whose actual values are to be synthesized.
In order to ensure that logarithm and exponentiation are well-defined over each $I(\fn{f},\loc)$, we impose the following restriction ($\S$) on our template:
\begin{quote}
(\S) $\ln{x},x^r$ (resp. $\ln{(x-y+1)},(x-y+1)^r$) appear in $\eta(\fn{f},\loc)$ only when $x-1\ge 0$ (resp. $x-y\ge 0$) can be inferred from the invariant $I(\fn{f},\loc)$.
\end{quote}
To infer $x-1\ge 0$ or $x-y\ge 0$ from $I(\fn{f},\loc)$, we utilize Farkas' Lemma.

\begin{theorem}[Farkas' Lemma~\cite{FarkasLemma,SchrijverPolyhedra}]\label{thm:farkas}
Let $\mathbf{A}\in\Rset^{m\times n}$, $\mathbf{b}\in\Rset^m$, $\mathbf{c}\in\Rset^{n}$ and $d\in\Rset$.
Assume that $\{\mathbf{x}\mid \mathbf{A}\mathbf{x}\le \mathbf{b}\}\ne\emptyset$.
Then
$\{\mathbf{x}\mid \mathbf{A}\mathbf{x}\le \mathbf{b}\}\subseteq \{\mathbf{x}\mid \mathbf{c}^{\mathrm{T}}\mathbf{x}\le d\}$
iff there exists $\mathbf{y}\in\Rset^m$ such that $\mathbf{y}\ge \mathbf{0}$, $\mathbf{A}^\mathrm{T}\mathbf{y}=\mathbf{c}$ and $\mathbf{b}^{\mathrm{T}}\mathbf{y}\le d$.
\end{theorem}

By Farkas' Lemma, there exists an algorithm that infers whether $x-1\ge 0$ (or $x-y\ge 0$) holds under $I(\fn{f},\loc)$ in polynomial time through emptiness checking of polyhedra (cf.~\cite{DBLP:books/daglib/0090562}) since  $I(\fn{f},\loc)$ involves only linear (degree-$1$) polynomials in our setting.

Then $\eta$ naturally induces a function $\llbracket\eta\rrbracket$ from
$\left\{(\fn{f},\loc,\nu)\mid \fn{f}\in\fnames, \loc\in\slocs{f},  \nu\in\Aval{f}\right\}$
into $[0,\infty]$ parametric over template variables such that  $\llbracket\eta\rrbracket(\fn{f},\loc,\nu)={\llbracket}{\eta(\fn{f},\loc)}{\rrbracket}(\nu)$ for all appropriate stack elements $(\fn{f},\loc,\nu)$.
Note that $\llbracket\eta\rrbracket$ is well-defined since logarithm and exponentiation is well-defined over satisfaction sets given by $I$.

\vspace{-1.5em}
\subsection{Step~2 of \synalgo}
\vspace{-0.8em}

\noindent{\bf Step~2: Computation of $\widehat{\llbracket\eta\rrbracket}$.} Let $\eta$ be the template constructed from Step 1.
This step computes $\widehat{\llbracket\eta\rrbracket}$ from $\eta$ by the expansion construction of
significant labels (Section~\ref{sect:mfunc})
which transforms a function $g$ into $\widehat{g}$.
Recall the function $\llbracket\mathfrak{e} \rrbracket$ for $\mathfrak{e}$ is defined in Step 1(a).
Formally, based on the template $\eta$ from Step 1, the algorithm computes $\widehat{\llbracket\eta\rrbracket}$, with the exception that template variables appearing in $\eta$ are treated as undetermined constants.
Then $\widehat{\llbracket\eta\rrbracket}$ is a function parametric over the template variables in $\eta$.

By an easy induction, each $\widehat{\llbracket\eta\rrbracket}(\fn{f},\loc,\centerdot)$ can be represented by an expression in the form:
\begin{equation}\label{eq:hatform}
\max\big\{\textstyle\sum_{j}\mathbf{1}_{\phi_{1j}}\cdot h_{1j},~\dots,~\textstyle\sum_{j}\mathbf{1}_{\phi_{mj}}\cdot h_{mj}\big\}
\end{equation}
\begin{compactenum}
\item each $\phi_{ij}$ is a propositional arithmetic predicate over $\pvars{f}$ such that for each $i$, $\bigvee_j\phi_{ij}$ is tautology and $\phi_{ij_1}\wedge\phi_{ij_2}$ is unsatisfiable whenever $j_1\ne j_2$, and
\item each $h_{ij}$ takes the form similar to (\ref{eq:synform}) with the difference that  (i) each $c_i$ is either a scalar or a template variable appearing in $\eta$ and (ii) each $g_i$ is a finite product whose every multiplicand is either some $x\in\pvars{f}$, or some $\lfloor\mathfrak{e}\rfloor$ with $\mathfrak{e}$ being an instance of $\langle\mathit{expr}\rangle$, or some $\ln{\mathfrak{e}}$ (or $\mathfrak{e}^r$, depending on $\mathrm{op}$) with $\mathfrak{e}$ being an instance of $\langle\mathit{expr}\rangle$.
\end{compactenum}
For this step we use the fact that all propositional arithmetic predicates can be put in disjunctive normal form.
For detailed description see Appendix~\ref{app:synalgo}.

\vspace{-1.5em}
\subsection{Step~3 of \synalgo}
\vspace{-0.8em}

This step generates constraint triples from $\widehat{\llbracket\eta\rrbracket}$ computed in Step 2.
By applying non-negativity and C2-C5 to $\widehat{\llbracket\eta\rrbracket}$ (computed in Step 2), the algorithm establishes constraint triples which will be interpreted as universally-quantified logical formulas later.

\smallskip\noindent{\bf Constraint Triples.} A \emph{constraint triple} is a triple $(\fn{f}, \phi,\mathfrak{e})$ where (i) $\fn{f}\in\fnames$, (ii) $\phi$ is a propositional arithmetic predicate over $\pvars{f}$ which is a conjunction of atomic formulae of the form $\mathfrak{e}'\ge 0$ with $\mathfrak{e}'$ being an arithmetic expression, and (iii) $\mathfrak{e}$ is an expression taking the form similar to (\ref{eq:synform})
with the difference that (i) each $c_i$ is either a scalar, or a template variable $c$ appearing in $\eta$,
or its reverse $-c$, and (ii) each $g_i$ is a finite product whose every multiplicand is
either some $x\in\pvars{f}$, or some $\lfloor\mathfrak{e}\rfloor$ with $\mathfrak{e}$ being an instance of $\langle\mathit{expr}\rangle$,
or some $\ln{\mathfrak{e}}$ (or $\mathfrak{e}^r$, depending on $\mathrm{op}$) with $\mathfrak{e}$ being an instance of $\langle\mathit{expr}\rangle$.

For each constraint triple $(\fn{f}, \phi,\mathfrak{e})$, the function $\llbracket\mathfrak{e}\rrbracket$ on $\Aval{f}$ is defined in the way such that each $\llbracket\mathfrak{e}\rrbracket(\nu)$ is the evaluation result of $\mathfrak{e}$ when assigning $\nu(x)$ to each $x\in\pvars{f}$;
under (\S) (of Step~1(b)), logarithm and exponentiation will always be well-defined.

\smallskip\noindent{\bf Semantics of Constraint Triples.} A constraint triple $(\fn{f}, \phi,\mathfrak{e})$ encodes the following logical formula:
$\forall \nu\in\Aval{f}.\left(\nu\models\phi\rightarrow \llbracket\mathfrak{e}\rrbracket(\nu)\ge 0\right)$\enskip.
Multiple constraint triples are grouped into a single logical formula through conjunction.

\smallskip\noindent{\bf Step~3: Establishment of Constraint Triples.}
Based on $\widehat{\llbracket\eta\rrbracket}$ (computed in the previous step), the algorithm generates constraint triples at each significant label, then group all generated constraint triples together in a conjunctive way.
To be more precise, at every significant label $\loc$ of some function name $\fn{f}$, the algorithm generates constraint triples through non-negativity of measure functions and conditions C2--C5;
after generating the constraint triples for each significant label,
the algorithm groups them together in the conjunctive fashion to form a single collection of constraint triples.
For details procedure see Appendix~\ref{app:synalgo}.

\begin{example}\label{ex:new:step1to3}
Consider our running example (cf. Example~\ref{ex:new:running}).
Let the input quadruple be $(1,\log,-,1)$ and invariant (at label 1) be $n\ge 1$ (length of array should be positive).
In Step 1, the algorithm assigns the template $\eta(\mathsf{f},1,n)=c_1\cdot n+c_2\cdot \ln n+c_3$ at label 1 and $\eta(\mathsf{f},4,n)=0$ at label 4.
In Step 2, the algorithm computes template at other labels and obtains that
$\eta(\mathsf{f},2,n)=1+c_1\cdot \left\lfloor{n}/{2}\right\rfloor+c_2\cdot \ln{\left\lfloor{n}/{2}\right\rfloor}+c_3$ and $\eta(\mathsf{f},3,n)=1$.
In Step 3, the algorithm establishes the following three constraint triples $\mathfrak{q}_1,\mathfrak{q}_2,\mathfrak{q}_3$:
\begin{compactitem}
\item $\mathfrak{q}_1:=(\mathsf{f}, n-1\ge 0, c_1\cdot n+c_2\cdot \ln n+c_3)$ from the logical formula $\forall n. (n\ge 1)\rightarrow c_1\cdot n+c_2\cdot \ln n+c_3\ge 0$ for non-negativity of measure functions;
\item $\mathfrak{q}_2:=(\mathsf{f}, n-1\ge 0\wedge 1-n\ge 0, c_1\cdot n+c_2\cdot \ln n+c_3-2)$ and $\mathfrak{q}_3:=(\mathsf{f}, n-2\ge 0, c_1\cdot (n-\left\lfloor{n}/{2}\right\rfloor) +c_2\cdot (\ln n-\ln{\left\lfloor{n}/{2}\right\rfloor})-2$ from resp. logical formulae
\begin{compactitem}
\item $\forall n. (n\ge 1\wedge n\le 1)\rightarrow c_1\cdot n+c_2\cdot \ln n+c_3\ge 2$ and
\item $\forall n. (n\ge 2)\rightarrow c_1\cdot n+c_2\cdot \ln n+c_3\ge c_1\cdot \left\lfloor{n}/{2}\right\rfloor + c_2\cdot \ln{\left\lfloor{n}/{2}\right\rfloor}+c_3+2$
\end{compactitem}
for C4 (at label 1).\qed
\end{compactitem}
\end{example}

\vspace{-1.5em}
\subsection{Step~4 of \synalgo}
\vspace{-0.8em}

\noindent{\bf Step~4: Solving Constraint Triples.}
To check whether the logical formula encoded by the generated constraint triples is valid, the algorithm follows a sound method which abstracts each multiplicand other than scalar variables in the form (\ref{eq:hatform}) as a stand-alone variable, and transforms the validity of the formula into a system of linear inequalities over template variables appearing in $\eta$ through Handelman's Theorem and linear programming.
The main idea is that the algorithm establishes tight linear inequalities for those abstraction variables by investigating properties for the abstracted arithmetic expressions, and use linear programming to solve the formula based on the linear inequalities for abstraction variables.
We note that validity of such logical formulae are generally undecidable since they involve non-polynomial terms such as logarithm~\cite{Goedel1934}.

Below we describe how the algorithm transforms a constraint triple into one without logarithmic or exponentiation term.
Given any finite set $\Gamma$ of polynomials over $n$ variables, we define
$\Sat(\Gamma):=\left\{\vec{x}\in\Rset^n\mid f(\vec{x})\ge 0\mbox{ for all }f\in\Gamma\right\}$\enskip.
In the whole step, we let $(\fn{f}, \phi,\mathfrak{e}^*)$ be any constraint triple such that
$\phi=\bigwedge_{j}\mathfrak{e}_j\ge 0$;
moreover, we maintain a finite set $\Gamma$ of linear (degree-$1$) polynomials over scalar and freshly-added variables.
Intuitively, $\Gamma$ is related to both the set of all $\mathfrak{e}_j$'s (so that $\Sat(\Gamma)$ is somehow the satisfaction set of $\phi$) and the finite subset of polynomials in Theorem~\ref{thm:handelman}.

\smallskip\noindent{\bf Step~4(a): Abstraction of Logarithmic, Exponentiation, and Floored Expressions.}
The first sub-step involves the following computational steps, where Items 2-4 handle variables for abstraction,
and Item~6 is approximation of floored expressions, and other steps are straightforward.


\begin{compactenum}

\item {\em Initialization.} First, the algorithm maintains a finite set of linear (degree-$1$) polynomials $\Gamma$ and sets it initially to the empty set.

\item {\em Logarithmic, Exponentiation and Floored Expressions.} Next, the algorithm computes the following subsets of $\langle\mathit{expr}\rangle$:
\begin{compactitem}
\item $\mathcal{E}_L:=\{ \mathfrak{e}\mid \mbox{ $\ln{\mathfrak{e}}$ appears in $\mathfrak{e}^*$ (as sub-expression)} \}$ upon $\mathrm{op}=\log$.

\item $\mathcal{E}_E:=\{ \mathfrak{e}\mid \mbox{ ${\mathfrak{e}}^r$ appears in $\mathfrak{e}^*$ (as sub-expression)} \}$ upon $\mathrm{op}=\mathrm{exp}$.

\item $\mathcal{E}_F:=\{ \mathfrak{e} \mid \mbox{ ${\mathfrak{e}}$ appears in $\mathfrak{e}^*$ and takes the form $\lfloor\frac{\centerdot}{c}\rfloor$} \}$.
\end{compactitem}
Let $\mathcal{E}:=\mathcal{E}_L \cup \mathcal{E}_E \cup \mathcal{E}_F$.

\item {\em Variables for Logarithmic, Exponentiation and Floored Expressions.}
Next, for each $\mathfrak{e}\in\mathcal{E}$, the algorithm establishes fresh variables as follows:
\begin{compactitem}
\item a fresh variable $u_{\mathfrak{e}}$ which represents $\ln{\mathfrak{e}}$ for $\mathfrak{e}\in\mathcal{E}_L$;
\item two fresh variables $v_{\mathfrak{e}},v'_{\mathfrak{e}}$ such that $v_{\mathfrak{e}}$ indicates ${\mathfrak{e}}^r$ and $v'_{\mathfrak{e}}$ for ${\mathfrak{e}}^{r-1}$
for $\mathfrak{e}\in\mathcal{E}_E$;
\item a fresh variable $w_{\mathfrak{e}}$ indicating $\mathfrak{e}$ for $\mathfrak{e}\in\mathcal{E}_F$.
\end{compactitem}
We note that $v'_\mathfrak{e}$ is introduced in order to have a more accurate approximation for $v_\mathfrak{e}$ later.
After this step, the algorithm sets $N$ to be the number of all variables (i.e., all scalar variables and all fresh variables added up to this point).
In the rest of this section, we consider an implicit linear order over all scalar and freshly-added variables so that a valuation of these variables can be treated as a vector in $\Rset^{N}$.

\item {\em Variable Substitution (from ${\mathfrak{e}}$ to $\widetilde{\mathfrak{e}}$).}
Next, for each $\mathfrak{e}$ which is either $\mathfrak{t}$ or some $\mathfrak{e}_j$ or some expression in $\mathcal{E}$,
the algorithm computes $\tilde{\mathfrak{e}}$ as the expression obtained from $\mathfrak{e}$ by substituting (i) every possible $u_{\mathfrak{e}'}$ for $\ln{\mathfrak{e}'}$, (ii) every possible $v_{\mathfrak{e}'}$ for ${\left(\mathfrak{e}'\right)}^r$ and (iii) every possible $w_{\mathfrak{e}'}$ for $\mathfrak{e}'$ such that
$\mathfrak{e}'$ is a sub-expression of $\mathfrak{e}$ which does not appear as sub-expression in some other sub-expression $\mathfrak{e}''\in\mathcal{E}_F$ of $\mathfrak{e}$.
From now on, any $\mathfrak{e}$ or $\widetilde{\mathfrak{e}}$ or is deemed as a polynomial over scalar and freshly-added variables. Then any $\mathfrak{e}(\vec{x})$ or $\widetilde{\mathfrak{e}}(\vec{x})$ is the
result of polynomial evaluation under the correspondence between variables and coordinates of $\vec{x}$ specified by the linear order.

\item {\em Importing $\phi$ into $\Gamma$.} The algorithm adds all $\widetilde{\mathfrak{e}_j}$ into $\Gamma$.

\item {\em Approximation of Floored Expressions.}
For each $\mathfrak{e}\in \mathcal{E}_F$ such that $\mathfrak{e}=\lfloor\frac{\mathfrak{e}'}{c}\rfloor$,
the algorithm adds linear constraints for $w_{\mathfrak{e}}$ recursively on the nesting depth of floor operation as follows.
\begin{compactitem}
\item {\em Base Step.} If $\mathfrak{e}=\lfloor\frac{\mathfrak{e}'}{c}\rfloor$ and $\mathfrak{e}'$ involves no nested floored expression,
then the algorithm adds into $\Gamma$ either (i)~$\widetilde{\mathfrak{e}'}-c\cdot w_{\mathfrak{e}}\mbox{ and }c\cdot w_{\mathfrak{e}}-\widetilde{\mathfrak{e}'}+c-1$
when $c\ge 1$, which is derived from
$\frac{\mathfrak{e}'}{c}-\frac{c-1}{c}\le \mathfrak{e}\le \frac{\mathfrak{e}'}{c}$\enskip,
or (ii)~
$c\cdot w_{\mathfrak{e}}-\widetilde{\mathfrak{e}'}\mbox{ and } \widetilde{\mathfrak{e}'}-c\cdot w_{\mathfrak{e}}-c-1$
when $c\le -1$, which follows from
$\frac{\mathfrak{e}'}{c}-\frac{c+1}{c}\le \mathfrak{e}\le \frac{\mathfrak{e}'}{c}$\enskip.
Second, given the current $\Gamma$, the algorithm finds the largest constant $t_{\mathfrak{e}'}$ through Farkas' Lemma such that
\[
\forall \vec{x}\in\Rset^N. \left(\vec{x}\in\Sat(\Gamma)\rightarrow \widetilde{\mathfrak{e}'}(\vec{x})\ge t_{\mathfrak{e}'}\right)
\]
holds; if such $t_{\mathfrak{e}'}$ exists, the algorithm adds the constraint
$w_{\mathfrak{e}}\ge \left\lfloor
\frac{t_{\mathfrak{e}'}}{c}\right\rfloor$
into $\Gamma$.
\item {\em Recursive Step.} If $\mathfrak{e}=\lfloor\frac{\mathfrak{e}'}{c}\rfloor$ and $\mathfrak{e}'$ involves some nested floored expression, then
the algorithm proceeds almost in the same way as for the Base Step, except that
$\widetilde{\mathfrak{e}'}$ takes the role of $\mathfrak{e}'$.
(Note that $\widetilde{\mathfrak{e}'}$ does not involve nested floored expresions.)
\end{compactitem}

\item {\em Emptiness Checking.} The algorithm checks whether $\Sat(\Gamma)$ is empty or not in polynomial time in the size of $\Gamma$
(cf.~\cite{DBLP:books/daglib/0090562}).
If $\Sat(\Gamma)=\emptyset$, then the algorithm discards this constraint triple with no linear inequalities generated, and proceeds to other constraint triples; otherwise,
the algorithm proceeds to the remaining steps.

\end{compactenum}

\begin{example}\label{ex:new:step4a}
We continue with Example~\ref{ex:new:step1to3}.
In Step 4(a), the algorithm first establishes fresh variables $u:=\ln n$, $v:=\ln{\left\lfloor {n}/{2}\right\rfloor}$ and $w:=\left\lfloor {n}/{2}\right\rfloor$, then finds that (i) $n-2\cdot w\ge 0$, (ii) $2\cdot w-n+1\ge 0$ and (iii) $n-2\ge 0$ (as $\Gamma$) implies that $w-1\ge 0$.
After Step 4(a), the constraint triples after variable substitution and their $\Gamma$'s are as follows:
\begin{compactitem}
\item $\widetilde{\mathfrak{q}_1}=(\mathsf{f}, n-1\ge 0, c_1\cdot n+c_2\cdot u+c_3)$ and $\Gamma_1=\{n-1\}$;
\item $\widetilde{\mathfrak{q}_2}=(\mathsf{f}, n-1\ge 0\wedge 1-n\ge 0, c_1\cdot n+c_2\cdot u+c_3-2)$ and $\Gamma_2=\{n-1,1-n\}$;
\item $\widetilde{\mathfrak{q}_3}:=(\mathsf{f}, n-2\ge 0, c_1\cdot (n-w) +c_2\cdot (u-v)-2)$ and $\Gamma_3=\{n-2,n-2\cdot w,2\cdot w-n+1,w-1\}$. \qed
\end{compactitem}
\end{example}

For the next sub-step we will use Lagrange's Mean-Value Theorem (LMVT) to approximate logarithmic and exponentiation terms.

\begin{theorem}[Lagrange's Mean-Value Theorem{~\cite[Chapter 6]{BasicCalculus}}]\label{thm:lagrange}
Let $f:[a,b]\rightarrow\Rset$~~(for $a<b$) be a function continuous on $[a,b]$ and differentiable on $(a,b)$.
Then there exists a real number $\xi\in (a,b)$ such that $f'(\xi)=\frac{f(b)-f(a)}{b-a}$.
\end{theorem}

\noindent{\bf Step~4(b): Linear Constraints for Abstracted Variables.}
The second sub-step consists of the following computational steps which establish into $\Gamma$ linear constraints for logarithmic or exponentiation terms.
We present the details for logarithm, while similar technical details for exponentiation terms
are in Appendix~\ref{app:synalgo}.
Below we denote by $\mathcal{E}'$ either the set $\mathcal{E}_L$ when $\mathrm{op}=\log$ or $\mathcal{E}_E$ when $\mathrm{op}=\mathrm{exp}$.
Recall the $\widetilde{\mathfrak{e}}$ notation is defined in the Variable Substitution (Item~4) of Step~4(a).

\begin{compactenum}

\item {\em Lower-Bound for Expressions in $\mathcal{E}'$.}
For each $\mathfrak{e}\in\mathcal{E}'$,
we find the largest constant $t_{\mathfrak{e}}\in\Rset$ such that the logical formula
$\forall \vec{x}\in\Rset^N. \left(\vec{x}\in\Sat(\Gamma)\rightarrow \widetilde{\mathfrak{e}}(\vec{x})\ge t_{\mathfrak{e}}\right)$
holds,
This can be solved by Farkas' Lemma and linear programming,
since $\widetilde{\mathfrak{e}}$ is linear.
Note that as long as $\Sat(\Gamma)\ne\emptyset$, it follows from (\S) (in Step~1(b))
that $t_\mathfrak{e}$ is well-defined (since $t_\mathfrak{e}$ cannot be arbitrarily large) and $t_{\mathfrak{e}}\ge 1$.

\item {\em Mutual No-Smaller-Than Inequalities over $\mathcal{E}'$.}
For each pair $(\mathfrak{e},\mathfrak{e}')\in\mathcal{E}'\times\mathcal{E}'$ such that $\mathfrak{e}\ne\mathfrak{e}'$, the algorithm finds real numbers $r_{\left(\mathfrak{e},\mathfrak{e}'\right)},b_{(\mathfrak{e},\mathfrak{e}')}$ through Farkas' Lemma and linear programming such that (i) $r_{(\mathfrak{e},\mathfrak{e}')}\ge 0$ and (ii) both the logical formulae
\[
\forall \vec{x}\in\Rset^N. \left[\vec{x}\in\Sat(\Gamma)\rightarrow \widetilde{\mathfrak{e}}(\vec{x})-\left( r_{\mathfrak{e},\mathfrak{e}'}\cdot \widetilde{\mathfrak{e}'}(\vec{x}) +b_{\mathfrak{e},\mathfrak{e}'}\right)\ge 0\right]
\quad \text{ and }
\]
\[
\forall \vec{x}\in\Rset^N. \left[\vec{x}\in\Sat(\Gamma) \rightarrow  r_{\mathfrak{e},\mathfrak{e}'}\cdot \widetilde{\mathfrak{e}'}(\vec{x}) +b_{\mathfrak{e},\mathfrak{e}'}\ge 1\right]
\]
hold.
The algorithm first finds the maximal value $r^*_{\mathfrak{e},\mathfrak{e}'}$ over all feasible $(r_{\mathfrak{e},\mathfrak{e}'},b_{\mathfrak{e},\mathfrak{e}'})$'s, then finds the maximal $b^*_{\mathfrak{e},\mathfrak{e}'}$ over all feasible $(r^*_{\mathfrak{e},\mathfrak{e}'},b_{\mathfrak{e},\mathfrak{e}'})$'s.
If such $r^*_{\mathfrak{e},\mathfrak{e}'}$ does not exist, the algorithm simply leaves $r^*_{\mathfrak{e},\mathfrak{e}'}$ undefined.
Note that once $r^*_{\mathfrak{e},\mathfrak{e}'}$ exists and  $\Sat(\Gamma)\ne\emptyset$, then $b^*_{\mathfrak{e},\mathfrak{e}'}$ exists since $b_{\mathfrak{e},\mathfrak{e}'}$ cannot be arbitrarily large once $r^*_{\mathfrak{e},\mathfrak{e}'}$ is fixed.

\item {\em Mutual No-Greater-Than Inequalities over $\mathcal{E}'$.}
For each pair $(\mathfrak{e},\mathfrak{e}')\in\mathcal{E}'\times\mathcal{E}'$ such that $\mathfrak{e}\ne\mathfrak{e}'$, the algorithm finds real numbers $\mathsf{r}_{\left(\mathfrak{e},\mathfrak{e}'\right)},\mathsf{b}_{(\mathfrak{e},\mathfrak{e}')}$ through Farkas' Lemma and linear programming such that (i) $\mathsf{r}_{(\mathfrak{e},\mathfrak{e}')}\ge 0$ and (ii) the logical formula
\[
\forall \vec{x}\in\Rset^N. \left[\vec{x}\in\Sat(\Gamma)\rightarrow \left( \mathsf{r}_{\mathfrak{e},\mathfrak{e}'}\cdot \widetilde{\mathfrak{e}'}(\vec{x}) +\mathsf{b}_{\mathfrak{e},\mathfrak{e}'}\right)-\widetilde{\mathfrak{e}}(\vec{x})\ge 0\right]
\]
holds.
The algorithm first finds the minimal value $\mathsf{r}^*_{\mathfrak{e},\mathfrak{e}'}$ over all feasible $(\mathsf{r}_{\mathfrak{e},\mathfrak{e}'},\mathsf{b}_{\mathfrak{e},\mathfrak{e}'})$'s, then finds the minimal $\mathsf{b}^*_{\mathfrak{e},\mathfrak{e}'}$ over all feasible $(\mathsf{r}^*_{\mathfrak{e},\mathfrak{e}'},\mathsf{b}_{\mathfrak{e},\mathfrak{e}'})$'s.
If such $\mathsf{r}^*_{\mathfrak{e},\mathfrak{e}'}$ does not exists, the algorithm simply leaves $\mathsf{r}^*_{\mathfrak{e},\mathfrak{e}'}$ undefined.
Note that once $\mathsf{r}^*_{\mathfrak{e},\mathfrak{e}'}$ exists and  $\Sat(\Gamma)$ is non-empty, then $\mathsf{b}^*_{\mathfrak{e},\mathfrak{e}'}$ exists since $\mathsf{b}_{\mathfrak{e},\mathfrak{e}'}$ cannot be arbitrarily small once $\mathsf{r}^*_{\mathfrak{e},\mathfrak{e}'}$ is fixed.


\item {\em Constraints from Logarithm.} For each variable $u_{\mathfrak{e}}$, the algorithm  adds into $\Gamma$ first the polynomial expression
$\widetilde{\mathfrak{e}}-\left(\mathbf{1}_{t_{\mathfrak{e}}\le e}\cdot e+\mathbf{1}_{t_{\mathfrak{e}}> e}\cdot \frac{t_{\mathfrak{e}}}{\ln t_{\mathfrak{e}}}\right)\cdot u_{\mathfrak{e}}$
from the fact that the function $z\mapsto\frac{z}{\ln{z}}$ ($z\ge 1$) has global minima at $e$ (so that the inclusion of this polynomial expression is sound), and then the polynomial expression $u_{\mathfrak{e}}-\ln{t_{\mathfrak{e}}}$ due to the definition of $t_{\mathfrak{e}}$.

\item {\em Constraints from Exponentiation.} For each variable $v_{\mathfrak{e}}$, the algorithm adds into $\Gamma$ polynomial expressions $v_{\mathfrak{e}}-t_{\mathfrak{e}}^{r-1}\cdot \widetilde{\mathfrak{e}}$ and
$v_{\mathfrak{e}}-t_{\mathfrak{e}}^{r}$ due to the definition of $t_{\mathfrak{e}}$.
And for each variable $v'_{\mathfrak{e}}$, the algorithm adds
(i) $v'_{\mathfrak{e}}-t^{r-1}_{\mathfrak{e}}$ and
(ii) either $v'_{\mathfrak{e}}-t^{r-2}_{\mathfrak{e}}\cdot \tilde{\mathfrak{e}}$
when $r\ge 2$ or
$\widetilde{\mathfrak{e}}-t^{2-r}_{\mathfrak{e}}\cdot v'_{\mathfrak{e}}$
when $1<r<2$.


\item {\em Mutual No-Smaller-Than Inequalities over $u_{\mathfrak{e}}$'s.}
For each pair $(\mathfrak{e},\mathfrak{e}')\in\mathcal{E}'\times\mathcal{E}'$ such that
$\mathfrak{e}\ne\mathfrak{e}'$ and $r^*_{\mathfrak{e},\mathfrak{e}'},b^*_{\mathfrak{e},\mathfrak{e}'}$ are successfully found and $r^*_{\mathfrak{e},\mathfrak{e}'}>0$, the algorithm adds
\[
u_{\mathfrak{e}}-\ln{r^*_{\mathfrak{e},\mathfrak{e}'}}-u_{\mathfrak{e}'}+\mathbf{1}_{b^*_{\mathfrak{e},\mathfrak{e}'}< 0}\cdot {\left(t_{\mathfrak{e}'}+\frac{b^*_{\mathfrak{e},\mathfrak{e}'}}{r^*_{\mathfrak{e},\mathfrak{e}'}}\right)}^{-1}\cdot\left(-\frac{b^*_{\mathfrak{e},\mathfrak{e}'}}{r^*_{\mathfrak{e},\mathfrak{e}'}}\right)
\]
into $\Gamma$.
This is due to the fact that
$\llbracket\mathfrak{e}\rrbracket-\left( r^*_{\mathfrak{e},\mathfrak{e}'}\cdot \llbracket\mathfrak{e}'\rrbracket +b^*_{\mathfrak{e},\mathfrak{e}'}\right)\ge 0$
implies the following:
\begin{eqnarray*}
\ln{\llbracket\mathfrak{e}\rrbracket} &\ge & \ln {r^*_{\mathfrak{e},\mathfrak{e}'}}+\ln \left(\llbracket\mathfrak{e}'\rrbracket+({b^*_{\mathfrak{e},\mathfrak{e}'}}/{r^*_{\mathfrak{e},\mathfrak{e}'}})\right) \\
&= & \ln{r^*_{\mathfrak{e},\mathfrak{e}'}}+\ln{\llbracket\mathfrak{e}'\rrbracket}+\left(\ln{ \left(\llbracket\mathfrak{e}'\rrbracket+({b^*_{\mathfrak{e},\mathfrak{e}'}}/{r^*_{\mathfrak{e},\mathfrak{e}'}})\right)}-\ln{\llbracket\mathfrak{e}'\rrbracket}\right) \\
&\ge & \ln {r^*_{\mathfrak{e},\mathfrak{e}'}}+\ln {\llbracket\mathfrak{e}'\rrbracket}-\mathbf{1}_{b^*_{\mathfrak{e},\mathfrak{e}'}< 0}\cdot {\left(t_{\mathfrak{e}'}+({b^*_{\mathfrak{e},\mathfrak{e}'}}/{r^*_{\mathfrak{e},\mathfrak{e}'}})\right)}^{-1}\cdot\left(-{b^*_{\mathfrak{e},\mathfrak{e}'}}/{r^*_{\mathfrak{e},\mathfrak{e}'}}\right),
\end{eqnarray*}
where the last step is obtained from LMVT (Theorem~\ref{thm:lagrange}) and by distinguishing whether $b^*_{\mathfrak{e},\mathfrak{e}'}\ge 0$ or not, using the fact that the derivative of the natural-logarithm is the reciprocal function.
Note that one has
$t_{\mathfrak{e}'}+\frac{b^*_{\mathfrak{e},\mathfrak{e}'}}{r^*_{\mathfrak{e},\mathfrak{e}'}}\ge 1$
due to the maximal choice of $t_{\mathfrak{e}'}$.

\item {\em Mutual No-Greater-Than Inequalities over $u_{\mathfrak{e}}$'s.}
For each pair $(\mathfrak{e},\mathfrak{e}')\in\mathcal{E}'\times\mathcal{E}'$ such that
$\mathfrak{e}\ne\mathfrak{e}'$ and $\mathsf{r}^*_{\mathfrak{e},\mathfrak{e}'},\mathsf{b}^*_{\mathfrak{e},\mathfrak{e}'}$ are successfully found and $\mathsf{r}^*_{\mathfrak{e},\mathfrak{e}'}>0$, the algorithm adds
\[
u_{\mathfrak{e}'}+\ln{\mathsf{r}^*_{\mathfrak{e},\mathfrak{e}}}-u_{\mathfrak{e}}+\mathbf{1}_{\mathsf{b}^*_{\mathfrak{e},\mathfrak{e}'}\ge 0}\cdot t_{\mathfrak{e}'}^{-1}\cdot\frac{b^*_{\mathfrak{e},\mathfrak{e}'}}{r^*_{\mathfrak{e},\mathfrak{e}'}}
\]
into $\Gamma$.
This is because
$\left( \mathsf{r}^*_{\mathfrak{e},\mathfrak{e}'}\cdot \llbracket\mathfrak{e}'\rrbracket +\mathsf{b}^*_{\mathfrak{e},\mathfrak{e}'}\right)-\llbracket\mathfrak{e}\rrbracket\ge 0$
implies
\begin{eqnarray*}
\ln{\llbracket\mathfrak{e}\rrbracket} &\le & \ln {\mathsf{r}^*_{\mathfrak{e},\mathfrak{e}'}}+\ln \left(\llbracket\mathfrak{e}'\rrbracket+\frac{\mathsf{b}^*_{\mathfrak{e},\mathfrak{e}'}}{\mathsf{r}^*_{\mathfrak{e},\mathfrak{e}'}}\right) \\
&= & \ln{\mathsf{r}^*_{\mathfrak{e},\mathfrak{e}'}}+\ln{\llbracket\mathfrak{e}'\rrbracket}+\left(\ln{ \left(\llbracket\mathfrak{e}'\rrbracket+\frac{\mathsf{b}^*_{\mathfrak{e},\mathfrak{e}'}}{\mathsf{r}^*_{\mathfrak{e},\mathfrak{e}'}}\right)}-\ln{\llbracket\mathfrak{e}'\rrbracket}\right) \\
&\le & \ln {\mathsf{r}^*_{\mathfrak{e},\mathfrak{e}'}}+\ln {\llbracket\mathfrak{e}'\rrbracket}+\mathbf{1}_{\mathsf{b}^*_{\mathfrak{e},\mathfrak{e}'}\ge 0}\cdot t_{\mathfrak{e}'}^{-1}\cdot\frac{\mathsf{b}^*_{\mathfrak{e},\mathfrak{e}'}}{\mathsf{r}^*_{\mathfrak{e},\mathfrak{e}'}},
\end{eqnarray*}
where the last step is obtained from Lagrange's Mean-Value Theorem and by distinguishing whether $\mathsf{b}^*_{\mathfrak{e},\mathfrak{e}'}\ge 0$ or not.
Note that one has
\[
t_{\mathfrak{e}'}+\frac{\mathsf{b}^*_{\mathfrak{e},\mathfrak{e}'}}{\mathsf{r}^*_{\mathfrak{e},\mathfrak{e}'}}\ge 1
\]
due to the maximal choice of $t_{\mathfrak{e}'}$ and the fact that $\widetilde{\mathfrak{e}}$ (as a polynomial function) is everywhere greater than or equal to $1$ under $\Sat(\Gamma)$ (cf. (\S)).

\item {\em Mutual No-Smaller-Than Inequalities over $v_{\mathfrak{e}}$'s.} For each pair of variables of the form $(v_{\mathfrak{e}},v_{\mathfrak{e}'})$ such that $\mathfrak{e}\ne\mathfrak{e}'$, $r^*_{\mathfrak{e},\mathfrak{e}'},b^*_{\mathfrak{e},\mathfrak{e}'}$ are successfully found and $r^*_{\mathfrak{e},\mathfrak{e}'}>0, b^*_{\mathfrak{e},\mathfrak{e}'}\ge 0$,  the algorithm adds
\[
v_{\mathfrak{e}}-{\left(r^*_{\mathfrak{e},\mathfrak{e}'}\right)}^r\cdot\left(v_{\mathfrak{e}'}+r\cdot \frac{b^*_{\mathfrak{e},\mathfrak{e}'}}{r^*_{\mathfrak{e},\mathfrak{e}'}}\cdot v'_{\mathfrak{e}'}\right)
\]
into $\Gamma$.
This is due to the fact that $\llbracket\mathfrak{e}\rrbracket-\left( r^*_{\mathfrak{e},\mathfrak{e}'}\cdot \llbracket\mathfrak{e}'\rrbracket +b^*_{\mathfrak{e},\mathfrak{e}'}\right)\ge 0$
implies
\begin{eqnarray*}
{\llbracket\mathfrak{e}\rrbracket}^r &\ge & \left(r^*_{\mathfrak{e},\mathfrak{e}'}\right)^r\cdot {\left(\llbracket\mathfrak{e}'\rrbracket+\frac{b^*_{\mathfrak{e},\mathfrak{e}'}}{r^*_{\mathfrak{e},\mathfrak{e}'}}\right)}^r \\
&\ge & \left(r^*_{\mathfrak{e},\mathfrak{e}'}\right)^r\cdot\left({\llbracket\mathfrak{e}'\rrbracket}^r +{\left({\llbracket\mathfrak{e}'\rrbracket}+\frac{b^*_{\mathfrak{e},\mathfrak{e}'}}{r^*_{\mathfrak{e},\mathfrak{e}'}}\right)}^r-{\llbracket\mathfrak{e}'\rrbracket}^r\right) \\
&\ge &  {\left(r^*_{\mathfrak{e},\mathfrak{e}'}\right)}^r\cdot\left({\llbracket\mathfrak{e}'\rrbracket}^r +r\cdot{\llbracket\mathfrak{e}'\rrbracket}^{r-1}\cdot \frac{b^*_{\mathfrak{e},\mathfrak{e}'}}{r^*_{\mathfrak{e},\mathfrak{e}'}}\right)\enskip.
\end{eqnarray*}
where the last step is obtained from Lagrange's Mean-Value Theorem.

\item {\em Mutual No-Greater-Than Inequalities over $v_{\mathfrak{e}}$'s.} For each pair of variables of the form $(v_{\mathfrak{e}},v_{\mathfrak{e}'})$ such that $\mathfrak{e}\ne\mathfrak{e}'$, $\mathsf{r}^*_{\mathfrak{e},\mathfrak{e}'},\mathsf{b}^*_{\mathfrak{e},\mathfrak{e}'}$ are successfully found and $\mathsf{r}^*_{\mathfrak{e},\mathfrak{e}'}>0, \mathsf{b}^*_{\mathfrak{e},\mathfrak{e}'}\ge 0$,  the algorithm adds
\begin{align*}
{\left(\mathsf{r}^*_{\mathfrak{e},\mathfrak{e}'}\right)}^r\cdot\left(v_{\mathfrak{e}'}+
\left(\mathbf{1}_{\mathsf{b}^*_{\mathfrak{e},\mathfrak{e}'}\le 0} +\mathbf{1}_{\mathsf{b}^*_{\mathfrak{e},\mathfrak{e}'}> 0}\cdot M^{r-1}\right)\cdot r\cdot \frac{\mathsf{b}^*_{\mathfrak{e},\mathfrak{e}'}}{\mathsf{r}^*_{\mathfrak{e},\mathfrak{e}'}}\cdot v'_{\mathfrak{e}'}\right)-v_{\mathfrak{e}}
\end{align*}
into $\Gamma$, where $M:=\frac{\mathsf{b}^*_{\mathfrak{e},\mathfrak{e}'}}{\mathsf{r}^*_{\mathfrak{e},\mathfrak{e}'}\cdot t_{\mathfrak{e}'}}+1$.
This is due to the fact that $\left(\mathsf{r}^*_{\mathfrak{e},\mathfrak{e}'}\cdot \llbracket\mathfrak{e}'\rrbracket +\mathsf{b}^*_{\mathfrak{e},\mathfrak{e}'}\right)-\llbracket\mathfrak{e}\rrbracket\ge 0$
implies
\begin{eqnarray*}
{\llbracket\mathfrak{e}\rrbracket}^r &\le & \left(\mathsf{r}^*_{\mathfrak{e},\mathfrak{e}'}\right)^r\cdot {\left(\llbracket\mathfrak{e}'\rrbracket+\frac{\mathsf{b}^*_{\mathfrak{e},\mathfrak{e}'}}{\mathsf{r}^*_{\mathfrak{e},\mathfrak{e}'}}\right)}^r \\
&\le & \left(\mathsf{r}^*_{\mathfrak{e},\mathfrak{e}'}\right)^r\cdot\left({\llbracket\mathfrak{e}'\rrbracket}^r +{\left({\llbracket\mathfrak{e}'\rrbracket}+\frac{\mathsf{b}^*_{\mathfrak{e},\mathfrak{e}'}}{\mathsf{r}^*_{\mathfrak{e},\mathfrak{e}'}}\right)}^r-{\llbracket\mathfrak{e}'\rrbracket}^r\right) \\
&\le &  {\left(\mathsf{r}^*_{\mathfrak{e},\mathfrak{e}'}\right)}^r\cdot\bigg({\llbracket\mathfrak{e}'\rrbracket}^r +\left(\mathbf{1}_{\mathsf{b}^*_{\mathfrak{e},\mathfrak{e}'}\le 0} +\mathbf{1}_{\mathsf{b}^*_{\mathfrak{e},\mathfrak{e}'}> 0}\cdot M^{r-1}\right)\cdot r\cdot \frac{\mathsf{b}^*_{\mathfrak{e},\mathfrak{e}'}}{\mathsf{r}^*_{\mathfrak{e},\mathfrak{e}'}}\cdot {\llbracket\mathfrak{e}'\rrbracket}^{r-1} \bigg)
\end{eqnarray*}
where the last step is obtained from Lagrange's Mean-Value Theorem and the fact that $\llbracket\mathfrak{e}'\rrbracket\ge t_{\mathfrak{e}'}$ implies $\llbracket\mathfrak{e}'\rrbracket+\frac{\mathsf{b}^*_{\mathfrak{e},\mathfrak{e}'}}{\mathsf{r}^*_{\mathfrak{e},\mathfrak{e}'}}\le M\cdot \llbracket\mathfrak{e}'\rrbracket$.
\end{compactenum}
Although in Item~4 and Item~6 above, we have logarithmic terms such as
$\ln t_{\mathfrak{e}}$ and $\ln{r^*_{\mathfrak{e},\mathfrak{e}'}}$,
both $t_{\mathfrak{e}}$ and ${r^*_{\mathfrak{e},\mathfrak{e}'}}$ are already
determined constants, hence their approximations can be used.
After Step 4, the constraint triple $(\fn{f}, \phi,\mathfrak{e}^*)$ is transformed into $(\fn{f}, \bigwedge_{h\in\Gamma}h\ge 0,\widetilde{\mathfrak{e}^*})$.

\begin{example}\label{ex:new:step4b}
We continue with Example~\ref{ex:new:step4a}.
In Step 4(b), the algorithm establishes the following non-trivial inequalities:
\begin{compactitem}
\item ({\em From Item 2,3 in Step 4(b) for $\widetilde{\mathfrak{q}_3}$}) $w\ge 0.5\cdot n-0.5, w\le 0.5\cdot n$ and $n\ge 2\cdot w, n\le 2\cdot w+1$;
\item ({\em From Item 4 in Step 4(b) for $\widetilde{\mathfrak{q}_1},\widetilde{\mathfrak{q}_2}$}) $n-e\cdot u\ge 0$ and $u\ge 0$;
\item ({\em From Item 4 in Step 4(b) for $\widetilde{\mathfrak{q}_3}$}) $n-e\cdot u\ge 0, u-\ln{2}\ge 0$ and $w-e\cdot v\ge 0, v\ge 0$;
\item ({\em From Item 6,7 in Step 4(b) for $\widetilde{\mathfrak{q}_3}$}) $u-v-\ln{2}\ge 0$ and $v-u+\ln{2}+\frac{1}{2}\ge 0$.
\end{compactitem}
After Step 4(b), $\Gamma_i$'s ($1\le i\le 3$) are updated as follows:
\begin{compactitem}
\item $\Gamma_1=\{n-1, n-e\cdot u, u\}$ and $\Gamma_2=\{n-1, 1-n, n-e\cdot u, u\}$;
\item $\Gamma_3=\{n-2,n-2\cdot w,2\cdot w-n+1,w-1,n-e\cdot u, u-\ln{2}, w-e\cdot v, v, u-v-\ln{2}, v-u+\ln{2}+\frac{1}{2}\}$.\qed
\end{compactitem}
\end{example}

\begin{remark}\label{remark:step4}
The key difficulty 
is to handle logarithmic and exponentiation
terms.
In Step~4(a) we abstract such terms with fresh variables and perform sound approximation
of floored expressions.
In Step~4(b) we use Farkas' Lemma and LMVT to soundly transform
logarithmic or exponentiation terms to polynomials.
\qed
\end{remark}

\begin{remark}
The aim of Step 4(b) is to approximate logarithmic and exponentiation terms by linear inequalities they satisfy.
In the final step, those linear inequalities suffice to solve our problem.
As to extensibility, Step 4(b) may be extended to other non-polynomial terms by constructing similar linear inequalities they satisfy.\qed
\end{remark}

\vspace{-1.5em}
\subsection{Step~5 of \synalgo}
\vspace{-0.8em}

This step is to solve the template variables in the template established in Step 1, based on the sets $\Gamma$ computed in Step 4. For this step, we use Handelman's Theorem.

\begin{definition}[Monoid]
Let $\Gamma$ be a finite subset of some polynomial ring $\POLYS$ such that all elements of $\Gamma$ are polynomials of degree $1$.
The \emph{monoid} of $\Gamma$ is defined by:
$\mbox{\sl Monoid}(\Gamma):=\left\{\prod_{i=1}^k h_i \mid k\in\Nset_0\mbox{ and }h_1,\dots,h_k\in\Gamma\right\}$~~.
\end{definition}

\begin{theorem}[Handelman's Theorem~\cite{HandelmanTheorem}]
\label{thm:handelman}
Let $\POLYS$ be the polynomial ring with variables $x_1,\dots, x_m$ (for $m\ge 1$).
Let $g\in\POLYS$ and $\Gamma$ be a finite subset of $\POLYS$ such that all elements of $\Gamma$ are polynomials of degree $1$.
If (i) the set $\Sat(\Gamma)$ is compact and non-empty and (ii) $g(\vec{x})>0$ for all $\vec{x}\in \Sat(\Gamma)$, then
\begin{equation}\label{eq:handelman}
\textstyle g=\sum_{i=1}^n c_i\cdot u_i
\end{equation}
for some $n\in\Nset$, non-negative real numbers $c_1,\dots,c_n\ge 0$ and $u_1,\dots,u_n\in\mbox{\sl Monoid}(\Gamma)$.
\end{theorem}

Basically, Handelman's Theorem gives a characterization of positive polynomials over  polytopes.
In this paper, we concentrate on Eq.~(\ref{eq:handelman}) which provides a sound form for a non-negative polynomial over a general (i.e. possibly unbounded) polyhedron.
The following proposition shows that Eq.~(\ref{eq:handelman}) encompasses a simple proof system for non-negative polynomials over polyhedra.

\begin{proposition}\label{prop:handelman}
Let $\Gamma$ be a finite subset of some polynomial ring $\POLYS$ such that all elements of $\Gamma$ are polynomials of degree $1$.
Let the collection of deduction systems $\{\vdash_{k}\}_{k\in\Nset}$ be generate by the following rules:
\begin{eqnarray*}
\cfrac{h\in\Gamma}{\vdash_1 h\ge 0}\qquad\cfrac{c\in\Rset, c\ge 0}{\vdash_1 c\ge 0}\qquad\cfrac{\vdash_k h\ge 0, c\in\Rset, c\ge 0}{\vdash_k c\cdot h\ge 0} \\
\cfrac{\vdash_k h_1\ge 0, \vdash_k h_2\ge 0}{\vdash_k h_1+h_2\ge 0}\qquad  \cfrac{\vdash_{k_1} h_1\ge 0, \vdash_{k_2} h_2\ge 0}{\vdash_{k_1+k_2} h_1\cdot h_2\ge 0}\enskip.
\end{eqnarray*}
Then for all $k\in\Nset$ and polynomials $g\in\POLYS$, if $\vdash_k g\ge 0$ then
$g=\sum_{i=1}^n c_i\cdot u_i$ for some $n\in\Nset$, non-negative real numbers $c_1,\dots,c_n\ge 0$ and $u_1,\dots,u_n\in\mbox{\sl Monoid}(\Gamma)$ such that every $u_i$ is a product of no more than $k$ polynomials in $\Gamma$.
\end{proposition}
\begin{proof}
By an easy induction on $k$.\qed
\end{proof}

\smallskip\noindent{\bf Step~5: Solving Unknown Coefficients in the Template.} Now we use the input parameter $k$ as the maximal number of multiplicands in each summand at the right-hand-side of Eq.~(\ref{eq:handelman}).
For any constraint triple $(\fn{f},\phi,\mathfrak{e}^*)$ which is generated in Step 3 and passes the emptiness checking in Item 7 of Step 4(a), the algorithm performs the following steps.
\begin{compactenum}
\item {\em Preparation for Eq.~(\ref{eq:handelman}).} The algorithm reads the set $\Gamma$ for $(\fn{f},\phi,\mathfrak{e}^*)$ computed in Step 4, and computes $\widetilde{\mathfrak{e}^*}$ from Item 4 of Step 4(a).
\item {\em Application of Handelman's Theorem.} First, the algorithm establishes a fresh coefficient variable $\lambda_h$ for each polynomial $h$ in $\mbox{Monoid}(\Gamma)$ with no more than $k$ multiplicands from $\Gamma$.
    Then, the algorithm establishes linear equalities over coefficient variables $\lambda_h$'s and template variables in the template $\eta$ established in Step 1 by
equating coefficients of the same monomials at the left- and right-hand-side of the following polynomial equality
$\widetilde{\mathfrak{\mathfrak{e}^*}}=\sum_{h} \lambda_h\cdot h$\enskip.
Second, the algorithm incorporates all constraints of the form $\lambda_h\ge 0$.
\end{compactenum}
Then the algorithm collects all linear equalities and inequalities established in Item 2 above conjunctively as a single system of linear inequalities and solves it through linear-programming algorithms; if no feasible solution exists, the algorithm fails without output, otherwise the algorithm outputs the function $\widehat{\llbracket\eta\rrbracket}$ where all template variables in the template $\eta$ are resolved
by their values in the solution.
We now state the soundness of our approach for synthesis of measure functions (proof in Appendix~\ref{app:synalgo}).

\begin{theorem}\label{thm:mainthmsynth}
Our algorithm, \synalgo, is a sound
approach for the \recterm\ problem, i.e., if \synalgo\ succeeds to
synthesize a function $g$ on $\left\{(\fn{f},\loc,\nu)\mid \fn{f}\in\fnames, \loc\in\slocs{f},  \nu\in\Aval{f}\right\}$,
then $\widehat{g}$ is a measure function and hence an upper bound on
the termination-time function.
\end{theorem}

\begin{remark}\label{remark:step5}
While Step~4 transforms logarithmic and exponentiation terms to polynomials,
we need in Step~5 a sound method to solve polynomials with linear programming.
We achieve this with Handelman's Theorem.
\qed
\end{remark}

\begin{example}\label{ex:new:step5}
Continue with Example~\ref{ex:new:step4b}.
In the final step (Step 5), the unknown coefficients $c_i$'s ($1\le i\le 3$) are to be resolved through (\ref{eq:handelman}) so that logical formulae encoded by $\widetilde{\mathfrak{q}_i}$'s are valid (w.r.t updated $\Gamma_i$'s).
Since to present the whole technical detail would be too cumbersome, we present directly a feasible solution for $c_i$'s and how they fulfill (\ref{eq:handelman}).
Below we choose the solution that $c_1=0$, $c_2=\frac{2}{\ln{2}}$ and $c_3=2$.
Then we have that
\begin{compactitem}
\item (From $\widetilde{\mathfrak{q}_1}$) $c_2\cdot u+c_3=\lambda_1\cdot u+ \lambda_2$ where $\lambda_1:=\frac{2}{\ln{2}}$ and $\lambda_2:=2$;
\item (From $\widetilde{\mathfrak{q}_2}$) $c_2\cdot u+c_3-2=\lambda_1\cdot u$;
\item (From $\widetilde{\mathfrak{q}_3}$) $c_2\cdot (u-v)-2=\lambda_1\cdot (u-v-\ln{2})$.
\end{compactitem}
Hence by Theorem~\ref{thm:mfunc}, $\overline{T}(\mathsf{f}, 1, n)\le \eta(\mathsf{f},1,n)=\frac{2}{\ln{2}}\cdot \ln{n}+2$.
It follows that {\sc Binary-Search} runs in $\mathcal{O}(\log{n})$ in worst-case. \qed
\end{example}

\begin{remark}\label{rem:novelty}
We remark two aspects of our algorithm.
\begin{compactenum}
\item {\em Scalability.}
Our algorithm only requires solving linear inequalities.
Since linear-programming solvers have been widely studied and experimented,
the scalability of our approach directly depends on the linear-programming
solvers.
Hence the approach we present is a relatively scalable one.

\item {\em Novelty.} A key novelty of our approach is to obtain 
non-polynomial bounds (such as $\mathcal{O}(n \log n)$ , $\mathcal{O}(n^r)$, where
$r$ is not integral) through linear programming.
The novel technical steps are: (a)~use of abstraction variables;
(b)~use of LMVT and Farkas' lemma to obtain sound
linear constraints over abstracted variables; and
(c)~use of Handelman's Theorem to solve the unknown coefficients in
polynomial time.
\qed
\end{compactenum}
\end{remark}


\vspace{-2em}
\section{Experimental Results} \label{sec:experimental-results}
\vspace{-1em}
For worst-case upper bounds of non-trivial form, we consider four classical examples.

\smallskip\noindent{\em Merge-Sort and Closest-Pair.}
We consider the classical Merge-Sort problem~\cite[Chapter~2]{DBLP:books/daglib/0023376} and the closest pair
problem consider a set of $n$ two-dimensional points and asks for the pair of points that have shortest
Euclidean distance between them (cf.~\cite[Chapter~33]{DBLP:books/daglib/0023376})
(see Appendix~\ref{app:experiments} for the pseudo-code).
For both problems we obtain an $\mathcal{O}(n \log n)$ bound.

\smallskip\noindent{\em Strassen's Algorithm.}
We consider one of the classic sub-cubic algorithm for Matrix multiplication.
The Strassen algorithm (cf.~\cite[Chapter 4]{DBLP:books/daglib/0023376}) has a worst-case running time of $n^{\log_27}$.
We present the pseudo-code of Strassen's algorithm in our programming language
in Appendix~\ref{app:experiments}.
Using a template with $n^{2.9}$, our algorithm synthesizes a measure function.

\smallskip\noindent{\em Karatsuba's Algorithm.}
We consider two polynomials
$p_1=a_0 + a_1 x + a_2 x^2 + \ldots + a_n x^{n-1}$
and $p_2=b_0 + b_1 x + b_2 x^2 + \ldots + b_n x^{n-1}$,
where the coefficients $a_i$'s and $b_i$'s are represented as arrays.
The computational problem asks to compute the coefficients of the polynomial
obtained by multiplication of $p_1$ and $p_2$, and considers that $n$ is a power of~2.
While the most naive algorithm is quadratic, Karatsuba's algorithm (cf.~\cite{KnuthAllBooks})
is a classical sub-quadratic algorithm for the problem with running
time $n^{\log_23}$.
We present the pseudo-code Karatsuba's algorithm in our programming language in Appendix~\ref{app:experiments}.
Using a template with $n^{1.6}$, our algorithm synthesizes a measure function
(basically, using constraints as illustrated in Example~\ref{ex:narrationb}).

The above four examples show that our sound approach can synthesize
non-trivial worst-case complexity bounds for several classical algorithms.

\noindent{\em Invariants.}
In the experiments we derive simple invariants from the programs directly from the prerequisites of procedures and guards of while-loops. Alternatively, they can be derived automatically using~\cite{DBLP:conf/cav/ColonSS03}.

\smallskip\noindent{\em Results.}
Below we present experimental results on the examples explained above.
We implement our algorithm that basically generates a set of linear constraints,
where we use lp{\_}solve~\cite{lpsolve2016} for solving linear programs.
Our experimental results are presented in Table~\ref{tbl:experimentalresults}, where all numbers are rounded to $10^{-2}$ and $n$ represents
the input length.
All results were obtained
on an Intel i3-4130 CPU 3.4 GHz 8 GB of RAM.

\vspace{-1em}
\begin{table}
\caption{Experimental results where $\eta(\loc_0)$ is the part of measure function at the initial label.}
\label{tbl:experimentalresults}
\centering
\begin{tabular}{|c|c|c|}
\hline
Example &  Time (in Seconds) & $\eta(\loc_0)$  \\
\hline
Merge-Sort & {6} & $25.02 \cdot n\cdot \ln{n} + 21.68\cdot n -20.68$    \\
\hline
Closest-Pair & {11} &  $128.85\cdot n \cdot\ln{n} + 108.95\cdot n - 53.31$\\
\hline
Karatsuba   & {3} & $2261.55\cdot n^{1.6} + 1$   \\
\hline
Strassen    & {7} & $954.2\cdot n^{2.9} + 1$  \\
\hline
\end{tabular}
\end{table}
\vspace{-1em}

\section{Related Work}
\vspace{-1em}
In this section we discuss the related work.
The termination of recursive programs or other temporal properties has already been extensively
studied~\cite{DBLP:conf/pldi/CookPR06,DBLP:conf/esop/KuwaharaTU014,DBLP:conf/tacas/CookSZ13,DBLP:conf/sas/Urban13,DBLP:conf/popl/CousotC12,DBLP:conf/popl/LeeJB01,DBLP:journals/toplas/Lee09,DBLP:journals/fmsd/CookPR09,DBLP:conf/vmcai/AlurC10}.
Our work is most closely related to automatic amortized analysis~\cite{DBLP:conf/aplas/HoffmannH10,DBLP:conf/esop/HoffmannH10,DBLP:conf/popl/HofmannJ03,DBLP:conf/esop/HofmannJ06,DBLP:conf/csl/HofmannR09,DBLP:conf/fm/JostLHSH09,DBLP:conf/popl/JostHLH10,Hoffman1,DBLP:conf/popl/GimenezM16,DBLP:conf/cav/SinnZV14,DBLP:conf/sas/AliasDFG10}, as well as the SPEED project~\cite{SPEED1,SPEED2,DBLP:conf/cav/GulavaniG08}.
There are two key differences of our methods as compared to previous works.
First, our methods are based on extension of ranking functions to non-deterministic recursive
programs, whereas previous works either use potential functions, abstract interpretation, or size-change.
Second, while none of the previous methods can derive non-polynomial bounds
such as $\mathcal{O}(n^r)$, where $r$ is not an integer, our approach gives
an algorithm to derive such non-polynomial bounds, and surprisingly using
linear programming.

The approach of recurrence relations for worst-case analysis is  explored in~\cite{DBLP:conf/icfp/Grobauer01,DBLP:journals/tcs/FlajoletSZ91,DBLP:journals/entcs/AlbertAGGPRRZ09,DBLP:conf/sas/AlbertAGP08,DBLP:conf/esop/AlbertAGPZ07}.
A related result is by Albert~\emph{et al.}~\cite{DBLP:conf/sas/AlbertAGP08} who considered using evaluation trees for solving recurrence relations, which can derive the worst-case bound for Merge-Sort.
Another approach through theorem proving is explored in~\cite{DBLP:conf/popl/SrikanthSH17}.
The approach is to iteratively generate control-flow paths and then to obtain worst-case bounds over generated paths through theorem proving (with arithmetic theorems).

Ranking functions for intra-procedural analysis has been widely studied~\cite{BG05,DBLP:conf/cav/BradleyMS05,DBLP:conf/tacas/ColonS01,DBLP:conf/vmcai/PodelskiR04,DBLP:conf/pods/SohnG91,DBLP:conf/vmcai/Cousot05,DBLP:journals/fcsc/YangZZX10,DBLP:journals/jossac/ShenWYZ13}.
Most works have focused on linear or polynomial ranking functions~\cite{DBLP:conf/tacas/ColonS01,DBLP:conf/vmcai/PodelskiR04,DBLP:conf/pods/SohnG91,DBLP:conf/vmcai/Cousot05,DBLP:journals/fcsc/YangZZX10,DBLP:journals/jossac/ShenWYZ13}.
Such approach alone can only derive polynomial bounds for
programs.
When integrated with evaluation trees, polynomial ranking functions can derive exponential bounds such as $\mathcal{O}(2^n)$~\cite{DBLP:journals/toplas/BrockschmidtE0F16}.
In contrast, we directly synthesize non-polynomial ranking functions without the help of evaluation trees.
Ranking functions have been extended to ranking
supermartingales~\cite{SriramCAV,HolgerPOPL,DBLP:conf/popl/ChatterjeeFNH16,DBLP:conf/cav/ChatterjeeFG16,DBLP:journals/corr/ChatterjeeF17,ChatterjeeNZ2017} 
for probabilistic programs without recursion.
These works cannot derive non-polynomial bounds. 

Several other works present proof rules for deterministic programs~\cite{DBLP:journals/fac/Hesselink94}
as well as for probabilistic programs~\cite{JonesPhdThesis,OLKMLICS2016}.
None of these works can be automated.
Other related approaches are sized types~\cite{DBLP:journals/lisp/ChinK01,DBLP:conf/icfp/HughesP99,DBLP:conf/popl/HughesPS96},
and polynomial resource bounds~\cite{DBLP:conf/tlca/ShkaravskaKE07}.
Again none of these approaches can yield bounds like $\mathcal{O}(n \log n)$
or $\mathcal{O}(n^r)$, for $r$ non-integral.

Below we compare three most related works  ~\cite{DBLP:conf/sas/AlbertAGP08,DBLP:journals/toplas/BrockschmidtE0F16,DBLP:conf/cav/GulavaniG08}.

\noindent{\bf Comparison with \cite{DBLP:conf/sas/AlbertAGP08}.} First, \cite{DBLP:conf/sas/AlbertAGP08} uses synthesis of linear ranking functions in order to bound the number of nodes or the height of an evaluation tree for a cost relation system. We use expression abstraction to synthesize non-linear bounds.
Second, \cite{DBLP:conf/sas/AlbertAGP08} uses branching factor of a cost relation system to bound the number of nodes, which typically leads to exponential bounds. Our approach produces efficient bounds.
Third, \cite{DBLP:conf/sas/AlbertAGP08} treats Merge-sort in a very specific way: first, there is a ranking function with a discount-factor 2 to bound the height of an evaluation tree logarithmically, then there is a linear bound for levels in an evaluation tree which does not increase when the levels become deeper, and multiplying them together produces $\mathcal{O}(n \log n)$ bound. We do not rely on these heuristics.
Finally, \cite{DBLP:conf/sas/AlbertAGP08} is not applicable to non-direct-recursive cost relations, while our approach is applicable to all recursive programs.

\noindent{\bf Comparison with \cite{DBLP:journals/toplas/BrockschmidtE0F16}.}
To derive non-polynomial bounds, \cite{DBLP:journals/toplas/BrockschmidtE0F16} integrate polynomial ranking functions with evaluation trees for recursive programs so that exponential bounds such as $\mathcal{O}(2^n)$ can be derived for Fibonacci numbers.
In contrast, we directly synthesize non-polynomial ranking functions without the help of evaluation trees.

\noindent{\bf Comparison with \cite{DBLP:conf/cav/GulavaniG08}.} \cite{DBLP:conf/cav/GulavaniG08} generates bounds through abstract interpretation using inference systems over expression abstraction with logarithm, maximum, exponentiation, etc.
In contrast, we employ different method through ranking functions, also with linear-inequality system over expression abstraction with logarithm and exponentiation. The key difference w.r.t expression abstraction is that \cite{DBLP:conf/cav/GulavaniG08} handles in extra maximum and square root, while we consider in extra floored expressions and finer linear inequalities between e.g., $\log n,\log(n+1)$ or $n^{1.6},(n+1)^{1.6}$ through Lagrange's Mean-Value Theorem.

\vspace{-1.5em}
\section{Conclusion}
\vspace{-1em}
In this paper, we developed an approach to obtain non-polynomial
worst-case bounds for recursive programs through
(i) abstraction of logarithmic and exponentiation terms and (ii) Farkas' Lemma, LMVT, and Handelman's Theorem.
Moreover our approach obtains such bounds using linear programming, which thus
is an efficient approach.
Our approach obtains non-trivial worst-case complexity bounds for classical recursive programs:
$\mathcal{O}(n\log{n})$-complexity for both Merge-Sort
and the divide-and-conquer Closest-Pair algorithm,
$\mathcal{O}(n^{1.6})$ for Karatsuba's algorithm for polynomial multiplication,
and $\mathcal{O}(n^{2.9})$ for Strassen's algorithm for matrix multiplication.
The bounds we obtain for Karatsuba's and Strassen's algorithm are close to the optimal
bounds known.
An interesting future direction is to extend our technique to data-structures.
Another future direction is to investigate the application of our approach to invariant generation.

\subsubsection*{Acknowledgements}
We thank all reviewers for valuable comments.
The research is partially supported by Vienna Science and Technology Fund (WWTF) ICT15-003,
Austrian Science Fund (FWF) NFN Grant No. S11407-N23 (RiSE/SHiNE), ERC Start grant (279307: Graph Games), the Natural Science Foundation of China (NSFC) under Grant No. 61532019 and the CDZ project CAP (GZ 1023).

\clearpage
{\scriptsize
\bibliographystyle{splncs03}
\bibliography{PL}
}

\clearpage
\appendix

\section{Evaluation of Arithmetic Expressions}\label{app:arith}

Below we fix a countable set $\mathcal{X}$ of scalar variables.
Given an arithmetic expression $\mathfrak{e}$ over $\mathcal{X}$ and a valuation on $\mathcal{X}$, the element $\mathfrak{e}(\nu)$ is defined inductively on the structure of $\mathfrak{e}$ as follows:
\begin{compactitem}
\item $c(\nu):=c$;
\item $x(\nu):=\nu(x)$;
\item ${\left\lfloor \frac{\mathfrak{e}}{c}\right\rfloor}(\nu):={\left\lfloor \frac{\mathfrak{e}(\nu)}{c}\right\rfloor}$; (Note that $c\ne 0$ by our assumption.)
\item ${(\mathfrak{e}+\mathfrak{e}')}(\nu):=\mathfrak{e}(\nu)+\mathfrak{e}'(\nu)$;
\item ${(\mathfrak{e}-\mathfrak{e}')}(\nu):=\mathfrak{e}(\nu)-\mathfrak{e}'(\nu)$;
\item ${(c*\mathfrak{e}')}(\nu):=c\cdot\mathfrak{e}'(\nu)$.
\end{compactitem}

\section{Semantics of Propositional Arithmetic Predicates}\label{app:predicate}

Let $\mathcal{X}$ be the set of scalar variables.
The satisfaction relation $\models$ between valuations $\nu$ and propositional arithmetic predicates $\phi$ is defined inductively as follows:
\begin{compactitem}
\item $\nu\models \mathfrak{e}\Join\mathfrak{e}'$ ($\Join\in\{\le,\ge\}$) if $\mathfrak{e}(\nu)\Join\mathfrak{e}'(\nu)$;
\item $\nu\models\neg\phi$ iff $\nu\not\models\phi$;
\item $\nu\models\phi_1\wedge\phi_2$ iff $\nu\models\phi_1$ and $\nu\models\phi_2$;
\item $\nu\models\phi_1\vee\phi_2$ iff $\nu\models\phi_1$ or $\nu\models\phi_2$.
\end{compactitem}

\section{Detailed Syntax for Recursive Programs}\label{app:detailedsyntax}

In the sequel, we fix a countable set of \emph{scalar variables}; and we also fix a countable set of \emph{function names}.
W.l.o.g, these two sets are pairwise disjoint.
Each scalar variable holds an integer upon instantiation.

\smallskip\noindent{\bf The Syntax.} The syntax of our recursive programs is illustrated by the grammar in Fig.~\ref{fig:syntax}.
Below we briefly explain the grammar.
\begin{compactitem}
\item \emph{Variables:} Expressions $\langle\mathit{pvar}\rangle$ range over scalar variables.
\item \emph{Function Names:} Expressions $\langle\mathit{fname}\rangle$ range over function names.
\item \emph{Constants:} Expressions $\langle\mathit{int}\rangle$ range over integers represented as decimal numbers, while expressions $\langle\mathit{nonzero}\rangle$ range over non-zero integers represented as decimal numbers.
\item \emph{Arithmetic Expressions:} Expressions $\langle\mathit{expr}\rangle$ range over linear arithmetic expressions consisting of scalar variables, floor operation (cf. $\left\lfloor\langle\mathit{expr}\rangle\slash\langle\mathit{nonzero}\rangle\right\rfloor$) and arithmetic operations.
\item \emph{Parameters:} Expressions $\langle\mathit{plist}\rangle$ range over lists of scalar variables, and expressions $\langle\mathit{vlist}\rangle$ range over lists of $\langle\mathit{expr}\rangle$ expressions.
\item \emph{Boolean Expressions:} Expressions $\langle\mathit{bexpr}\rangle$ range over propositional arithmetic predicates over scalar variables.
\item \emph{Statements:} Various types of assignment statements are indicated by `$:=$'; `\textbf{skip}' is the statement that does nothing;
conditional branch or demonic non-determinism is indicated by the keyword `\textbf{if}', while $\langle\mathit{bexpr}\rangle$ indicates conditional branch and $\star$ indicates demonic non-determinism;
while-loops are indicated by the keyword `\textbf{while}';
sequential compositions are indicated by semicolon;
finally, function calls are indicated by $\langle\mathit{fname}\rangle\left(\langle\mathit{vlist}\rangle\right)$.
\item \emph{Programs:} Each recursive program $\langle\mathit{prog}\rangle$ is a sequence of function entities, for which each function entity $\langle\mathit{func}\rangle$ consists of a function name followed by a list of parameters (composing a function declaration) and a curly-braced statement.
\end{compactitem}

\begin{figure}
\begin{minipage}{0.50\textwidth}
\centering
\begin{align*}
&\langle \mathit{prog}\rangle ::= \,\langle\mathit{func}\rangle\langle\mathit{prog}\rangle \mid \langle\mathit{func}\rangle
\\
&\langle \mathit{func}\rangle ::= \,\langle\mathit{fname}\rangle\mbox{`$($'}\langle plist\rangle\mbox{`$)$'}\mbox{`$\{$'}\langle stmt\rangle\mbox{`$\}$'}
\\
&\langle plist\rangle ::= \,\langle \mathit{pvar}\rangle\mid \langle \mathit{pvar}\rangle\mbox{`$,$'}\langle \mathit{plist}\rangle
\\
\vspace{\baselineskip}
\\
&\langle \mathit{stmt}\rangle ::= \mbox{`\textbf{skip}'}\mid \langle\mathit{pvar}\rangle \,\mbox{`$:=$'}\, \langle\mathit{expr} \rangle\\
& \mid \langle\mathit{fname}\rangle\mbox{`$($'}\langle\mathit{vlist}\rangle\mbox{`$)$'}\\
& \mid \mbox{`\textbf{if}'} \, \langle\mathit{bexpr}\rangle\,\mbox{`\textbf{then}'} \, \langle \mathit{stmt}\rangle \, \mbox{`\textbf{else}'} \, \langle \mathit{stmt}\rangle \,\mbox{`\textbf{fi}'}
\\
& \mid \mbox{`\textbf{if}'} \, \star\,\mbox{`\textbf{then}'} \, \langle \mathit{stmt}\rangle \, \mbox{`\textbf{else}'} \, \langle \mathit{stmt}\rangle \,\mbox{`\textbf{fi}'}
\\
&\mid  \mbox{`\textbf{while}'}\, \langle\mathit{bexpr}\rangle \, \text{`\textbf{do}'} \, \langle \mathit{stmt}\rangle \, \text{`\textbf{od}'}
\\
&\mid \langle\mathit{stmt}\rangle \, \text{`;'} \, \langle \mathit{stmt}\rangle
\end{align*}
\end{minipage}
\begin{minipage}{0.50\textwidth}
\begin{align*}
&\langle\mathit{expr} \rangle ::= \langle \mathit{int} \rangle \mid \langle\mathit{pvar}\rangle \\
&  \mid \langle\mathit{expr} \rangle\, \mbox{`$+$'} \,\langle\mathit{expr} \rangle \mid \langle\mathit{expr} \rangle\, \mbox{`$-$'} \,\langle\mathit{expr} \rangle \\
&
\mid \langle \mathit{int} \rangle \,\mbox{`$*$'} \, \langle\mathit{expr}\rangle \mid \mbox{`$\Big\lfloor$'}\frac{\langle\mathit{expr}\rangle}{\langle\mathit{nonzero}\rangle}\mbox{`$\Big\rfloor$'}
\\
&\langle\mathit{vlist}\rangle ::= \langle\mathit{expr}\rangle \mid  \langle\mathit{expr}\rangle\mbox{`,'} \langle\mathit{vlist}\rangle
\\
\vspace{\baselineskip}
\\
&\langle\mathit{literal} \rangle ::= \langle\mathit{expr} \rangle\, \mbox{`$\leq$'} \,\langle\mathit{expr} \rangle \mid \langle\mathit{expr} \rangle\, \mbox{`$\geq$'} \,\langle\mathit{expr} \rangle
\\
&\langle \mathit{bexpr}\rangle ::= \langle \mathit{literal} \rangle \mid \neg \langle\mathit{bexpr}\rangle\\
&\mid \langle \mathit{bexpr} \rangle \, \mbox{`\textbf{or}'} \, \langle\mathit{bexpr}\rangle
\mid \langle \mathit{bexpr} \rangle \, \mbox{`\textbf{and}'} \, \langle\mathit{bexpr}\rangle
\end{align*}
\end{minipage}
\caption{Syntax of Recursive Programs}
\label{fig:syntax}
\end{figure}

\smallskip\noindent{\bf Assumptions.}
W.l.o.g, we consider further syntactical restrictions
for simplicity:
\begin{compactitem}
\item \emph{Function Entities:} we consider that every parameter list $\langle\mathit{plist}\rangle$ contains no duplicate scalar variables, and the function names from function entities are distinct.
\item \emph{Function Calls:} we consider that no function call involves some function name without function entity (i.e., undeclared function names).
\end{compactitem}

\section{Control-Flow Graphs for Recursive Programs}\label{app:cfg}

\noindent\textbf{Intuitive Description.} It is intuitively clear that any recursive program can be transformed into a corresponding CFG: one first constructs each $\transitions{\fn{f}}$ (for $\fn{f}\in\fnames$) for each of its function bodies and then groups them together.
To construct each $\transitions{\fn{f}}$, we first construct the partial relation $\transitions{P,\fn{f}}$ inductively on the structure of $P$ for each statement $P$ appearing in the function body of $\fn{f}$, then define $\transitions{\fn{f}}$ as $\transitions{P_\fn{f},\fn{f}}$ for which $P_\fn{f}$ is the function body of $\fn{f}$.
Each relation $\transitions{P,\fn{f}}$ involves two distinguished labels, namely $\lin{P,\fn{f}}$ and $\lout{P,\fn{f}}$, that intuitively represent the label assigned to the first instruction to be executed in $P$ and the terminal program counter of $P$, respectively;
after the inductive construction, $\lin{\fn{f}},\lout{\fn{f}}$ are defined as $\lin{P_{\fn{f}},\fn{f}},\lout{P_{\fn{f}},\fn{f}}$, respectively.

\noindent\textbf{From Programs to CFG's.} In this part, we demonstrate inductively how the control-flow graph of a recursive program can be constructed.
Below we fix a recursive program $W$ and denote by $\fnames$ the set of function names appearing in $W$.
For each function name $\fn{f}\in\fnames$, we define $P_\fn{f}$ to be the function body of
$\fn{f}$, and define $\pvars{f}$ to be the set of scalar variables appearing in $P_\fn{f}$ and the parameter list of $\fn{f}$.

The control-flow graph of $W$ is constructed by first constructing the counterparts $\{\transitions{\fn{f}}\}_{\fn{f}\in\fnames}$ for each of its function bodies and then grouping them together.
To construct each $\transitions{\fn{f}}$, we first construct the partial relation $\transitions{P,\fn{f}}$ inductively on the structure of $P$ for each statement $P$ which involves variables solely from $\pvars{f}$, then define $\transitions{\fn{f}}$ as $\transitions{P_\fn{f},\fn{f}}$.

Let $\fn{f}\in\fnames$. Given an assignment statement of the form ${x}{:=}{\mathfrak{e}}$ involving variables solely from $\pvars{f}$ and a valuation $\nu\in\Val{f}$, we denote by ${\nu}{\assgn{\mathfrak{e}}{x}}$ the valuation over $\pvars{f}$ such that
\[
\left({\nu}{\assgn{\mathfrak{e}}{x}}\right)(q)=
\begin{cases} \nu(q) & \mbox{ if } q\in\pvars{f}\backslash\{x\} \\
\mathfrak{e}(\nu) & \mbox{ if }q=x
\end{cases}\enskip.
\]
Given a function call $\fn{g}(\mathfrak{e}_1,\dots,\mathfrak{e}_k)$ with variables solely from $\pvars{f}$ and its declaration being $\fn{g}(q_1,\dots,q_k)$, and
a valuation $\nu\in\Val{f}$, we define ${\nu}{[\fn{g},\{\mathfrak{e}_j\}_{1\le j\le k}]}$ to be a valuation over $\pvars{g}$ by:
\[
{\nu}{[\fn{g},\{\mathfrak{e}_j\}_{1\le j\le k}]}(q):=\begin{cases}
\mathfrak{e}_j(\nu) & \mbox{ if }q=q_j\mbox{ for some }j \\
0 & \mbox{ if }q\in \pvars{g}\backslash\{q_1,\dots,q_k\}
\end{cases}\enskip.
\]

Now the inductive construction for each $\transitions{P,\fn{f}}$ is demonstrated as follows.
For each statement $P$ which involves variables solely from $\pvars{f}$, the relation $\transitions{P,\fn{f}}$ involves two distinguished labels, namely $\lin{P,\fn{f}}$ and $\lout{P,\fn{f}}$, that intuitively represent the label assigned to the first instruction to be executed in $P$ and the terminal program counter of $P$, respectively.
After the inductive construction, $\lin{\fn{f}},\lout{\fn{f}}$ are defined as $\lin{P_{\fn{f}},\fn{f}},\lout{P_{\fn{f}},\fn{f}}$, respectively.

\begin{compactenum}
\item {\em Assignments.}
For $P$ of the form ${x}{:=}{\mathfrak{e}}$ or $\mbox{\textbf{skip}}$, $\transitions{P,\fn{f}}$ involves a new assignment label $\lin{P,\fn{f}}$ (as the initial label) and a new branching label $\lout{P,\fn{f}}$ (as the terminal label), and contains a sole triple $\left(\lin{P,\fn{f}},\nu\mapsto{\nu}{\assgn{\mathfrak{e}}{x}},\lout{P,\fn{f}}\right)$
or
$\left(\lin{P,\fn{f}},\nu\mapsto\nu,\lout{P,\fn{f}}\right)$,
respectively.
\item {\em Function Calls.}
For $P$ of the form $\fn{g}(\mathfrak{e}_1,\dots,\mathfrak{e}_k)$,
$\transitions{P,\fn{f}}$ involves a new call label $\lin{P,\fn{f}}$ and a new branching label $\lout{P,\fn{f}}$, and contains a sole triple $\left(\lin{P,\fn{f}},\left(\fn{g},\nu\mapsto\nu[\fn{g},\{\mathfrak{e}_j\}_{1\le j\le k}]\right),\lout{P,\fn{f}}\right)$.

\item {\em Sequential Statements.}
For ${P}{=}{Q_1;Q_2}$, we take the disjoint union of $\transitions{Q_1,\fn{f}}$ and  $\transitions{Q_2,\fn{f}}$, while redefining $\lout{Q_1,\fn{f}}$ to be $\lin{Q_2,\fn{f}}$ and putting $\lin{P,\fn{f}}:=\lin{Q_1,\fn{f}}$ and $\lout{P,\fn{f}}:=\lout{Q_2,\fn{f}}$.
\item {\em If-Branches.}
For ${P}{=}{\textbf{if $\phi$ then }Q_1 \textbf{ else } Q_2 \textbf{ fi}}$ with $\phi$ being a propositional arithmetic predicate, we first add two new branching labels $\lin{P},\lout{P}$, then take the disjoint union of $\transitions{Q_1,\fn{f}}$ and $\transitions{Q_2,\fn{f}}$ while simultaneously identifying both $\lout{Q_1,\fn{f}}$ and $\lout{Q_2,\fn{f}}$ with $\lout{P,\fn{f}}$, and finally obtain $\transitions{P,\fn{f}}$ by adding two triples $(\lin{P,\fn{f}},\phi,\lin{Q_1,\fn{f}})$ and $(\lin{P,\fn{f}},\neg\phi,\lin{Q_2,\fn{f}})$ into the disjoint union of $\transitions{Q_1,\fn{f}}$ and $\transitions{Q_2,\fn{f}}$.
\item {\em While-Loops.}
For ${P}{=}{\mbox{ \textbf{while} } \phi \mbox{ \textbf{do} }Q \mbox{ \textbf{od}}}$,
we add a new branching label $\lout{P,\fn{f}}$ as a terminal label and obtain $\transitions{P,\fn{f}}$ by adding triples $(\lout{Q,\fn{f}},\phi,\lin{Q,\fn{f}})$ and $(\lout{Q,\fn{f}},\neg\phi,\lout{P,\fn{f}})$ into $\transitions{Q,\fn{f}}$, and define $\lin{P,\fn{f}}:=\lout{Q,\fn{f}}$.
\item {\em Demonic Nondeterminism.} For ${P}{=}{\textbf{if $\star$ then }Q_1 \textbf{ else } Q_2 \textbf{ fi}}$,
we first add a new demonic label $\lin{P}$ and a new branching labels $\lout{P}$, then take the disjoint union of $\transitions{Q_1,\fn{f}}$ and $\transitions{Q_2,\fn{f}}$ while simultaneously identifying both $\lout{Q_1,\fn{f}}$ and $\lout{Q_2,\fn{f}}$ with $\lout{P,\fn{f}}$, and finally obtain $\transitions{P,\fn{f}}$ by adding two triples $(\lin{P,\fn{f}},\star,\lin{Q_1,\fn{f}})$ and $(\lin{P,\fn{f}},\star,\lin{Q_2,\fn{f}})$ into the disjoint union of $\transitions{Q_1,\fn{f}}$ and $\transitions{Q_2,\fn{f}}$.
\end{compactenum}

\section{Definition for Reachability}

\begin{definition}[Reachability]\label{def:reachability}
Let $\fn{f}^*$ be a function name and $\phi^*$ be a propositional arithmetic predicate over $\pvars{f^*}$.
A configuration $w$ is \emph{reachable} w.r.t $\fn{f}^*,\phi^*$ if there exist a scheduler $\pi$ and a
stack element $(\fn{f}^*,\lin{\fn{f}^*},\nu)$ such that (i) $\nu\models\phi^*$ and (ii)
$w$ appears in the run $\rho((\fn{f}^*,\lin{\fn{f}^*},\nu),\pi)$.
A stack element $\mathfrak{c}$ is \emph{reachable} w.r.t $\fn{f}^*,\phi^*$
if there exists a configuration $w$ reachable w.r.t $\fn{f}^*,\phi^*$ such that $w=\mathfrak{c}\cdot w'$ for some configuration $w'$.
\end{definition}

\section{Proof for Theorem~\ref{thm:mfunc}: Soundness}\label{app:thmmfuncsoundness}

\noindent\textbf{Theorem~\ref{thm:mfunc}.}
{\em (Soundness).}
For all measure functions $g$ w.r.t $\fn{f}^*,\phi^*$, it holds that for all valuations $\nu\in\Aval{\fn{f}^*}$ such that $\nu\models\phi^*$, we have $\overline{T}(\fn{f}^*,\lin{\fn{f}^*},\nu)\le g(\fn{f}^*,\lin{\fn{f}^*},\nu)$.
\begin{proof}
Define the function $h$ from the set of configurations into $[0,\infty]$ as follows:
\[
h(w):=\sum_{k=0}^n g(\fn{f}_k,\loc_k,\nu_k)\mbox{ for }w= \left\{(\fn{f}_k,\loc_k,\nu_k)\right\}_{0\le k\le n}
\]
where $h(\varepsilon):=0$.
We show that $h$ can be deemed as a ranking function over the set of reachable configurations w.r.t $\fn{f}^*,\phi^*$.

Let $\nu\in\Aval{f^*}$ be any valuation such that $\nu\models\phi^*$ and $\pi$ be any scheduler for $P$.
Moreover, let $\rho\left((\fn{f}^*,\lin{\fn{f}^*}, \nu),\pi\right)=\{w_n\}_{n\in\mathbb{N}_0}$.
Since the case $g(\fn{f}^*,\lin{\fn{f}^*},\nu)=\infty$ is trivial, we only consider the case $g(\fn{f}^*,\lin{\fn{f}^*},\nu)<\infty$.

By Definition~\ref{def:reachability}, every $w_n$ is reachable w.r.t $\fn{f}^*,\phi^*$.
Hence, by Definition~\ref{def:mfunc}, (a) for all $n\in\Nset_0$, if $w_n\ne\varepsilon$ then $h(w_n)\ge 1$, and $h(w_n)=0$ otherwise.
Furthermore, by Definition~\ref{def:mfunc} and our semantics,
one can easily verify that (b) for all $n\in\Nset_0$ , if $w_n\ne\varepsilon$ then $h(w_{n+1})\le h(w_{n})-1$.

To see (b), consider for example the function-call case where $w_n=(\fn{f},\loc,\nu)\cdot w'$, $\loc\in\flocs{f}$ and $(\loc,(\fn{g},f),\loc')$ is the only triple in $\transitions{\fn{f}}$ with source label $\loc$.
Then by our semantics, $h(w_n)=g(\fn{f},\loc,\nu)+h(w')$ and $h(w_{n+1})=g(\fn{g},\lin{\fn{g}},f(\nu))+g(\fn{f},\loc',\nu)+h(w')$.
Thus by C3, $h(w_{n+1})+1\le h(w_n)$.
The other cases (namely assignment, branching and nondeterminism) can be verified similarly through a direct investigation of our semantics and an application of C2, C4 or C5.

From (a), (b) and the fact that $h(w_0)=g(\fn{f}^*,\lin{\fn{f}^*}, \nu)<\infty$, one has that $m:=T\left((\fn{f}^*,\lin{\fn{f}^*}, \nu),\pi\right)<\infty$ (or otherwise (a) and (b) cannot hold simultaneously).
Furthermore, from an easy inductive proof based on (b), one has that $h(w_m)\le h(w_{0})-m$.
Together with $h(w_m)=0$ (from (a)), one obtains that
\[
T\left((\fn{f}^*,\lin{\fn{f}^*}, \nu),\pi\right)=m\le h(w_{0})=g(\fn{f}^*,\lin{\fn{f}^*},\nu)\enskip.
\]
Hence, $\overline{T}(\fn{f}^*,\lin{\fn{f}^*}, \nu)\le g(\fn{f}^*,\lin{\fn{f}^*},\nu)$\enskip.\qed
\end{proof}

\section{Proof for Theorem~\ref{thm:mfunc}: Completeness}\label{app:thmmfunc}

We recall that for each $\fn{f}\in\fnames$ and $\loc\in\locs{f}\backslash\{\lout{\fn{f}}\}$,
we define $D_{\fn{f},\loc}$ to be the set of all valuations $\nu$ w.r.t $\fn{f}$
such that $(\fn{f},\loc,\nu)$ is reachable w.r.t $\fn{f}^*,\phi^*$.

We introduce some notations for schedulers. Let $\pi$ be a scheduler for $P$ and $\mathfrak{c}$ be a non-terminal stack element.
We define $\mathrm{post}(\pi,\mathfrak{c})$ to be the scheduler such that for any non-empty finite word of configurations $w_0\dots w_n$ with $w_n$ being non-deterministic,
\[
\mathrm{post}(\pi,\mathfrak{c})(w_0\dots w_n)=\pi(\mathfrak{c}\cdot w_0\dots w_n)\enskip.
\]
In the case that $\mathfrak{c}$ is not non-deterministic, we define $\mathrm{pre}(\pi,\mathfrak{c})$ to be one of the  schedulers such that for any non-empty finite word of configurations $w_0\dots w_n$ with $w_n$ being non-deterministic, \[
\mathrm{pre}(\pi,\mathfrak{c})(\mathfrak{c}\cdot w_0\dots w_n)=\pi(w_0\dots w_n)\enskip;
\]
the decisions of $\mathrm{pre}(\pi,\mathfrak{c})$ at finite words not starting with $\mathfrak{c}$ will be irrelevant.
In the case that $\mathfrak{c}=(\fn{f},\loc,\nu)$ is non-deterministic, for any given $\loc'\in\locs{f}$ such that $(\loc,\star,\loc')\in\transitions{\fn{f}}$,
we define $\mathrm{pre}(\pi,(\mathfrak{c},\loc'))$ to be one of the schedulers such that (i)
\[
\mathrm{pre}(\pi,(\mathfrak{c},\loc'))(\mathfrak{c})=\loc'
\]
and
(ii) for any non-empty finite word of configurations $w_0\dots w_n$ with $w_n$ being non-deterministic,
\[
\mathrm{pre}(\pi,(\mathfrak{c},\loc'))(\mathfrak{c}\cdot w_0\dots w_n)=\pi(w_0\dots w_n)\enskip;
\]
again, the decisions of $\mathrm{pre}(\pi,(\mathfrak{c},\loc'))$ at finite words not starting with $\mathfrak{c}$ will be irrelevant.

\noindent\textbf{Theorem~\ref{thm:mfunc}.}
{\em (Completeness).}
$\overline{T}$ is a measure function w.r.t $\fn{f}^*,\phi^*$.
\begin{proof}
We prove that $\overline{T}$ satisfies the conditions C1--C5 in Definition~\ref{def:mfunc} with equality.
Condition C1 follows directly from the definition. Below we prove that C2--C5 hold.

Let $\mathfrak{c}=(\fn{f},\loc,\nu)$ be a non-terminal stack element such that $\nu\in D_{\fn{f},\loc}$.
Below we clarify several cases on $\mathfrak{c}$.

\paragraph{Case 1 (cf. C2):} $\loc\in\alocs{f}\backslash\{\lout{\fn{f}}\}$ and $(\loc,f,\loc')$ is the only triple in $\transitions{\fn{f}}$ with source label $\loc$. Consider any scheduler $\pi$ for $W$.
By our semantics, one can prove easily that
\begin{compactitem}
\item $T\left((\fn{f},\loc,\nu),\pi\right)= 1+T\left(\left(\fn{f},\loc',f(\nu)\right),\mathrm{post}(\pi,\mathfrak{c})\right)$, and
\item $T\left((\fn{f},\loc,\nu),\mathrm{pre}(\pi,\mathfrak{c})\right)= 1+T\left(\left(\fn{f},\loc',f(\nu)\right),\pi\right)$.
\end{compactitem}
Since $\pi$ is arbitrarily, by taking the supremum at the both sides of the equalities above one has that
\[
\overline{T}(\fn{f},\loc',f(\nu))+1= \overline{T}(\fn{f},\loc,\nu)\enskip.
\]

\paragraph{Case 2 (cf. C3):} $\loc\in\flocs{f}\backslash\{\lout{\fn{f}}\}$ and $(\loc,(\fn{g},f),\loc')$ is the only triple in $\transitions{\fn{f}}$ with source label $\loc$.

We first consider the case $\overline{T}(\fn{g},\lin{\fn{g}},\nu)=\infty$, meaning that schedulers $\pi$ can make $T\left((\fn{g},\lin{\fn{g}},\nu),\pi\right)$ arbitrarily large. Since for any scheduler $\pi$, $T\left((\fn{f},\loc,\nu),\mathrm{pre}(\pi,\mathfrak{c})\right)\ge 1+T\left((\fn{g},\lin{\fn{g}},\nu),\pi\right)$,
one has that
\[
\overline{T}(\fn{f},\loc,\nu)=1+\overline{T}(\fn{g},\lin{\fn{g}},f(\nu))+\overline{T}(\fn{f},\loc',\nu)(=\infty)\enskip.
\]
Then we consider the case that $\overline{T}(\fn{g},\lin{\fn{g}},\nu)<\infty$. On one hand, let $\pi$ be any scheduler for $W$.
Since $T\left((\fn{g},\lin{\fn{g}},\nu),\mathrm{post}(\pi,\mathfrak{c})\right)<\infty$, one can find a (unique) finite prefix $\gamma$ of the run $\rho\left((\fn{g},\lin{\fn{g}},\nu),\mathrm{post}(\pi,\mathfrak{c})\right)$ consisting of only non-empty (i.e. not $\varepsilon$) configurations which describes the finite execution of the function call $\fn{g}$ under the scheduler $\mathrm{post}(\pi,\mathfrak{c})$.
By our semantics, one can prove easily that
\begin{align*} T\left((\fn{f},\loc,\nu),\pi\right)=1+T\left(\left(\fn{g},\lin{\fn{g}},f(\nu)\right),\mathrm{post}(\pi,\mathfrak{c})\right) +T\left((\fn{f},\loc',\nu),\pi'\right)
\end{align*}
where $\pi'$ is the scheduler such that for any non-empty finite word of configurations $w_0\dots w_n$ with $w_n$ being non-deterministic, $\pi'(w_0\dots w_n)=\pi(\mathfrak{c}\cdot \gamma \cdot w_0\dots w_n)$.
It follows from the arbitrary choice of $\pi$ that
\[
\overline{T}(\fn{f},\loc,\nu)\le 1+\overline{T}(\fn{g},\lin{\fn{g}},f(\nu))+\overline{T}(\fn{f},\loc',\nu)\enskip.
\]
On the other hand, let $\pi_1,\pi_2$ be any two schedulers for $P$.
Since $\overline{T}\left(\fn{g},\lin{\fn{g}},\nu\right)<\infty$, one can find a (unique) finite prefix $\gamma$ of the run $\rho\left((\fn{g},\lin{\fn{g}},\nu),\pi_1\right)$ consisting of only non-empty (i.e. not $\varepsilon$) configurations which describes the finite execution of the function call $\fn{g}$ under the scheduler $\pi_1$.
Then
\[ T\left((\fn{f},\loc,\nu),\pi\right)=1+T\left(\left(\fn{g},\lin{\fn{g}},f(\nu)\right),\pi_1\right)+T\left((\fn{f},\loc',\nu),\pi_2\right)
\]
where $\pi$ is one of the schedulers such that (i) $\pi(\mathfrak{c}\cdot\beta)=\pi_1(\beta)$ whenever $\beta$ ends at a non-deterministic configuration and is a prefix of $\gamma$ (including $\gamma$) and (ii) $\pi(\mathfrak{c}\cdot\gamma\cdot \beta)=\pi_2(\beta)$ whenever $\beta$ is non-empty and ends at a non-deterministic configuration.
Thus, by the arbitrary choice of $\pi_1,\pi_2$, one has that
\[
1+\overline{T}\left(\fn{g},\lin{\fn{g}},f(\nu)\right)+\overline{T}\left(\fn{f},\loc',\nu\right)\le \overline{T}\left(\fn{f},\loc,\nu\right)\enskip.
\]
In either case, we have
\[
1+\overline{T}\left(\fn{g},\lin{\fn{g}},f(\nu)\right)+\overline{T}\left(\fn{f},\loc',\nu\right)= \overline{T}\left(\fn{f},\loc,\nu\right)\enskip.
\]

\paragraph{Case 3 (cf. C4):} $\loc\in\clocs{f}\backslash\{\lout{\fn{f}}\}$ and $(\loc, \phi,\loc_1),(\loc, \neg\phi,\loc_2)$ are namely two triples in $\transitions{\fn{f}}$ with source label $\loc$.
By our semantics, one can easily prove that
\begin{align*}
T\left((\fn{f},\loc,\nu),\pi\right)=1+\mathbf{1}_{\nu\models\phi}\cdot T\left((\fn{f},\loc_1,\nu),\mathrm{post}(\pi,\mathfrak{c})\right)+\mathbf{1}_{\nu\models\neg\phi}\cdot T\left((\fn{f},\loc_2,\nu),\mathrm{post}(\pi,\mathfrak{c})\right)
\end{align*}
for any scheduler $\pi$ for $W$, and
\begin{align*}
T\left((\fn{f},\loc,\nu),\pi\right)=1+\mathbf{1}_{\nu\models\phi}\cdot T\left((\fn{f},\loc_1,\nu),\pi_1\right)+\mathbf{1}_{\nu\models\neg\phi}\cdot T\left((\fn{f},\loc_2,\nu),\pi_2\right)
\end{align*}
for any schedulers $\pi_1,\pi_2$ for $P$,
where $\pi$ is either $\mathrm{pre}\left(\pi_1,\mathfrak{c}\right)$ if $\nu\models\phi$, or $\mathrm{pre}\left(\pi_2,\mathfrak{c}\right)$ if $\nu\models\neg\phi$.
By taking the supremum at the both sides of the equalities above, one has
\begin{align*}
\overline{T}\left(\fn{f},\loc,\nu\right)=1+\mathbf{1}_{\nu\models\phi}\cdot \overline{T}\left(\fn{f},\loc_1,\nu\right)+\mathbf{1}_{\nu\models\neg\phi}\cdot \overline{T}\left(\fn{f},\loc_2,\nu\right)\enskip.
\end{align*}

\paragraph{Case 4 (cf. C5):} $\loc\in\dlocs{f}\backslash\{\lout{\fn{f}}\}$ and $(\loc, \star,\loc_1),(\loc, \star,\loc_2)$ are namely two triples in $\transitions{\fn{f}}$ with source label $\loc$.
By our semantics, one easily proves that
\begin{align*}
T\left((\fn{f},\loc,\nu),\pi\right)=1+\mathbf{1}_{\pi(\mathfrak{c})=\loc_1}\cdot T\left((\fn{f},\loc_1,\nu), \mathrm{post}(\pi,\mathfrak{c})\right)+\mathbf{1}_{\pi(\mathfrak{c})=\loc_2}\cdot T\left((\fn{f},\loc_2,\nu), \mathrm{post}(\pi,\mathfrak{c})\right)\enskip.
\end{align*}
for any scheduler $\pi$ (for $P$), and
\begin{align*}
T\left((\fn{f},\loc,\nu),\pi\right)=1+\max\{T\left((\fn{f},\loc_1,\nu), \pi_1\right), T\left((\fn{f},\loc_2,\nu), \pi_2\right)\}\enskip.
\end{align*}
for any schedulers $\pi,\pi_1,\pi_2$ such that $\pi$ is either $\mathrm{pre}(\pi_1,(\mathfrak{c},\loc_1))$ if
$T\left((\fn{f},\loc_1,\nu), \pi_1\right)\ge T\left((\fn{f},\loc_2,\nu), \pi_2\right)$, or $\mathrm{pre}(\pi_2,(\mathfrak{c},\loc_2))$ if otherwise.
By taking the supremum at the both sides of the equalities above, one obtains
\begin{align*}
\overline{T}\left(\fn{f},\loc,\nu\right)=1+\max\left\{\overline{T}\left(\fn{f},\loc_1,\nu\right), \overline{T}\left(\fn{f},\loc_2,\nu\right)\right\}\enskip.
\end{align*}\qed
\end{proof}

\section{Omitted Details in Significant Label Construction}\label{app:propmfunc}

\paragraph{Expansion Construction (from $g$ to $\widehat{g}$)} In the following, we illustrate how one can obtain a measure function from a function defined only on significant labels.
Let $g$ be a function from
$\left\{(\fn{f},\loc,\nu)\mid \fn{f}\in\fnames, \loc\in\slocs{f},  \nu\in\Aval{f}\right\}$
into $[0,\infty]$.
The function \emph{expanded from} $g$, denoted by $\widehat{g}$, is a function from the set of all stack elements into $[0,\infty]$ inductively defined through the procedure described as follows.
\begin{compactenum}
\item {\em Initial Step.} If $\loc\in\slocs{f}$, then $\widehat{g}(\fn{f},\loc,\nu):=g(\fn{f},\loc,\nu)$\enskip.
\item {\em Termination.} If $\loc=\lout{\fn{f}}$, then $\widehat{g}(\fn{f},\loc,\nu):=0$\enskip.
\item {\em Assignment.} If $\loc\in\alocs{f}\backslash\slocs{f}$ with $(\loc,h,\loc')$ being the only triple in $\transitions{\fn{f}}$ and $\widehat{g}(\fn{f},\loc',\centerdot)$ is already defined, then
$\widehat{g}(\fn{f},\loc,\nu):=1+\widehat{g}(\fn{f},\loc',h(\nu))$.

\item {\em Branching.} If $\loc\in\clocs{f}\backslash\slocs{f}$ with $(\loc,\phi,\loc_1)$, $(\loc,\neg\phi,\loc_2)$ being namely the two triples in $\transitions{\fn{f}}$ and both $\widehat{g}(\fn{f},\loc_1,\centerdot)$ and $\widehat{g}(\fn{f},\loc_2,\centerdot)$ is already defined, then
$\widehat{g}(\fn{f},\loc,\nu):=\mathbf{1}_{\nu\models\phi}\cdot\widehat{g}(\fn{f},\loc_1,\nu)+\mathbf{1}_{\nu\models\neg\phi}\cdot\widehat{g}(\fn{f},\loc_2,\nu)+1$;
\item {\em Call.} If $\loc\in \flocs{f}\backslash\slocs{f}$ with $(\loc,(\fn{g},h),\loc')$ being the only triple in $\transitions{\fn{f}}$ and $\widehat{g}(\fn{f},\loc',\centerdot)$ is already defined, then
$\widehat{g}(\fn{f},\loc,\nu):=g(\fn{g},\lin{\fn{g}}, h(\nu))+\widehat{g}(\fn{f},\loc',\nu)+1$.
\item {\em Non-determinism.} If $\loc\in\dlocs{f}\backslash\slocs{f}$ with $(\loc,\star,\loc_1)$, $(\loc,\star,\loc_2)$ being namely the two triples in $\transitions{\fn{f}}$ with source label $\loc$ and both $\widehat{g}(\fn{f},\loc_1,\centerdot)$ and $\widehat{g}(\fn{f},\loc_2,\centerdot)$ is already defined, then
$\widehat{g}(\fn{f},\loc,\nu):=\max\left\{\widehat{g}(\fn{f},\loc_1,\nu), \widehat{g}(\fn{f},\loc_2,\nu)\right\}+1$.
\end{compactenum}

Note that in the previous expansion construction, we have not technically shown that $\widehat{g}$ is defined over all stack elements.
The following technical lemma shows that the function $\widehat{g}$ is indeed well-defined.

\begin{lemma}\label{lemm:expansion}
For each function $g$ from $\left\{(\fn{f},\loc,\nu)\mid \fn{f}\in\fnames, \loc\in\slocs{f},  \nu\in\Aval{f}\right\}$
into $[0,\infty]$, the function $\widehat{g}$ is well-defined.
\end{lemma}
\begin{proof}
Suppose that $\widehat{g}$ is not well-defined, i.e., there exists some $\fn{f}\in\fnames$ and $\loc_0\in\locs{f}$ such that $\widehat{g}(\fn{f},\loc_0,\centerdot)$ remains undefined.
Then by the inductive procedure, there exists a triple $(\loc_0,\alpha,\loc_1)$ in $\transitions{\fn{f}}$ such that $\widehat{g}(\fn{f},\loc_1,\centerdot)$ remains undefined.
With the same reasoning, one can inductively construct an infinite sequence $\{\loc_j\}_{j\in\Nset_0}$ such that each $\widehat{g}(\fn{f},\loc_j,\centerdot)$ remains undefined.
Since $\locs{f}$ is finite, there exist $j_1,j_2$ such that $j_1\ne j_2$ and $\loc_{j_1}=\loc_{j_2}$. It follows from our semantics that there exists $j^*$ such that $j_1\le j^*\le j_2$ and $\loc_{j^*}$ corresponds to the initial label of a while-loop in $W$.
Contradiction to the fact that $\loc_{j^*}\in\slocs{f}$.\qed
\end{proof}

\begin{proposition}\label{prop:mfunc}
Let $g$ be a function from
$\left\{(\fn{f},\loc,\nu)\mid \fn{f}\in\fnames, \loc\in\slocs{f},  \nu\in\Aval{f}\right\}$
into $[0,\infty]$.
Let $I$ be an invariant w.r.t $\fn{f}^*,\phi^*$.
Consider that for all stack elements $(\fn{f},\loc,\nu)$ such that $\loc\in\slocs{f}$ and  $\nu\models I(\fn{f},\loc)$,
the following conditions hold:
\begin{compactitem}
\item \textbf{C2':} if $\loc\in\alocs{f}$ and $(\loc,f,\loc')$ is the only triple in $\transitions{\fn{f}}$ with source label $\loc$, then  $\widehat{g}(\fn{f},\loc',f(\nu))+1\le \widehat{g}(\fn{f},\loc,\nu)$;
\item \textbf{C3':} if $\loc\in\flocs{f}$ and $(\loc,(\fn{g},f),\loc')$ is the only triple in $\transitions{\fn{f}}$ with source label $\loc$, then $1+\widehat{g}(\fn{g},\lin{\fn{g}},f(\nu))+\widehat{g}(\fn{f},\loc',\nu)\le \widehat{g}(\fn{f},\loc,\nu)$;
\item \textbf{C4':} if $\loc\in\clocs{f}$ and $(\loc, \phi,\loc_1),(\loc, \neg\phi,\loc_2)$ are namely two triples in $\transitions{\fn{f}}$ with source label $\loc$, then
$\mathbf{1}_{\nu\models\phi}\cdot \widehat{g}(\fn{f},\loc_1,\nu)+\mathbf{1}_{\nu\models\neg\phi}\cdot \widehat{g}(\fn{f},\loc_2,\nu)+1\le \widehat{g}(\fn{f},\loc,\nu)$\enskip;
\item \textbf{C5':} if $\loc\in\dlocs{f}$ and $(\loc, \star,\loc_1),(\loc, \star,\loc_2)$ are namely two triples in $\transitions{\fn{f}}$ with source label $\loc$, then
$\max\{\widehat{g}(\fn{f},\loc_1,\nu), \widehat{g}(\fn{f},\loc_2,\nu)\}+1\le \widehat{g}(\fn{f},\loc,\nu)$.
\end{compactitem}
Then $\widehat{g}$ is a measure function w.r.t $\fn{f}^*,\phi^*$.
\end{proposition}
\begin{proof}
The proof follows directly from the fact that (i) all valuations in $D_{\fn{f},\loc}$ satisfy $I(\fn{f},\loc)$, for all $\fn{f}\in\fnames$ and $\loc\in\locs{f}\backslash\{\lout{\fn{f}}\}$, (ii) C2'-C5' directly implies C2-C5 for $\loc\in\slocs{f}$ and
(iii) C1-C5 are automatically satisfied for $\loc\not\in\slocs{f}$ by the expansion construction of
significant labels.
\qed
\end{proof}

\section{Omitted Details for Sect.~\ref{sect:synalg}}\label{app:synalgo}

In this section, we present the omitted details on our algorithm for synthesizing measure functions.
As mentioned before, the synthesis algorithm is designed to synthesize one function over valuations at each function name and appropriate significant label, so that conditions C2'-C5' in Proposition~\ref{prop:mfunc} are fulfilled.

Below we fix an input recursive program $P$ with its CFG taking the form
($\dag$), an input invariant $I$ in disjunctive normal form and an input pair of parameters $(d,\mathrm{op},r,k)$.
We demonstrate our algorithm step in step as follows.

\subsection{Step~1 of \synalgo}
All details are presented in the main article.

\subsection{Step~2 of \synalgo}

\paragraph*{Computation of $\widehat{\llbracket\eta\rrbracket}$}

In the computation, we use the fact that for real-valued functions $\{f_i\}_{1\le i\le m}$, $\{g_j\}_{1\le j\le n}$ and $h$, it holds that
\[
\max{{\{f_i\}}_{i}}+\max{{\{g_j\}}_{j}}=\max\{f_i+g_j\}_{i,j}
\]
and
\[
h\cdot \max\{f_i\}_i=\max\{h\cdot f_i\}_i
\]
provided that $h$ is everywhere non-negative,
where the maximum function over a finite set of functions is defined in pointwise fashion.
Moreover, we use the facts that (i)
\[
\mathbf{1}_{\nu\models\phi_1}\cdot\mathbf{1}_{\nu\models\phi_2}=\mathbf{1}_{\nu\models\phi_1\wedge\phi_2}
\]
for all propositional arithmetic predicates $\phi_1,\phi_2$ and valuations $\nu$, and (ii)
\[
\left(\sum_{i=1}^{m}\mathbf{1}_{\phi_i}\cdot g_i\right)+\left(\sum_{j=1}^n\mathbf{1}_{\psi_j}\cdot h_j\right)=\sum_{i=1}^{m}\sum_{j=1}^n \mathbf{1}_{\phi_i\wedge\psi_j}\cdot (h_i+g_j)
\]
provided that
(a) $\bigvee_i\phi_i,\bigvee_j\psi_j$ are both tautology and (b) $\phi_{i_1}\wedge\phi_{i_2},\psi_{j_1}\wedge\psi_{j_2}$ are both unsatisfiable whenever $i_1\ne i_2$ and $j_1\ne j_2$, and (iii) propositional arithmetic predicates are closed under substitution of expressions in $\langle\mathit{expr}\rangle$ for scalar variables.

\subsection{Step~3 of \synalgo}

\paragraph*{Establishment of Constraint Triples.}
Based on $\widehat{\llbracket\eta\rrbracket}$, the algorithm generates constraint triples at each significant label, then group all generated constraint triples together in a conjunctive way.

Let $\loc$ be any significant label at any function name $\fn{f}$, the algorithm generates constraint triples at $\fn{f},\loc$ as follows.
W.l.o.g, let $I(\fn{f},\loc)=\bigvee_{l}\Psi_l$ where each $\Psi_l$ is a conjunction of atomic formulae of the form $\mathfrak{e}\ge 0$.
The algorithm first generate constraint triples related to non-negativity of measure functions.
\begin{compactitem}
\item \textbf{Non-negativity}: The algorithm generates the collection of constraint triples
\[
\left\{\left(\fn{f},\Psi_l, \eta(\fn{f},\loc)\right)\right\}_{l}\enskip.
\]
\end{compactitem}
Then the algorithm generates constraint triples through C2'-C5' as follows.
\begin{compactitem}
\item \textbf{Case 1 (C2')}: $\loc\in\slocs{f}\cap\alocs{f}$ and $(\loc,f,\loc')$ is the sole triple in $\transitions{\fn{f}}$ with source label $\loc$.
As $\widehat{\llbracket\eta\rrbracket}(\fn{f},\loc',\centerdot)$ can be represented in the form (\ref{eq:hatform}), $\widehat{\llbracket\eta\rrbracket}(\fn{f},\loc',f(\centerdot))$ can also be represented by an expression
\[
\max\left\{\sum_{j}\mathbf{1}_{\phi_{1j}}\cdot h_{1j},~\dots,~\sum_{j}\mathbf{1}_{\phi_{mj}}\cdot h_{mj}\right\}
\]
in the form (\ref{eq:hatform}).
Let a disjunctive normal form of each formula $I(\fn{f},\loc)\wedge \phi_{ij}$ be $\bigvee_{l}\Phi_{ij}^l$, where each $\Phi_{ij}^l$ is a conjunction of atomic formulae of the form $\mathfrak{e}'\ge 0$.
Then the algorithm generates the collection of constraint triples
\[
\left\{\left(\fn{f},\Phi_{ij}^l, \eta(\fn{f},\loc)-h_{ij}-1\right)\right\}_{i,j,l}\enskip.
\]
\item \textbf{Case 2 (C3')}: $\loc\in\flocs{f}$ and $(\loc,(\fn{g},f),\loc')$ is the sole triple in $\transitions{\fn{f}}$ with source label $\loc$.
Let $\widehat{\llbracket\eta\rrbracket}\left(\fn{g},\lin{\fn{g}},f(\centerdot)\right)+\widehat{\llbracket\eta\rrbracket}(\fn{f},\loc',\centerdot)$
be represented by the expression
\[
\max\left\{\sum_{j}\mathbf{1}_{\phi_{1j}}\cdot h_{1j},~\dots,~\sum_{j}\mathbf{1}_{\phi_{mj}}\cdot h_{mj}\right\}
\]
in the form (\ref{eq:hatform}).
Let a disjunctive normal form of each formula $I(\fn{f},\loc)\wedge \phi_{ij}$ be $\bigvee_{l}\Phi_{ij}^l$, where each $\Phi_{ij}^l$ is a conjunction of atomic formulae of the form $\mathfrak{e}'\ge 0$.
Then the algorithm generates the collection of constraint triples
\[
\left\{\left(\fn{f},\Phi_{ij}^l, \eta(\fn{f},\loc)-h_{ij}-1\right)\right\}_{i,j,l}\enskip.
\]
\item \textbf{Case 3 (C4')}: $\loc\in\clocs{f}$ and $(\loc, \phi,\loc_1),(\loc, \neg\phi,\loc_2)$ are namely two triples in $\transitions{\fn{f}}$ with source label $\loc$.
Let $h$ be the function (parametric on template variables)
\[
\mathbf{1}_\phi\cdot \widehat{\llbracket\eta\rrbracket}(\fn{f},\loc_1,\centerdot)+\mathbf{1}_{\neg\phi}\cdot \widehat{\llbracket\eta\rrbracket}(\fn{f},\loc_2,\centerdot)\enskip.
\]
Then $h$ can be represented by an expression
\[
\max\left\{\sum_{j}\mathbf{1}_{\phi_{1j}}\cdot h_{1j},~\dots,~\sum_{j}\mathbf{1}_{\phi_{mj}}\cdot h_{mj}\right\}
\]
in the form (\ref{eq:hatform}).
Let a disjunctive normal form of each formula $I(\fn{f},\loc)\wedge \phi_{ij}$  be $\bigvee_{l}\Phi_{ij}^l$, where each $\Phi_{ij}^l$ is a conjunction of atomic formulae of the form $\mathfrak{e}'\ge 0$.
Then the algorithm generates the collection of constraint triples
\[
\left\{\left(\fn{f},\Phi_{ij}^l, \eta(\fn{f},\loc)-h_{ij}-1\right)\right\}_{i,j,l}\enskip.
\]
\item \textbf{Case 4 (C5')} $\loc\in\dlocs{f}$ and $(\loc, \star,\loc_1),(\loc, \star, \loc_2)$ are namely two triples in $\transitions{\fn{f}}$ with source label $\loc$.
Let $h$ be the function (parametric on template variables)
\[
\max\left\{\widehat{\llbracket\eta\rrbracket}(\fn{f},\loc_1,\centerdot), \widehat{\llbracket\eta\rrbracket}(\fn{f},\loc_2,\centerdot)\right\}
\]
otherwise.
Then $h$ can be represented by an expression
\[
\max\left\{\sum_{j}\mathbf{1}_{\phi_{1j}}\cdot h_{1j},~\dots,~\sum_{j}\mathbf{1}_{\phi_{mj}}\cdot h_{mj}\right\}
\]
in the form (\ref{eq:hatform}).
Let a disjunctive normal form of each formula $I(\fn{f},\loc)\wedge \phi_{ij}$  be $\bigvee_{l}\Phi_{ij}^l$.
Then the algorithm generates the collection of constraint triples
\[
\left\{\left(\fn{f},\Phi_{ij}^l, \eta(\fn{f},\loc)-h_{ij}-1\right)\right\}_{i,j,l}\enskip.
\]
\end{compactitem}
After generating the constraint triples for each significant label,
the algorithm group them together in the conjunctive fashion to form a single collection of constraint triples.

\subsection{Step~4 of \synalgo}

All details of Step~4 is in the main text.


\vspace{-1em}
\subsection{Step~5 of \synalgo}
\vspace{-0.5em}

All details are presented in the main article.

\smallskip
\noindent\textbf{Theorem~\ref{thm:mainthmsynth}.}
Our algorithm, \synalgo, is a sound
approach for the \recterm\ problem, i.e., if \synalgo\ succeeds to
synthesize a function $g$ on $\left\{(\fn{f},\loc,\nu)\mid \fn{f}\in\fnames, \loc\in\slocs{f},  \nu\in\Aval{f}\right\}$,
then $\widehat{g}$ is a measure function and hence an upper bound on
the pessimistic termination time.
\begin{proof}
The proof follows from the facts that (i) once the logical formulae encoded by the constraint triples (cf. semantics of constraint triples specified in Step 3) are fulfilled by template variables, then $\widehat{\llbracket\eta\rrbracket}$ (obtained in Step 2, with $\eta$ being the template established in Step 1) satisfies the conditions specified in Proposition~\ref{prop:mfunc}, (ii) the variable abstraction (i.e., the linear inequalities) in Step 4 is a sound over-approximation for floored expressions, logarithmic terms and exponentiation terms, and (iii) Handelman's Theorem provides a sound form for positive polynomials over polyhedra.\qed
\end{proof}

\section{Illustration on Merge-Sort}
We now present a step-by-step illustration of our entire method
(from syntax, to semantics, to all the steps of the algorithm) on the classical
Merge-Sort. This illustrates the quite involved aspects
of our approach.


\label{ex:mergesort}
\subsection{Program Implementation}
 Figure~\ref{fig:mergesort} represents an implementation for the Merge-Sort algorithm~\cite[Chapter 2]{DBLP:books/daglib/0023376} in our language.
The numbers on the leftmost side are the labels assigned to statements which represent program counters,
where the function $\mathsf{mergesort}$ starts from label $1$ and ends at $7$, and $\mathsf{merge}$ starts from label $8$ and ends at $22$.

In detail, the scalar variable $i$ (resp. $j$) in both the parameter list of $\mathsf{mergesort}$ and that of $\mathsf{merge}$ represents the start (resp. the end) of the array index; the scalar variable $k$ in the parameter list of $\mathsf{merge}$ represents the separating index for merging two sub-arrays between $i,j$; moreover, in $\mathsf{merge}$, the demonic nondeterministic branch (at program counter $12$) abstracts the comparison between array entries and the \textbf{skip}'s
at program counters $13,15,20$ represent corresponding array assignment statements in a real implementation for Merge-Sort.

The replacement of array-relevant operations with either \textbf{skip} or demonic non-determinism preserves worst-case complexity as demonic non-determinism over-approximates conditional branches involving array entries.

\lstset{language=prog}
\lstset{tabsize=3}
\newsavebox{\progmergesort}
\begin{lrbox}{\progmergesort}
\begin{lstlisting}[mathescape]
$\mathsf{mergesort}(i, j)$ {
1:  if $1\le i$ and $i\le j-1$ then
2:      $k:=\lfloor \frac{i+j}{2}\rfloor$;
3:      $\mathsf{mergesort}(i, k)$;
4:      $\mathsf{mergesort}(k+1, j)$;
5:      $\mathsf{merge}(i,j,k)$
6:    else  skip
    fi
7: }
\end{lstlisting}
\end{lrbox}

\lstset{language=prog}
\lstset{tabsize=3}
\newsavebox{\progmerge}
\begin{lrbox}{\progmerge}
\begin{lstlisting}[mathescape]
$\mathsf{merge}(i,j,k)$ {
8:  $m:=i$;  
9:  $n:=k+1$;  
10: $l:=i$;
11:  while $l\le j$ do
12:    if $\star$ then
13:        skip;
14:		   $m:=m+1$
      else
15:        skip;  
16:		   $n:=n+1$
      fi;
17:   $l:=l+1$
    od;
18: $l:=i$;
19: while $l\le j$ do
20:   skip;
21:   $l:=l+1$
    od
22: }
\end{lstlisting}
\end{lrbox}

\begin{figure}
\begin{minipage}{0.50\textwidth}
\centering
\usebox{\progmergesort}
\end{minipage}
\begin{minipage}{0.50\textwidth}
\centering
\usebox{\progmerge}
\end{minipage}
\caption{A program that implements Merge-Sort}
\label{fig:mergesort}
\end{figure}

\subsection{Control Flow Graph and Significant Labels}
To begin the analysis, the algorithm must first obtain the control-flow graph (CFG) of the program. This is depicted in Figures~\ref{tbl:illustrations},
~\ref{fig:cfgmergesort} and ~\ref{fig:cfgmerge}.

 For brevity we define $\phi:= 1 \le i \wedge i \le j-1$, $\psi:=l\le j$, and let $\mbox{\sl id}$ indicate the identity function. The update functions ($f_i$'s and $g_i$'s) are given in Figure~\ref{tbl:illustrations}.
Moreover, in Figure~\ref{tbl:illustrations},
any $\overline{q}$ is the concrete entity held by the scalar variable $q$ under the valuation at runtime, and every function is represented in the form ``$p\leftarrow q$'' meaning ``$q$ assigned to $p$'', where only the relevant variable is shown for assignment functions.

Figure~\ref{fig:cfgmergesort} shows that $\mathsf{mergesort}$ (cf.~Fig.~\ref{fig:mergesort}) starts from the if-branch with guard $\phi$ (label $1$).
Then if $\phi$ is satisfied (i.e., the length of the array is greater than one), the program steps into the recursive step which is composed of labels 2--5; otherwise, the program proceeds to terminal label $7$ through $6$ with nothing done.

Figure~\ref{fig:cfgmerge} shows that $\mathsf{merge}$ (cf. Fig.~\ref{fig:mergesort}) starts from a series of assignments (labels 8--10) and then enters the while-loop (labels 11--17) which does the merging of two sub-arrays; after the while-loop starting from $11$, the program enters the part for copying array-content back (labels 18--21) and finally enters the terminal label $22$.

The types of labels are straightforward. In $\mathsf{mergesort}$, labels $2,6$ are assignment labels, label $1$ is a branching label, labels $3\mbox{--}5$ are call labels. In $\mathsf{merge}$, labels $8\mbox{--}10,13\mbox{--}18,20\mbox{--}21$ are assignment labels, labels $11,19$ are branching labels, and label $12$ is a demonic non-deterministic point. 

\begin{figure}
\begin{minipage}{0.30\textwidth}
\centering
\begin{tabular}{|c|c|c|}
\hline
$i$ & $f_i$ \\
\hline
$1$ & $k\leftarrow \lfloor (\overline{i} + \overline{j})/2\rfloor$  \\
\hline
$2$ & $(i, j)\leftarrow (\overline{i}, \overline{k})$  \\
\hline
$3$ & $(i,j)\leftarrow (\overline{k}+1, \overline{j})$ \\
\hline
$4$ & $(i,j,k)\leftarrow(\overline{i},\overline{j},\overline{k})$   \\
\hline
\hline
$i$ & $g_i$ \\
\hline
1   & $m\leftarrow \overline{i}$  \\
\hline
2   &  $n\leftarrow \overline{k}+1$ \\
\hline
3   &  $l\leftarrow \overline{i}$ \\
\hline
4   &  $m\leftarrow \overline{m}+1$ \\
\hline
5   &  $n\leftarrow \overline{n}+1$ \\
\hline
6   &  $l\leftarrow \overline{l}+1$ \\
\hline
\end{tabular}
\caption{Illustration for Fig.~\ref{fig:cfgmergesort} and Fig.~\ref{fig:cfgmerge}}
\label{tbl:illustrations}
\end{minipage}
\begin{minipage}{0.05\textwidth}
~
\end{minipage}
\begin{minipage}{0.20\textwidth}
\centering
\begin{tikzpicture}[x = 1.5cm]

\node[label] (if)         at (0,0)         {$1$};
\node[label] (body)       at (0,-1)        {$2$};
\node[label] (firsthalf)  at (0,-2)        {$3$};
\node[label] (secondhalf) at (0,-3)        {$4$};
\node[label] (merge)      at (0,-4)        {$5$};
\node[label] (skip)       at (-0.7,-2)       {$6$};
\node[label] (end)        at (0,-5)        {$7$};

\draw[tran] (if) to node[auto, font=\scriptsize] {$\phi$} (body);
\draw[tran] (body) to node[auto, font=\scriptsize] {$f_1$} (firsthalf);

\draw[tran] (firsthalf) to node[auto, font=\scriptsize] {$\left(\mathsf{mergesort},f_2\right)$} (secondhalf);

\draw[tran] (secondhalf) to node[auto, font=\scriptsize] {$\left(\mathsf{mergesort},f_3\right)$} (merge);

\draw[tran] (merge) to node[auto, font=\scriptsize] {$(\mathsf{merge}, f_4)$}  (end);

\draw[tran] (if) to node[left, font=\scriptsize] {$\neg\phi$} (skip);

\draw[tran] (skip) to node[left, font=\scriptsize] {$\mbox{\sl id}$} (end);
\end{tikzpicture}
\caption{The part of CFG for $\mathsf{mergesort}$ in Fig.~\ref{fig:mergesort}}
\label{fig:cfgmergesort}
\end{minipage}
\begin{minipage}{0.05\textwidth}
~
\end{minipage}
\begin{minipage}{0.35\textwidth}
\centering
\begin{tikzpicture}[x = 1.5cm]
\node[label] (massgn)         at (-2.1,0)       {$8$};
\node[label] (nassgn)         at (-1.4,0)       {$9$};
\node[label] (lassgn)         at (-0.7,0)        {$10$};
\node[label] (merge)          at (0,0)        {$11$};
\node[label] (lreassgn)       at (-1,-1)       {$18$};
\node[label] (arcomp)         at (1,-1)        {$12$};
\node[label] (tmpassgn1)      at (1,-2)        {$13$};
\node[label] (minc)           at (1,-3)        {$14$};
\node[label] (tmpassgn2)      at (0,-2)        {$15$};
\node[label] (ninc)           at (0,-3)        {$16$};
\node[label] (linc)           at (1,-4)        {$17$};
\node[label] (copyback)       at (-1,-2)       {$19$};
\node[label] (arcopy)         at (-1,-3)       {$20$};
\node[label] (lreinc)         at (-1,-4)       {$21$};
\node[label] (end)            at (-2,-3)       {$22$};

\node        (temp1)          at (1.7,1)         {};
\node        (temp2)          at  (-0.5,1)     {};

\draw[tran] (massgn) to node[auto, font=\scriptsize] {$g_1$} (nassgn);
\draw[tran] (nassgn) to node[auto, font=\scriptsize] {$g_2$} (lassgn);
\draw[tran] (lassgn) to node[auto, font=\scriptsize] {$g_3$} (merge);
\draw[tran] (merge)  to node[auto, font=\scriptsize] {$\psi$} (arcomp);
\draw[tran] (merge)  to node[left, font=\scriptsize] {$\neg\psi$} (lreassgn);
\draw[tran] (arcomp)  to node[auto, font=\scriptsize] {$\star$} (tmpassgn1);
\draw[tran] (arcomp)  to node[left, font=\scriptsize] {$\star$} (tmpassgn2);

\draw[tran] (tmpassgn1)  to node[auto, font=\scriptsize] {$\mbox{\sl id}$} (minc);
\draw[tran] (tmpassgn2)  to node[auto, font=\scriptsize] {$\mbox{\sl id}$} (ninc);
\draw[tran] (minc)  to node[auto, font=\scriptsize] {$g_4$} (linc);
\draw[tran] (ninc)  to node[auto, font=\scriptsize] {$g_5$} (linc);
\draw[tran] (linc)  -- (linc-|temp1) -- node[auto, font=\scriptsize] {$g_6$}  (merge-|temp1) --  (merge);

\draw[tran] (lreassgn) to node[auto, font=\scriptsize] {$g_3$} (copyback);

\draw[tran] (copyback) to node[auto, font=\scriptsize] {$\psi$} (arcopy);
\draw[tran] (copyback) to node[left, font=\scriptsize] {$\neg\psi$} (end);
\draw[tran] (arcopy) to node[auto, font=\scriptsize] {$\mbox{\sl id}$} (lreinc);
\draw[tran] (lreinc) -- (lreinc-|temp2) -- node[right, font=\scriptsize] {$g_6$}  (copyback-|temp2) --  (copyback);
\end{tikzpicture}
\caption{The part of CFG for $\mathsf{merge}$ in Fig.~\ref{fig:mergesort}}
\label{fig:cfgmerge}
\end{minipage}
\end{figure}

In this program, significant labels are $1$ and $8$ (beginning points of functions) and $11$ and $19$ (while loop headers).

\subsection{Step 1: Invariants and Templates}

As mentioned in Section \ref{sect:mfunc}, automatic obtaining of invariants is a standard problem with several known techniques\cite{DBLP:conf/popl/CousotC77,DBLP:conf/cav/ColonSS03}. Since invariant generation is not a main part of our algorithm, we simply use the straightforward invariants in Table \ref{tbl:invar} to demonstrate our algorithm.

\begin{table}[]
	\centering

	\begin{tabular}{|c|c|c|}
		\hline
		\textbf{$\fn{f}$} & \textbf{$\loc$} & \textbf{$I(\fn{f}, \loc)$} \\ \hline \hline
		$\mathsf{mergesort}$    & $1$                  & $i \ge 0 \wedge j \ge i$       \\ \hline
		$\mathsf{merge}$        & $8$                  & $i \ge 0 \wedge j \ge i$   \\ \hline
		$\mathsf{merge}$		& $11$				   & $l \ge i \wedge l \le j + 1$ \\ \hline
		$\mathsf{merge}$        & $19$                 & $l \ge i \wedge l \le j + 1$   \\ \hline
	\end{tabular}
\caption{Invariants Used for Significant Labels}
\label{tbl:invar}
\end{table}

The algorithm constructs a template $\eta(\fn{f}, \loc, \centerdot)$ of the form shown in (\ref{eq:synform}) at every significant label. Due to the large amount of space needed to illustrate the method on the complete template, for the sake of this example, we restrict our templates to the forms presented in Table \ref{tbl:templates} instead. Since a template with this restricted form is feasible, the whole template consisting of all products of pairs of terms is also feasible. Here the $c_i$'s are variables that the algorithm needs to synthesize.

Note that the results reported in Section \ref{sec:experimental-results} were obtained by our implementation in the general case, where the templates were not given as part of input and were generated automatically by our algorithm in their full form according to (\ref{eq:synform}).

\begin{table}[]
	\centering
	
	\begin{tabular}{|c|c|c|}
		\hline
		\textbf{$\fn{f}$}    & \textbf{$\loc$} & \textbf{$\eta(\fn{f}, \loc)$}          \\ \hline \hline
		$\mathsf{mergesort}$ & $1$          & $c_1 (j - i + 1) \ln(j - i + 1) + c_2$ \\ \hline
		$\mathsf{merge}$     & $8$          & $c_3 (j - i + 1) + c_4$                \\ \hline
		$\mathsf{merge}$	 & $11$			& $c_5 l + c_6 j + c_7 i + c_8$ \\ \hline
		$\mathsf{merge}$     & $19$         & $c_9 l + c_{10} j + c_{11}$                  \\ \hline
	\end{tabular}
\caption{Templates at Significant Labels}
\label{tbl:templates}
\end{table}

\newcommand{\etahat}{\widehat{\llbracket\eta\rrbracket}}
\newcommand{\ex}{\mathfrak{e}}
\subsection{Step~2: Computation of $\etahat$}

In this step, the algorithm expands the $\eta$ generated for significant labels in the previous step to obtain $\etahat$ for all labels. This results in a function of form (\ref{eq:hatform}). In order to make equations easier to read, we use the notation $\binom{\phi}{p}$ to denote $\mathbf{1}_{\phi}\cdot p$ in this section. We also use $\etahat(\fn{f},\loc)[p \leftarrow \ex]$ to denote the result of replacing every occurrence of the variable $p$ in $\etahat(\fn{f}, \loc)$ with the expression $\ex$.

\emph{End points of functions.} The algorithm sets $\etahat(\fn{mergesort}, 7)$ and $\etahat(\fn{merge},22)$ to $0$, since these correspond to the terminal labels of their respective functions.

\emph{Significant Labels.} For each significant label $\loc$ of a function $\fn{f}$, the algorithm, by definition, sets $\etahat(\fn{f},\loc)$ to:
$$
\binom{I(\fn{f},\loc)}{\eta(\fn{f},\loc)} + \binom{\neg I(\fn{f}, \loc)}{0}.
$$

Therefore we have:

$$	\etahat(\fn{mergesort}, 1) = \binom{i \ge 0 \wedge j \ge i}{c_1 (j-i+1) \ln(j-i+1) + c_2} + \binom{i<0 \vee j<i}{0},$$
$$	\etahat(\fn{merge}, 8) = \binom{i \ge 0 \wedge j \ge i}{c_3 (j-i+1) + c_4} + \binom{i<0 \vee j<i}{0},$$
$$	\etahat(\fn{merge}, 11) = \binom{l \ge i \wedge l \le j+1}{c_5 l + c_6 j + c_7 i + c_8} + \binom{l < i \vee l > j+1}{0},$$
$$	\etahat(\fn{merge}, 19) = \binom{l \ge i \wedge l \le j+1}{c_9 l + c_{10} j + c_{11}} + \binom{l < i \vee l > j+1}{0}.$$

\emph{Expansion to Other Labels.} By applying C1-C5 the algorithm calculates $\etahat$ for all labels in the following order:

$$
\begin{matrix*}[l]
\etahat(\fn{mergesort},6) = 1 + \etahat(\fn{mergesort}, 7), \\
\etahat(\fn{mergesort}, 5) = 1 + \etahat(\fn{mergesort}, 7) + \etahat(\fn{merge}, 8),\\
\etahat(\fn{mergesort}, 4) = 1 + \etahat(\fn{mergesort}, 5) + \etahat({\fn{mergesort},1})[i \leftarrow k+1], \\
\etahat(\fn{mergesort}, 3) = 1 + \etahat(\fn{mergesort}, 4) + \etahat({\fn{mergesort},1})[j \leftarrow k],\\
\etahat(\fn{mergesort}, 2) = 1 + \etahat(\fn{mergesort}, 3)[k \leftarrow \lfloor (i+j) /2 \rfloor],\\
\\
\etahat(\fn{merge}, 21) = 1 + \etahat(\fn{merge}, 19)[l \leftarrow l+1],\\
\etahat(\fn{merge}, 20) = 1 + \etahat(\fn{merge}, 21),\\
\etahat(\fn{merge}, 18) = 1 + \etahat(\fn{merge}, 19)[l \leftarrow i],\\
\etahat(\fn{merge}, 17) = 1 + \etahat(\fn{merge}, 11)[l \leftarrow l+1],\\
\etahat(\fn{merge}, 16) = 1 + \etahat(\fn{merge}, 17)[n \leftarrow n + 1],\\
\etahat(\fn{merge}, 15) = 1 + \etahat(\fn{merge}, 16),\\
\etahat(\fn{merge}, 14) = 1 + \etahat(\fn{merge}, 17)[m \leftarrow m+1],\\
\etahat(\fn{merge}, 13) = 1 + \etahat(\fn{merge}, 14),\\
\etahat(\fn{merge}, 12) = 1 + \max\left\{ \etahat(\fn{merge}, 13), \etahat(\fn{merge}, 15) \right\},\\
\etahat(\fn{merge}, 10) = 1 + \etahat(\fn{merge}, 11)[l \leftarrow i],\\
\etahat(\fn{merge}, 9) = 1 + \etahat(\fn{merge},10)[n \leftarrow k+1].

\end{matrix*}
$$

The concrete values of $\etahat$'s calculated as above are too long to present in the paper. Therefore we only include some of them to illustrate the method:
$$
\etahat(\fn{mergesort},6) = 1,
$$
$$
\etahat(\fn{merge}, 21) = \binom{l+1 \ge i \wedge l \le j}{c_9 l + c_9 + c_{10} j + c_{11} + 1} + \binom{l+1 < i \vee l > j}{1},
$$
$$
\etahat(\fn{merge}, 20) = \binom{l+1 \ge i \wedge l \le j}{c_9 l + c_9 + c_{10} j + c_{11} + 2} + \binom{l+1 < i \vee l > j}{2},$$
$$
\etahat(\fn{merge}, 18) = \binom{i \ge i \wedge i \le j+1}{c_9 i + c_{10} j + c_{11} + 1} + \binom{i<i \vee i>j+1}{1}.
$$

The algorithm utilizes Farkas lemma to simplify the $\etahat$'s and remove unnecessary terms. Hence, at the end of this step we will have:

$$
\etahat(\fn{merge}, 18) = \binom{i \leq j+1}{c_9 i + c_{10} j + c_{11}+1} + \binom{i>j+1}{1}.
$$

\subsection{Step 3: Establishment of Constraint Triples}
A constraint triple $(\fn{f}, \phi,\mathfrak{e})$ denotes
$\forall \nu\in\Aval{f}.\left(\nu\models\phi\rightarrow \llbracket\mathfrak{e}\rrbracket(\nu)\ge 0\right)$\enskip, i.e., each triple consists of a function name, a precondition $\phi$ and an expression $\ex$. This means that the unknown variables ($c_i$
's) should be assigned values in a way that whenever $\phi$ is satisfied, we have $\ex \ge 0$. When several triples are obtained, we group them together in a conjunctive manner. Note that $(\fn{f}, \phi \vee \psi,\mathfrak{e})$ can be broken into two triples $(\fn{f}, \phi,\mathfrak{e})$ and $(\fn{f}, \psi,\mathfrak{e})$.

The function $\etahat$, by definition, satisfies the conditions of a measure function in all non-significant labels. In order to obtain a correct measure function, we need to make sure that these conditions are fulfilled in significant labels, too. Concretely, the algorithm has to find $c_i$'s such that the following inequalities hold and therefore converts each of these inequalities to a series of constraint triples. 
$$
\begin{matrix*}[l]
\etahat(\fn{mergesort}, 1) \ge 0,\\
\etahat(\fn{merge}, 8) \ge 0,\\
\etahat(\fn{merge}, 11) \ge 0,\\
\etahat(\fn{merge}, 19) \ge 0,\\
\\
\etahat(\fn{mergesort}, 1) \ge 1 + \mathbf{1}_{1 \leq i \leq j - 1} \cdot \etahat(\fn{mergesort}, 2) + \mathbf{1}_{i<1 \vee i>j-1} \etahat(\fn{mergesort}, 6),\\
\etahat(\fn{merge}, 8) \ge 1 + \etahat(\fn{merge}, 9)[m \leftarrow i],\\
\etahat(\fn{merge}, 11) \ge 1 + \mathbf{1}_{l \le j} \etahat(\fn{merge}, 12) + \mathbf{1}_{l > j} \etahat(\fn{merge}, 18),\\
\etahat(\fn{merge}, 19) \ge 1 + \mathbf{1}_{l \le j} \etahat(\fn{merge}, 20) + \mathbf{1}_{l > j} \etahat(\fn{merge}, 22).

\end{matrix*}
$$

The first group of inequalities above dictate the non-negativity of the obtained measure function and the second group assure that it satisfies conditions C2'--C4'.

In this case, our algorithm creates a lot of triples and then uses Farkas lemma to simplify them and ignore triples that have contradictory preconditions or expressions that are always non-negative. Here we illustrate how some of these triples are obtained to cover the main ideas. 

For example, since we must have $\etahat(\fn{merge}, 8) \geq 0$, therefore:
$$
\binom{i \ge 0 \wedge j \ge i}{c_3 (j-i+1) + c_4} + \binom{i<0 \vee j<i}{0} \geq 0.
$$
This leads to the creation of three triples. The first part leads to $$T_1 := (\fn{merge}, i \ge 0 \wedge j \ge i, c_3 (j-i+1)+c_4),$$ and the second part leads to the following two triples that both get ignored because their expression is nonnegative even without considering the precondition:
$
(\fn{merge}, i<0, 0)
,
(\fn{merge}, j<i, 0)
.$

As a more involved example, the following is a term in $\etahat(\fn{merge}, 2)$:
$$
\tau := \binom{\phi_\tau}{\ex_\tau} := \binom{j - d - 1 \ge 0 \wedge d - i \ge 0 \wedge i \ge 0}{c_1 u_3 d - c_1 u_3 i + c_1 u_3 + 2 c_2 + c_1 u_4 j - c_1 u_4 d + c_3 j - c_3 i + c_3 + c_4 + 4},
$$
here $d := \lfloor (i+j)/2 \rfloor$, $u_3 := \ln(d - i + 1)$ and $u_4 := \ln(j - d).$ Note that at this stage, the term above is stored and used in its full expanded form by the algorithm and these definitions are only used in this illustration to shorten the length of this term. In step 4, the algorithm will create these variables and use the shortened representation from there on.

Since the inequality $\etahat(\fn{mergesort}, 1) \ge 1 + \mathbf{1}_{1 \leq i \leq j - 1} \cdot \etahat(\fn{mergesort}, 2)$ should be satisfied, we get:
$$
\binom{0 \le i \le j}{c_1 (j-i+1) \ln(j-i+1) + c_2} + \binom{i<0 \vee j<i}{0} \geq \mathbf{1}_{1 \le i \le j-1} (\tau + 1).
$$
This leads to the following constraint triples:
$$
(\fn{mergesort}, 0 \le i \le j \wedge 1 \leq i \leq j - 1 \wedge \phi_\tau, c_1 (j-i+1) \ln(j-i+1) + c_2 - \ex_\tau - 1),
$$
$$
(\fn{mergesort}, i < 0 \wedge 1 \leq i \leq j-1 \wedge \phi_\tau , -\ex_\tau - 1),
$$
$$
(\fn{mergesort}, j<i \wedge 1 \leq i \leq j-1 \wedge \phi_\tau, -\ex_\tau - 1).
$$

The last two triples are discarded because the algorithm uses Farkas Lemma and deduces that their conditions are unsatisfiable. The first triple is also simplified using Farkas Lemma. This leads to the following final triple:
$$
T_2 := (\fn{mergesort}, 1 \leq i \leq j - 1 \wedge \phi_\tau,  c_1 (j-i+1) \ln(j-i+1) + c_2 - \ex_\tau - 1).
$$

All other inequalities are processed in a similar manner. In this case, at the end of this stage, the algorithm has generated 19 simplified constraint triples.

\subsection{Step 4: Converting Constraint Triples to Linear Inequalities}
This is the main step of the algorithm. Now that the triples are obtained, we need to solve the system of triples to get concrete values for $c_i$'s and hence upper-bounds on the runtime of functions, but this is a non-linear optimization problem. At this step, we introduce new variables and use them to obtain linear inequalities.

\emph{Step 4(a): Abstraction of Logarithmic, Exponentiation and Floored Expressions.} At this stage the algorithm defines the following variables and replaces them in all constraint triples to obtain a short representation like the one we used for $\tau$:
$$
\begin{matrix*}[l]
d := \lfloor \frac{i+j}{2} \rfloor,\\
u_0 := \ln(j-i+1),\\
u_1 := \ln(j-k),\\
u_2 := \ln(k-i+1),\\
u_3 := \ln(d - i + 1),\\
u_4 := \ln(j-d).
\end{matrix*}
$$
In the description we also use $w$ to denote $i+j$, but the algorithm does not add this variable. Note that all logarithmic, floored and exponentiation terms are replaced by the new variables above and hence all triples become linear, but in order for them to reflect the original triples, new conditions should be added to them. To each triple $T$, we assign a set $\Gamma$ of linear polynomials such that their non-negativity captures the desired conditions.
For example, for the triple $T_1$ above, we have $\Gamma_1 = \left\{i, j-i \right\}$, because the only needed conditions are $i \ge 0$ and $j \ge i$ and since no new variable has appeared in this triple, there is no need for any additional constraint.

On the other hand, for $T_2$, we start by setting $\Gamma_2 = \left\{i-1, j-d-1, d-i \right\}$. Note that we could add $j-i-1$ and $i$ to $\Gamma_2$, too, but these can be deduced from the rest and hence the algorithm simplifies $\Gamma_2$ and discards them. 

Since the variables $u_0, u_3$ and $u_4$ appear in $T_2$ and are non-negative logarithmic terms, the algorithm adds $u_0, u_3, u_4$ to $\Gamma_2$.

Then the algorithm adds $w-2d = i + j - 2d$ and $2d - w + 1 = 2d - i - j + 1$ to $\Gamma_2$. This is due to the definition of $d$ as $\lfloor w/2 \rfloor$. Then it performs an emptiness checking to discard the triple if it has already become infeasible. 

\emph{Step 4(b): Linear Constraints for Abstracted (Logarithmic) Variables.} We find constraints for every logarithmic variable. As an example, since $j \ge d+1$ and $d\ge i$, it can be inferred that $j-i+1\geq 2$. This is the tightest obtainable bound and the algorithm finds it using Farkas Lemma. Now since $2 < e$ and $u_0 = \ln(j-i+1)$, the algorithm adds $(j-i+1) - e u_0$ to $\Gamma_2$ according to part (4) of Step 4(b). Similar constraints are added for other variables.

Next the algorithm adds constraints on the relation between $u_i$'s to $\Gamma_2$ as in parts (6) and (7). For example, since $\Gamma_2 \ge 0$ implies $(j-i+1) - 2 (j-d) \ge 0$ and $2 (j-d) \ge 1$, the algorithm infers that $u_0 - \ln 2 - u_4$ is non-negative and adds it to $\Gamma_2$ according to part (6).

Other constraints, and constraints for exponentiation variables, if present, will also be added according to step 4(b).

\subsection{Step 5: Solving Unknown Coefficients in the Template ($c_i$'s)}
After creating $\Gamma$'s, the algorithm attempts to find suitable values for the variables $c_i$ such that for each triple $T = (\fn{f}, \phi,\mathfrak{e})$ and its corresponding $\Gamma$, it is the case that $\Gamma \ge 0$ implies $\ex \ge 0$. Since all elements of $\Gamma$ are linear, we can use Handelman's theorem to reduce this problem to an equivalent system of linear inequalities. The algorithm does this for every triple and then appends all the resulting systems of linear inequalities together in a conjunctive manner and uses an LP-Solver to solve it. 
In this case the final result, obtained from \texttt{lpsolve} is as follows:
\begin{table}[]
	\centering
	\begin{tabular}{|c|c|c|c|}
		\hline
		\textbf{variable} & \textbf{value} & \textbf{variable}         & \textbf{value}     \\ \hline
		$c_1$             & 40.9650        & $c_7$                     & -3                 \\ \hline
		$c_2$             & 3              & $c_8$                     & 12                 \\ \hline
		$c_3$             & 9              & $c_9$                     & -3                 \\ \hline
		$c_4$             & 6              & $c_{10}$                  & 3                  \\ \hline
		$c_5$             & -6             & \multirow{2}{*}{$c_{11}$} & \multirow{2}{*}{4} \\ \cline{1-2}
		$c_6$             & 9              &                           &                    \\ \hline
	\end{tabular}
\end{table}

This means that the algorithm successfully obtained the upper-bound $$40.9650 (j-i+1) \ln(j-i+1) + 3$$ for the given Merge-Sort implementation.

For simplicity of illustration we do not show how to obtain a better leading constant. In the illustration above 
we show for constant $40.9650$, whereas our approach and implementation can obtain the constant $25.02$ 
(as reported in Table~\ref{tbl:experimentalresults}).

\section{Experimental Details}\label{app:experiments}

In this part, we present the details for our experimental results.

\smallskip\noindent{\em Pseudo-codes for Our Examples.} Figure~\ref{fig:karatsubaa} and Figure~\ref{fig:karatsubab} together account for Karatsuba's Algorithm.
Figure~\ref{fig:closestpaira} -- Figure~\ref{fig:closestpaird} demonstrate the divide-and-conquer algorithm for Closest-Pair problem.
Figure~\ref{fig:strassena} -- Figure~\ref{fig:strassenc} show the Strassen's algorithm.
Invariants are bracketed (i.e., $[\dots]$) at significant labels in the programs.

\begin{remark}[Approximation constants.]
We use approximation of constants upto four digits of precision.
For example, we use the interval $[2.7182,2.7183]$ (resp. $[0.6931, 0.6932]$) for tight approximation of $e$ (resp. $\ln{2}$).
We use similar approximation for constants such as $2^{0.6}$ and $2^{0.9}$.
\qed
\end{remark}

\begin{remark}[Input specifications.]
We note that as input specifications other than the input program, the invariants
can be obtained automatically~\cite{DBLP:conf/cav/ColonSS03,DBLP:conf/popl/CousotC77}.
In the examples, the invariants are even simpler, and obtained from the guards of
branching labels.
Besides the above we have quadruple $(d, \mathrm{op}, r,k)$. We discuss these parameter for the examples below.

\begin{itemize}
\item The type of bound to be synthesized is denoted by $\mathrm{op}$: For Merge-Sort and Closest-Pair it is $\log$ (denoting logarithmic terms in expression) and
for Strassen's and Karatsuba's algorithms, it is $\mathrm{exp}$ (denoting non-polynomial bounds with non-integral exponent).

\item The number of terms multiplied together in each summand of the general form as in (\ref{eq:synform}) is $d$.
For example, (i) for $n^{4.5} $ there is only one term, and hence $d=1$;
(ii) for $n^{3.9} \cdot \log n$  we have two terms, and hence $d=2$.
Therefore, even for worst-case non-polynomial bounds of higher degrees, the maximum degree for template is still small.
In all our examples $d$ is at most~2.

\item Recall that $r$ is an upper bound on the degree of exponent of the asymptotic bound of the measure
function.
Also observe that if $r$ is not specified, we can search for $r$ in a desired interval automatically using binary search.
For example, for Strassen's algorithm, the desired interval of the exponent is between $[2,3]$.
With $r=2$, our approach on Strassen's algorithm reports failure.
Thus a binary search for $r$ in the interval can obtain the exponent as $2.9$.

\item In all our examples $k=2$ for the parameter for the Handelmann's Theorem.

\end{itemize}
Thus the input for our algorithm is quite simple.
\qed
\end{remark}

\begin{remark}[Complexity.]
Given $d$ and $k$ for the template are constants the complexity of our algorithm is polynomial.
While our algorithm is exponential in these parameters, in all our examples the above parameters
are at most~2.
The approach we present is polynomial time (using linear programming) for several non-trivial
examples, and therefore a scalable one.
\qed
\end{remark}

\lstset{language=prog}
\lstset{tabsize=3}
\newsavebox{\progkaratsubaa}
\begin{lrbox}{\progkaratsubaa}
\begin{lstlisting}[mathescape]

//$\mbox{Initialize all array entries to be zero.}$
$\mathsf{initialize}(i,j)$ {
  $[i\le j]$
  $l := i$;
  $[l\le j+1]$
  while $l \le j$ do
    skip; $l := l+1$
  od
}

//$\mbox{Copy one array into another}.$
$\mathsf{copy}(i,j,m,n)$ {
  $[i\le j\wedge m\le n]$
  $k:=i$; $l:=m$;
  $[k\le j+1\wedge l\le n+1]$
  while $k \le j\wedge l\le n$ do
    skip; $k:=k+1$; $l:=l+1$
  od
}

//$\mbox{Add two arrays entrywise.}$
$\mathsf{add}(i,j,m,n)$ {
  $[i\le j\wedge m\le n]$
  $k:=i$; $l:=m$;
  $[k\le j+1\wedge l\le n+1]$
  while $k \le j\wedge l\le n$ do
    skip; $k:=k+1$; $l:=l+1$
  od
}

//$\mbox{Subtract two arrays entrywise.}$
$\mathsf{subtract}(i,j,m,n)$ {
  $[i\le j\wedge m\le n]$
  $k:=i$; $l:=m$;
  $[k\le j+1\wedge l\le n+1]$
  while $k \le j\wedge l\le n$ do
    skip; $k:=k+1$; $l:=l+1$
  od
}

\end{lstlisting}
\end{lrbox}
\begin{figure}
\centering
\usebox{\progkaratsubaa}
\caption{Auxiliary Function Calls for Karatsuba's Algorithm}
\label{fig:karatsubaa}
\end{figure}

\lstset{language=prog}
\lstset{tabsize=3}
\newsavebox{\progkaratsubab}
\begin{lrbox}{\progkaratsubab}
\begin{lstlisting}[mathescape]
//$\mbox{The program calculates the product of two polynomials.}$
//$\mbox{The degree should be arranged in increasing order.}$
//$\mbox{Array index starts from }1.$
//$n\mbox{ is the length of the arrays and should be a power of } 2.$
//$\mbox{The quadruple of input parameters is }(1, \mathrm{exp}, 1.6, 2)$.

$\mathsf{karatsuba}(n)$ {
  $\left[n\ge 1\right]$
  if $n \ge 2$ then
     $t:= \left\lfloor \frac{n}{2}\right\rfloor$;

     //$\mbox{ checking whether }n\mbox{ is even}$
     if $2 * t \le n$ and $2 * t \ge n$ then

       //$\mbox{sub-dividing arrays}$
       $\mathsf{copy}(1, t, 1, t)$;
       $\mathsf{copy}(t+1, n, 1, t)$;
       $\mathsf{copy}(1, t, 1, t)$;
       $\mathsf{copy}(t+1, n, 1, t)$;

       //$\mbox{adding the sub-arrays}$
       $\mathsf{copy}(1, t, 1, t)$; $\mathsf{add}(1, t, 1, t)$;
       $\mathsf{copy}(1, t, 1, t)$; $\mathsf{add}(1, t, 1, t)$;

       //$\mbox{recursive calls}$
       $\mathsf{karatsuba}(t)$;
       $\mathsf{karatsuba}(t)$;
       $\mathsf{karatsuba}(t)$;

       //$\mbox{combining step}$
       $\mathsf{subtract}(1, n-1, 1, n-1)$;
       $\mathsf{subtract}(1, n-1, 1, n-1)$;

       $\mathsf{initialize}(1, 2*n-1)$;
       $\mathsf{add}(1, n-1, 1, n-1)$;
       $\mathsf{add}(1, n-1, n, 2*n-2)$;
       $\mathsf{add}(t +1, n+t-1, 1, n-1)$
     else skip //$\mbox{If }n\mbox{ is not even, simply fail.}$
     fi
  else //$\mbox{trivial case}$
     skip
  fi
}

\end{lstlisting}
\end{lrbox}
\begin{figure}
\centering
\usebox{\progkaratsubab}
\caption{Main Function Call for Karatsuba's Algorithm}
\label{fig:karatsubab}
\end{figure}

\lstset{language=prog}
\lstset{tabsize=3}
\newsavebox{\progclosestpaira}
\begin{lrbox}{\progclosestpaira}
\begin{lstlisting}[mathescape]

//$\mbox{Copy one array into another}.$
$\mathsf{copy}(i,j,m,n)$ {
  $[i\le j\wedge m\le n]$
  $k:=i$; $l:=m$;
  $[k\le j+1\wedge l\le n+1]$
  while $k \le j\wedge l\le n$ do
    skip; $k:=k+1$; $l:=l+1$
  od
}

//$\mbox{sorting one array while adjusting another accordingly}$

$\mathsf{mergesort}(i, j)$ {
  [$i\le j$]
  if $i\le j-1$ then
    $k:=i+\lfloor \frac{j-i+1}{2}\rfloor-1$;
    $\mathsf{mergesort}(i, k)$;
    $\mathsf{mergesort}(k+1, j)$;
    $\mathsf{merge}(i,j,k)$
  else
    skip
  fi
}

$\mathsf{merge}(i,j,k)$ {
  [$i\le j$]
  $m:=i$; $n:=k+1$; $l:=i$;
  [$l\le j+1$]
  while $l\le j$ do
    if $\star$ then
       skip; $m:=m+1$
    else
       skip; $n:=n+1$
    fi;
    $l:=l+1$
  od;
  $l:=i$;
  [$l\le j+1$]
  while $l\le j$ do
    skip; $l:=l+1$
  od
}

\end{lstlisting}
\end{lrbox}
\begin{figure}
\centering
\usebox{\progclosestpaira}
\caption{Merge-Sort and Copy for Closest-Pair}
\label{fig:closestpaira}
\end{figure}

\lstset{language=prog}
\lstset{tabsize=3}
\newsavebox{\progclosestpairb}
\begin{lrbox}{\progclosestpairb}
\begin{lstlisting}[mathescape]

//$\mbox{Caculates the shortest distance of a finite set of points.}$
//$\mbox{One array stores }x\mbox{-coordinates and another stores }y\mbox{-coordinates.}$
//$\mbox{The quadruple of input parameters is }(2, \log, -, 2)$.

$\mathsf{clst}\_\mathsf{pair}\_\mathsf{main}(i, j)$ {
  [$i\le j$]

  //$\mbox{copying arrays}$
  $\mathsf{copy}(i, j, i, j)$; $\mathsf{copy}(i, j, i, j)$;

  //$\mbox{sorting arrays in }x\mbox{-coordinate}$
  $\mathsf{mergesort}(i,j)$;

  //$\mbox{sorting arrays in }y\mbox{-coordinate}$
  $\mathsf{mergesort}(i,j)$;

  //$\mbox{solving the result}$
  $\mathsf{clst}\_\mathsf{pair}(i, j)$
}

\end{lstlisting}
\end{lrbox}
\begin{figure}
\centering
\usebox{\progclosestpairb}
\caption{Main Function Call for Closest-Pair}
\label{fig:closestpairb}
\end{figure}

\lstset{language=prog}
\lstset{tabsize=3}
\newsavebox{\progclosestpairc}
\begin{lrbox}{\progclosestpairc}
\begin{lstlisting}[mathescape]
//$\mbox{principle recursive function call for solving Closest-Pair}$

$\mathsf{clst}\_\mathsf{pair}(i, j)$ {
  [$i\le j$]
  if $i\le j-3$ then
    //$\mbox{recursive case where there are at least }4\mbox{ points }$

    $k:=i+\lfloor \frac{j-i+1}{2}\rfloor-1$;
    $\mathsf{clst}\_\mathsf{pair}(i, k)$;
    $\mathsf{clst}\_\mathsf{pair}(k+1, j)$;

    //$\mbox{taking the minimum distance from the previous recursive calls}$
    skip;

    //$\mbox{fetch and scan the mid-line}$
    $\mathsf{fetch}\&\mathsf{scan}(i,j)$
  else
    //$\mbox{base case (fewer than }4\mbox{ points)}$
    skip
  fi
}

\end{lstlisting}
\end{lrbox}
\begin{figure}
\centering
\usebox{\progclosestpairc}
\caption{Principle Recursive Function Call for Closest-Pair}
\label{fig:closestpairc}
\end{figure}

\lstset{language=prog}
\lstset{tabsize=3}
\newsavebox{\progclosestpaird}
\begin{lrbox}{\progclosestpaird}
\begin{lstlisting}[mathescape]
//$\mbox{fetch and scan the mid-line}$
$\mathsf{fetch}\&\mathsf{scan}(i,j)$ {
  //$\mbox{fetching the points on the mid-line}$
  [$i\le j-3$]
  $l:=i$; $p:=i$;

  [$i\le j-3\wedge p\le j+1\wedge l\le j+1$]
  while $p\le j$ do
    if $\star$ then $l:=l+1$ else skip fi;
    $p:=p+1$
  od

  if $l\ge i+1$ and $l\le j+1$ then
    $p:=i$;

    //$\mbox{scanning the points on the mid-line}$
    [$p\le l$]
    while $p\le l-1$ do
      $m:=p+1$;

      //$\mbox{checking }7\mbox{ points ahead on the mid-line}$
      [$m\le p+8$]
      while $m-p\le 7$ and $m\le l-1$ do
        skip; $m:=m+1$
      od;
      $p:=p+1$
    od
  else skip fi
}

\end{lstlisting}
\end{lrbox}
\begin{figure}
\centering
\usebox{\progclosestpaird}
\caption{Other Function Calls for Closest-Pair}
\label{fig:closestpaird}
\end{figure}

\lstset{language=prog}
\lstset{tabsize=3}
\newsavebox{\progstrassena}
\begin{lrbox}{\progstrassena}

\begin{lstlisting}[mathescape]
//$\mbox{The program calculates the product of two matrices.}$
//$n\mbox{ is the row/column size of both matrices and should be a power of } 2.$
//$\mbox{Each of the matrices is stored in a two-dimensional array of dimension } n.$
//$\mbox{Array indices starts from } 1.$
//$\mbox{The quadruple of input parameters is }(2, \mathrm{exp}, 1.9, 2)$.

$\mathsf{strassen}(n)$ {
  $\left[n\ge 1\right]$
  if $n \ge 2$ then
     $t:= \left\lfloor \frac{n}{2}\right\rfloor$;

     //$\mbox{ checking whether }n\mbox{ is even}$
     if $2 * t \le n$ and $2 * t \ge n$ then
       //$\mbox{sub-dividing matrices}A$
       $\mathsf{matrixtoblocks}(n, t)$; $\mathsf{matrixtoblocks}(n, t)$;

       //$\mbox{sums of matrices}$
       $\mathsf{copy}(t)$; $\mathsf{add}(t)$; $\mathsf{copy}(t)$; $\mathsf{add}(t)$;
       $\mathsf{copy}(t)$; $\mathsf{add}(t)$; $\mathsf{copy}(t)$; $\mathsf{subtract}(t)$;
       $\mathsf{copy}(t)$; $\mathsf{add}(t)$; $\mathsf{copy}(t)$; $\mathsf{add}(t)$;
       $\mathsf{copy}(t)$; $\mathsf{subtract}(t)$; $\mathsf{copy}(t)$; $\mathsf{add}(t)$;
       $\mathsf{copy}(t)$; $\mathsf{subtract}(t)$; $\mathsf{copy}(t)$; $\mathsf{add}(t)$;

       //$\mbox{recursive calls}$
       $\mathsf{strassen}(t)$;
       $\mathsf{strassen}(t)$;
       $\mathsf{strassen}(t)$;
       $\mathsf{strassen}(t)$;
       $\mathsf{strassen}(t)$;
       $\mathsf{strassen}(t)$;
       $\mathsf{strassen}(t)$;

       //$\mbox{combining stage}$
       $\mathsf{copy}(t)$; $\mathsf{add}(t)$; $\mathsf{subtract}(t)$; $\mathsf{add}(t)$;
       $\mathsf{copy}(t)$; $\mathsf{add}(t)$; $\mathsf{copy}(t)$; $\mathsf{add}(t)$;

       $\mathsf{copy}(t)$; $\mathsf{add}(t)$; $\mathsf{subtract}(t)$; $\mathsf{add}(t)$;
       $\mathsf{blockstomatrix}(n,t)$
     else skip //$\mbox{If }n\mbox{ is not even, simply fail.}$
     fi
  else //$\mbox{trivial case}$
     skip
  fi
}

\end{lstlisting}
\end{lrbox}
\begin{figure}
\centering
\usebox{\progstrassena}
\caption{Main Function Call for Strassen's Algorithm}
\label{fig:strassena}
\end{figure}

\lstset{language=prog}
\lstset{tabsize=3}
\newsavebox{\progstrassenb}
\begin{lrbox}{\progstrassenb}
\begin{lstlisting}[mathescape]

//$\mbox{Partition a matrix into block matrices.}$
$\mathsf{matrixtoblocks}(n,t)$ {
  [$t\ge 1$]
  $i := 1$;
  [$t\ge 1\wedge i\le t+1$]
  while $i\le t$ do
    $j:=1$;

    [$t\ge 1\wedge i\le t\wedge j\le t+1$]
    while $j\le t$ do
      skip; $j:=j+1$
    od;
    $i:=i+1$
  od
}

//$\mbox{Construct a matrix from block matrices.}$
$\mathsf{blockstomatrix}(n,t)$ {
  [$t\ge 1$]
  $i := 1$;
  [$t\ge 1\wedge i\le t+1$]
  while $i\le t$ do
    $j:=1$;

    [$t\ge 1\wedge i\le t\wedge j\le t+1$]
    while $j\le t$ do
      skip; $j:=j+1$
    od;
    $i:=i+1$
  od
}

//$\mbox{Copy a square matrix into another}.$
$\mathsf{copy}(n)$ {
  $[n\ge 1]$
  $i:=1$;
  $[n\ge 1\wedge i\le n+1]$
  while $i \le n$ do
    $j:=1$;
    $[n\ge 1\wedge i\le n\wedge j\le n+1]$
    while $j\le n$ do
      skip; $j:=j+1$
    od;
    $i:=i+1$
  od
}

\end{lstlisting}
\end{lrbox}
\begin{figure}
\centering
\usebox{\progstrassenb}
\caption{Auxiliary Function Calls for Strassen's Algorithm}
\label{fig:strassenb}
\end{figure}

\lstset{language=prog}
\lstset{tabsize=3}
\newsavebox{\progstrassenc}
\begin{lrbox}{\progstrassenc}
\begin{lstlisting}[mathescape]

//$\mbox{Add two matrices (entrywise).}$
$\mathsf{add}(n)$ {
  $[n\ge 1]$
  $i:=1$;
  $[n\ge 1\wedge i\le n+1]$
  while $i \le n$ do
    $j:=1$;
    $[n\ge 1\wedge i\le n\wedge j\le n+1]$
    while $j\le n$ do
      skip; $j:=j+1$
    od;
    $i:=i+1$
  od
}

//$\mbox{Subtract two matrices (entrywise).}$
$\mathsf{subtract}(n)$ {
  $[n\ge 1]$
  $i:=1$;
  $[n\ge 1\wedge i\le n+1]$
  while $i \le n$ do
    $j:=1$;
    $[n\ge 1\wedge i\le n\wedge j\le n+1]$
    while $j\le n$ do
      skip;
      $j:=j+1$
    od;
    $i:=i+1$
  od
}

\end{lstlisting}
\end{lrbox}
\begin{figure}
\centering
\usebox{\progstrassenc}
\caption{Matrix Addition and Subtraction for Strassen's Algorithm}
\label{fig:strassenc}
\end{figure}

\end{document}